\documentclass{journalA4}

\usepackage{microtype}
\usepackage{booktabs}
\usepackage{xspace}
\usepackage{wrapfig}
\usepackage{pifont}
\graphicspath{{figures/}}

\title{Mapping polygons to the grid\linebreak with small Hausdorff and Fr\'echet distance\footnote{%To appear in the proceedings of the 24th European Symposium on Algorithms (ESA 2016). 
Research on the topic of this paper was initiated at the 1st Workshop on Applied Geometric Algorithms (AGA 2015) in Langbroek,
The Netherlands, supported by the Netherlands Organisation for Scientific Research (NWO)
under project no. 639.023.208.
NWO is supporting Q.\,W. Bouts, I. Kostitsyna, and W. Sonke under project no.~639.023.208, and K. Verbeek under project no.~639.021.541.
W. Meulemans is supported by Marie Sk\l{}odowska-Curie Action MSCA-H2020-IF-2014 656741.}}

\author{Quirijn W. Bouts\thanks{Dept.\ of Mathematics and Computer Science, TU Eindhoven, The Netherlands. {\tt \{ q.w.bouts | i.kostitsyna | w.m.sonke | k.a.b.verbeek \}@tue.nl}}
\and Irina Kostitsyna\footnotemark[2]
\and Marc van Kreveld\thanks{Dept.\ of Information and Computing Sciences, Utrecht University, The Netherlands. {\tt m.j.vankreveld@uu.nl}}
\and Wouter Meulemans\thanks{giCentre, City University London, United Kingdom. {\tt wouter.meulemans@city.ac.uk}}
\and Willem Sonke\footnotemark[2] 
\and Kevin Verbeek\footnotemark[2]}

\newtheorem{definition}{Definition} 
\renewcommand{\subparagraph}[1]{\smallskip\noindent{\bfseries\sffamily #1.}}
\newenvironment{cloneclaim}[1]
  {\noindent\bfseries #1 [repeated] \normalfont\em}
  {\normalfont}
\usepackage[figuresright]{rotating}

\newcommand{\cmark}{\ding{51}}
\newcommand{\xmark}{\ding{55}}
\newcommand{\module}[1]{\ensuremath{\mathcal{M}(#1)}\xspace}

\begin{document}

\maketitle

\begin{abstract}
We show how to represent a simple polygon $P$ by a grid (pixel-based) polygon $Q$ that is simple and whose Hausdorff or Fr\'echet distance to $P$ is small. For any simple polygon $P$, a grid polygon exists with constant Hausdorff distance between their boundaries and their interiors. Moreover, we show that with a realistic input assumption we can also realize constant Fr\'echet distance between the boundaries. We present algorithms accompanying these constructions, heuristics to improve their output while keeping the distance bounds, and experiments to assess the output.
\end{abstract}

\section{Introduction}
\label{sec:intro}

Transforming the representation of objects from the real plane onto a grid has been studied for decades
due to its applications in computer graphics, computer vision, and finite-precision
computational geometry~\cite{yao}. Two interpretations of the grid are possible: (i) the grid graph,
consisting of vertices at all points with integer coordinates, and horizontal and vertical edges
between vertices at unit distance; (ii) the pixel grid, where the only elements are pixels (unit squares). In the latter, one can choose between 4-neighbor or 8-neighbor grid
topology.
In this paper we adopt the pixel grid view with 4-neighbor topology.

The issues involved when moving from the real plane to a grid begin with the
definition of a line segment on a grid, known as a
\emph{digital straight segment}~\cite{klette}.
For example, it is already difficult to represent line segments such
that the intersection between any pair is a connected set (or empty). In general,
the challenge is to represent objects on a grid in such a way that certain properties
of those objects in the real plane transfer to related properties on the grid;
connectedness of the intersection of two line segments is an example of this.

While most of the research related to \emph{digital geometry} has the
graphics or vision perspective~\cite{klette,klette2}, computational geometry has made a number of
contributions as well. Besides finite-precision computational geometry~\cite{Devillers2006,yao}
these include snap rounding~\cite{berg,guibas,hershberger}, the integer hull~\cite{Althaus04,Harvey99},
and consistent digital rays with small Hausdorff distance~\cite{chun}.

\subparagraph{Mapping polygons} 
We consider the problem of representing a \emph{simple polygon} $P$
as a similar polygon in the grid (see Fig.~\ref{fig-killer}).
A \emph{grid cycle} is a simple cycle of edges and vertices of the grid graph.
A \emph{grid polygon} is a set of pixels whose boundary is a grid cycle.
This problem is motivated by
schematization of country or building outlines and by nonograms.

\begin{figure}[t]
	\centering
	\includegraphics[width=0.9\linewidth]{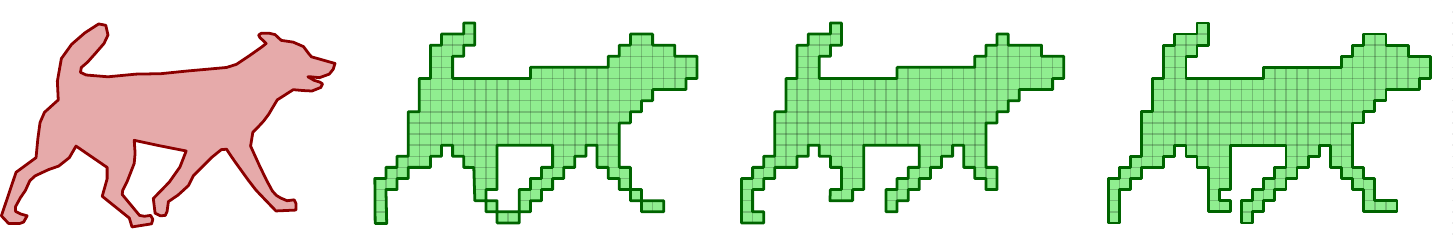}
	\caption{From left to right: input; symmetric-difference optimal result is not a grid polygon; grid polygon computed by our Fr\'echet algorithm; grid polygon computed by our Hausdorff algorithm.}
	\label{fig-killer}
\end{figure}

The most well-known form of schematization in cartography are metro maps, where metro
lines are shown in an abstract manner by polygonal lines whose edges typically have
only four orientations. It is common to also depict region outlines with these
orientations on such maps. It is possible to go one step further in schematization by
using only integer coordinates for the vertices, which often aligns vertices
vertically or horizontally, and leads to a more abstracted view. Certain types of cartograms
like mosaic maps~\cite{cano} are examples of maps following this visualization style.
The version based on a square grid is often used to show electoral votes after elections.
Another cartographic application of grid polygons lies in the schematization of building outlines~\cite{meulemans}.

Nonograms---also known as Japanese or picture logic puzzles---are popular
in puzzle books, newspapers, and in digital form. The objective is to reconstruct a
pixel drawing from a code that is associated with every row and column. 
The algorithmic problem of solving these puzzles is well-studied and
known to be NP-complete~\cite{Berend14}.
To generate a nonogram from a vector drawing, a grid polygon on a coarse grid should be made.
We are interested in the generation of grid polygons from shapes like animal outlines,
which could be used to construct nonograms. To our knowledge, two papers address
this problem. Ort\'iz-Garc\'ia et al.~\cite{Ortiz-Garcia2007} study the problem
of generating a nonogram from an image; both the black-and-white and color versions
are studied. Their approach uses image processing techniques and
heuristics. Batenburg et al.~\cite{batenburg2009} also start with an image, but
concentrate on generating nonograms from an image with varying difficulty levels,
according to some definition of difficulty.

Considering the above, our work also relates to image downscaling (e.g. \cite{kopf2013}), though this usually starts from a raster image instead of continuous geometric objects.
Kopf~et al.~\cite{kopf2013} apply their technique to vector images, stating that the outline remains connected where possible.
In contrast to our work, the quality is not measured as the geometric similarity and the conditions necessary to guarantee a connected outline remain unexplored.

\subparagraph{Similarity}
There are at least three common ways of defining the similarity of two simple polygons: the symmetric difference\footnote{The symmetric difference between two sets $A$ and $B$ is defined as the set $(A\setminus B)\cup(B\setminus A)$. When using symmetric difference as a quality measure, we actually mean the area of the symmetric difference.}, the Hausdorff distance~\cite{alt1}, and the Fr\'echet
distance~\cite{alt2}. 
The first does not consider similarity of the polygon boundaries, whereas the third usually applies
to boundaries only. 
The Hausdorff distance between polygon interiors and between polygon boundaries both exist and are different measures; this distance can be directed or undirected.
Let $X$ and $Y$ be two closed subsets of a metric space.
The (directed) \emph{Hausdorff distance} $d_H(X, Y)$ from $X$ to $Y$ is defined as the
maximum distance from any point in $X$ to its closest point in $Y$. The undirected version
is the maximum of the two directed versions.
To define the \emph{Fr\'echet distance}, let $X$ and $Y$ be two curves in the plane.
The Fr\'echet distance $d_F(X,Y)$ is the minimum leash length needed to let a man walk over $X$ and
a dog over $Y$, where neither may walk backwards (a formal definition
can be found in~\cite{alt2}).

\begin{figure}
	\centering
	\includegraphics{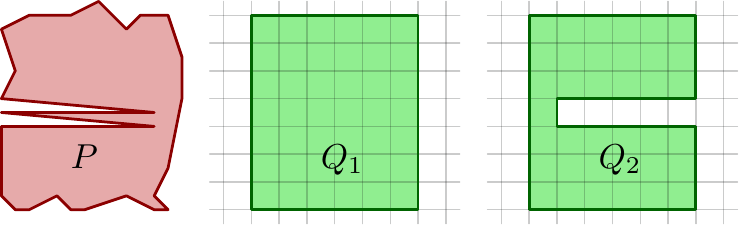}
	\caption{$d_H(P,Q_1)$ is small but $d_H(\partial P,\partial Q_1)$ is not.
	$d_H(P,Q_2)$ and $d_H(\partial P,\partial Q_2)$ are both small but the Fr\'echet
	distance $d_F(\partial P,\partial Q_2)$ is not.}
	\label{fig:intro}
\end{figure}

\subparagraph{Contributions}
In Section~\ref{sec:hausdorff} we show that any simple polygon $P$ admits a grid polygon $Q$
with $d_H(P, Q) \leq \frac{1}{2} \sqrt{2}$ and $d_H(Q, P) \leq \frac{3}{2} \sqrt{2}$ on
the unit grid.
Furthermore, the constructed polygon satisfies the same bounds between the boundaries $\partial P$ and $\partial Q$.
This is not equivalent, since the point that realizes the maximum smallest distance to the other
polygon may lie in the interior (Fig.~\ref{fig:intro}).
Our proof is constructive, but the construction often does not give intuitive results
(Fig.~\ref{fig:intro}, $P$ and $Q_2$).
Therefore, we extend our construction with heuristics that reduce the symmetric difference whilst keeping the Hausdorff
distance within $\frac{3}{2} \sqrt{2}$.
The Fr\'echet distance $d_F$~\cite{alt2} between two polygon boundaries is often considered
to be a better measure for similarity.
Unlike the Hausdorff distance, however, not every polygon boundary $\partial P$
can be represented by a grid cycle with constant Fr\'echet distance.
In Section~\ref{sec:frechet} we present a condition on the input polygon boundary
related to fatness (in fact, to $\kappa$-straightness~\cite{alt3}) and show
that it allows a grid cycle representation with constant Fr\'echet distance.
Finally, in Section~\ref{sec:experimental-results} we evaluate how our algorithms perform on realistic input polygons.

\section{Hausdorff distance}
\label{sec:hausdorff}

We consider the problem of constructing a grid polygon $Q$ with small Hausdorff distance to $P$. In Appendix~\ref{app-hausdorff} we prove the theorem below.
In this section we present an algorithm that achieves low Hausdorff distance between both the boundaries and the interiors of the input polygon $P$ and the resulting grid polygon $Q$. We first show how to construct such a grid polygon. Then, in Section~\ref{sec:hausdorff-algorithm}, we provide an efficient algorithm to compute $Q$. Finally, we describe heuristics that can be used to improve the results in practice.

\begin{theorem}
	\label{thm:hard}
	Given a polygon $P$, it is NP-hard to decide whether there exists a grid polygon $Q$ such that both $d_H(\partial P, \partial Q) \leq \frac{1}{2}$ and $d_H(\partial Q, \partial P)\leq \frac{1}{2}$.
\end{theorem}

\subsection{Construction}
\label{sec:hausdorff-construction}
We represent the grid polygon $Q$ as a set of cells (or pixels). We say that two cells are \emph{adjacent} if they share a segment. If two cells share only a point, then they are \emph{point-adjacent}. If two cells $c_1 \in Q$ and $c_2 \in Q$ are point-adjacent, and there is no cell $c \in Q$ that is adjacent to both $c_1$ and $c_2$, then $c_1$ and $c_2$ share a \emph{point-contact}. We construct $Q$ as the union of four sets $Q_1$, $Q_2$, $Q_3$, $Q_4$ (not necessarily disjoint). To define these sets, we define the \emph{module} $\module{c}$ of a cell $c$ as the $2 \times 2$-region centered at the center of $c$ (see Fig.~\ref{fig:module}). Furthermore, we assume the rows and columns are numbered, so we can speak of even-even cells, odd-odd cells, odd-even cells, and even-odd cells. The four sets are defined as follows; see also Fig.~\ref{fig:hausdorff-example}.

\begin{figure}[t]
	\begin{minipage}[b]{.27\textwidth}
		\centering
		\includegraphics{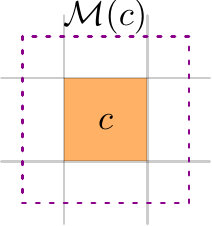}
		\caption{Module $\module{c}$ (dashed) of a cell $c$.}
		\label{fig:module}
	\end{minipage}
	\hfill
	\begin{minipage}[b]{.67\textwidth}
		\centering
		\includegraphics{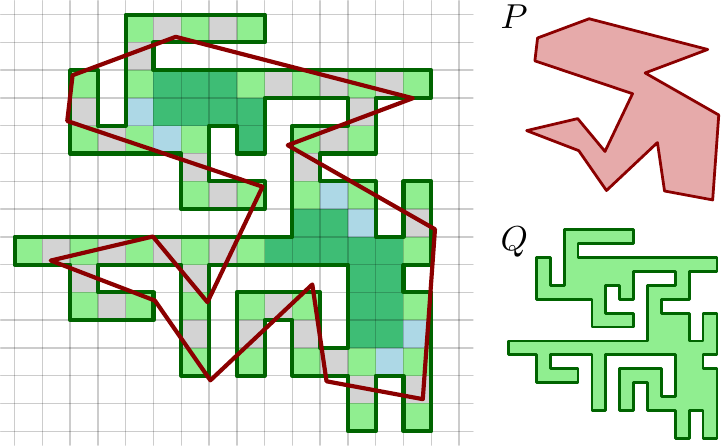}
		\caption[Example of the Hausdorff algorithm; the input and output are shown on the right.]{Example of the Hausdorff algorithm; the input and output are shown on the right. Colors: \includegraphics[page=1]{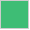}~$Q_1$, \includegraphics[page=2]{hausdorff-legend}~$Q_2$, \includegraphics[page=3]{hausdorff-legend}~$Q_3$, \includegraphics[page=4]{hausdorff-legend}~$Q_4$.}
		\label{fig:hausdorff-example}
	\end{minipage}
\end{figure}

\begin{description}
	\item[$\boldsymbol{Q_1}$:] All cells $c$ for which $\module{c} \subseteq P$.
	\item[$\boldsymbol{Q_2}$:] All even-even cells $c$ for which $\module{c} \cap P \neq \emptyset$.
	\item[$\boldsymbol{Q_3}$:] For all cells $c_1, c_2 \in Q_1 \cup Q_2$ that share a point-contact, the two cells that are adjacent to both $c_1$ and $c_2$ are in $Q_3$.
	\item[$\boldsymbol{Q_4}$:] A minimal set of cells that makes $Q$ connected, and where each cell $c \in Q_4$ is adjacent to two cells in $Q_2$ and $\module{c} \cap P \neq \emptyset$.
\end{description}

Set $Q_1 \cup Q_2$ is sufficient to achieve the desired Hausdorff distance. 
We add $Q_3$ to resolve point-contacts, and $Q_4$ to make the set $Q$ simply connected (a polygon without holes).
The lemma below is the crucial argument to argue that $Q$ is a grid polygon.

\begin{lemma}
	\label{lem:holefree}
	The set $Q_1 \cup Q_2$ is hole-free, even when including point-adjacencies.
\end{lemma}

%arxiv: wraplines set to 11 from 8
\begin{wrapfigure}[11]{r}{0.34\textwidth} 
	\centering	
	%\vspace{-1.5em} %ARXIV
	\includegraphics{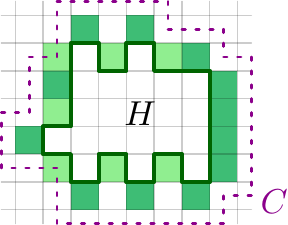}
	\vspace{-.1em}
	\caption{A hole in $Q$. Colors: \protect\includegraphics[page=1]{hausdorff-legend}~$Q_1 \cap B$; \protect\includegraphics[page=2]{hausdorff-legend}~$Q_2 \cap B$.}
	\label{fig:lemma-1}
\end{wrapfigure}
\noindent {\bfseries{Proof.}} % ARXIV
%\noindent {\sffamily \bfseries \textcolor{darkgray}{Proof.}}
For the sake of contradiction, let $H$ be a maximal set of cells comprising a hole. Let set $B$ contain all cells in $Q_1 \cup Q_2$ that surround $H$ and are adjacent to a cell in $H$. Since $Q_2$ contains only even-even cells, every cell in $Q_2 \cap B$ is (point-)adjacent to two cells in $Q_1 \cap B$ (see Fig.~\ref{fig:lemma-1}). Hence, the outer boundary of the union of all modules of cells in $Q_1 \cap B$ is a single closed curve $C$. Since $C \subset P$ due to the definition of $Q_1$, the interior of $C$ must also be in $P$. Since the module of every cell in $H$ lies completely inside $C$, they are also in $P$, so the cells in $H$ must all be in $Q_1$. This contradicts that $H$ is a hole.%\hfill$\blacktriangleleft$
\hfill$\square$ % ARXIV

\begin{lemma}
\label{lem-hs-q-is-simple}
	The set $Q$ is simply connected and does not contain point-contacts.
\end{lemma}
\begin{proof}
	Consider a point-contact between two cells $c_1, c_2 \in Q_1 \cup Q_2$ and a cell $c \notin Q_1 \cup Q_2$ that is
	adjacent to both $c_1$ and $c_2$ ($c \in Q_3$). Since $Q_2$ contains only even-even cells, we may assume that $c_1
	\in Q_1$. Recall that $\module{c_1} \subseteq P$ by definition. We may further assume that $c_1$ is an odd-odd cell, for
	otherwise a cell in $Q_2$ would eliminate the point-contact. Hence, all cells point-adjacent to $c_1$ are in $Q_1 \cup Q_2$,
	and thus $c$ has three adjacent cells in $Q_1 \cup Q_2$. This implies that adding $c \in Q_3$ to $Q_1 \cup
	Q_2$ cannot introduce point-contacts or holes. Similarly, cells in $Q_4$ connect two oppositely adjacent cells in
	$Q_2$, and thus cannot introduce point-contacts (or holes, by definition). Combining this with Lemma~\ref{lem:holefree} implies that $Q$ is hole-free and does not contain point-contacts.

	It remains to show that $Q$ is connected, that is, the set $Q_4$ exists. Consider two cells $c_1, c_2 \in Q$. We show that $c_1$ and $c_2$ are connected in $Q$. We may further assume that $c_1, c_2 \in Q_2$, as
	cells in $Q_1 \cup Q_3 \cup Q_4$ must be adjacent or point-adjacent to a cell in $Q_2$. Let $p \in \module{c_1} \cap P$,
	$q \in \module{c_2} \cap P$ and consider a path $\pi$ between $p$ and $q$ inside $P$. Every even-even cell $c$ with $\module{c} \cap \pi \neq \emptyset$ must be in $Q_2$. Furthermore, the modules of even-even cells cover the plane. Every cell connecting a consecutive pair of even-even cells intersecting $\pi$ satisfies the conditions of $Q_4$, and thus can be added to make $c_1$ and $c_2$ connected in $Q$.
\end{proof}

\subparagraph{Upper bounds}
To prove our bounds, note that $\module{c} \cap P \neq \emptyset$ holds for every
cell $c \in Q$. This is explicit for cells in $Q_1$, $Q_2$, and $Q_4$. For
cells in $Q_3$, note that these cells must be adjacent to a cell in $Q_1$, and thus contain a point in $P$.

\begin{lemma}\label{lem:upperboundPtoQ}
	$d_H(P, Q),\,d_H(\partial P, \partial Q) \leq \frac{1}{2}\sqrt{2}$.
\end{lemma}
\begin{proof}
	Let $p \in P$ and consider the even-even cell $c$ such that $p \in \module{c}$. Since $c \in Q_2$, the distance $d_H(p,
	Q) \leq d_H(p, c) \leq \frac{1}{2}\sqrt{2}$. Now consider a point $p \in \partial P$. There is a $2 \times 2$-set
	of cells whose modules contain $p$. This set contains an even-even cell $c \in Q$ and an odd-odd cell $c' \notin Q
	$. The latter is true, because odd-odd cells in $Q$ must be in $Q_1$. Therefore, the point $q$ shared by $c$ and
	$c'$ must be in $\partial Q$. Thus, $d_H(p, \partial Q) \leq d_H(p, q) \leq \frac{1}{2}\sqrt{2}$.
\end{proof}

\begin{lemma}\label{lem:upperboundQtoP}
	$d_H(Q, P),\,d_H(\partial Q, \partial P) \leq \frac{3}{2}\sqrt{2}$.
\end{lemma}
\begin{proof}
	Let $q$ be a point in $Q$ and let $c \in Q$ be the cell that contains $q$. Since $\module{c} \cap P \neq \emptyset$, we can choose a point $p \in \module{c} \cap P$. It directly follows that $d_H(q, P) \leq d_H(q, p) \leq \frac{3}{2}\sqrt{2}$. Now consider a point $q \in \partial Q$, and let $c \in Q$ and $c' \notin Q$ be two adjacent cells such that $q \in \partial c \cap \partial c'$. We claim that $(\module{c} \cup \module{c'}) \cap \partial P \neq \emptyset$. If $c \notin Q_1$, then $\module{c} \nsubseteq P$. As furthermore $\module{c} \cap P \neq \emptyset$, we have that $\module{c} \cap \partial P \neq \emptyset$. On the other hand, if $c \in Q_1$, then $\module{c} \subseteq P$, so $\module{c'} \cap P \neq \emptyset$. As furthermore $\module{c'} \nsubseteq P$ (otherwise $c' \in Q_1$), we have that $\module{c'} \cap \partial P \neq \emptyset$. Let $p \in (\module{c} \cup \module{c'}) \cap \partial P$. Then $d_H(q, \partial P) \leq d_H(q, p) \leq \frac{3}{2}\sqrt{2}$.
\end{proof}

\begin{theorem}
	For every simple polygon $P$ a simply connected grid polygon $Q$ without point-contacts exists such that
	$d_H(P, Q),\,d_H(\partial P, \partial Q) \leq \frac{1}{2}\sqrt{2}$ and $d_H(Q, P),\,d_H(\partial Q, \partial P) \leq \frac{3}{2}\sqrt{2}$.
\end{theorem}

\subparagraph{Lower bound}
Fig.~\ref{fig:hausdorff-lower-bound} illustrates a polygon $P$ for which no grid polygon $Q$ exists with low $d(Q, P)$. A naive construction results in a nonsimple polygon (left). To make it simple, we can either remove a cell (center) or add a cell (right). Both methods result in $d_H(Q, P) \geq 3/2 - \epsilon$. Alternatively, we can fill the entire upper-right part of the grid polygon (not shown), resulting in a high $d_H(Q, P)$. This leads to the following theorem.

\begin{theorem}
	For any $\epsilon > 0$, there exists a polygon $P$ for which no grid polygon $Q$ exists with $d(Q, P) < 3/2 - \epsilon$.
\end{theorem}

In the $L_\infty$ metric, the lower bound of $3/2 - \epsilon$ given in Fig.~\ref{fig:hausdorff-lower-bound} also holds. A straightforward modification of the upper-bound proofs can be used to show that the Hausdorff distance is at most $3/2$ in the $L_\infty$ metric. In other words, our bounds are tight under the $L_\infty$ metric.

% ARXIV: set figure to h
\begin{figure}[h]
	\centering
	\includegraphics{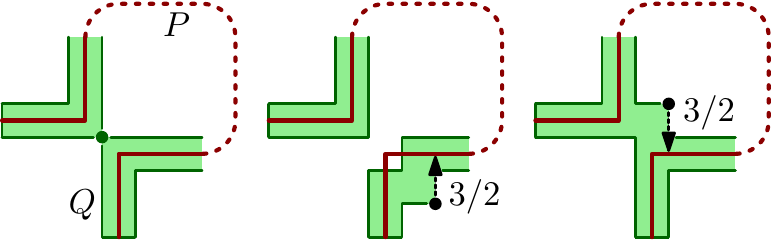}
	\caption{\label{fig:hausdorff-lower-bound} A polygon that does not admit a grid polygon with Hausdorff distance smaller than $3/2$. The brown line signifies an infinitesimally thin polygon.}
\end{figure}

\subsection{Algorithm}
\label{sec:hausdorff-algorithm}

To compute a grid polygon for a given polygon $P$ with $n$ edges, we need to determine the cells in the sets 
$Q_1$--$Q_4$. This is easy once we know which cells intersect $\partial P$. One way to
do this is to trace the edges of $P$ in the grid. The time this takes is proportional to the
number of crossings between cells and $\partial P$.
Let us denote the number of grid
cells that intersect $\partial P$ by $b$. Clearly, there are simple polygons with $\Theta(nb)$
polygon boundary-to-cell crossings. We show how to achieve a time bound of $O(n+B)$, where $B$ is the number of cells in the output.
The key idea is to first compute the Minkowski sum of $\partial P$ with a square of side length 2 and use that to quickly find the cells intersecting $\partial P$.

To compute this Minkowski sum we first compute the vertical decomposition of $\partial P$, see Fig.~\ref{fig:thickening}.
For every of the $O(n)$ quadrilaterals, determine the parts that are within vertical distance $1$ from 
the bounding edges.
The result $P'$ is a simple polygon with holes with a total of $O(n)$ edges,
and $\partial P \subset P'$.
We compute the horizontal decomposition of every hole and the exterior of $P'$
and determine all parts that are within horizontal distance $1$ from the bounding edges. We add this to $P'$, giving $P''$.
These steps take $O(n)$ time if we use Chazelle's triangulation algorithm~\cite{Chazelle1991}.
Essentially, the above steps constitute computing the Minkowski sum of $\partial P$ with a 
square of side length $2$, centered at the origin and axis-parallel. 

\begin{figure}
	\centering
	\includegraphics{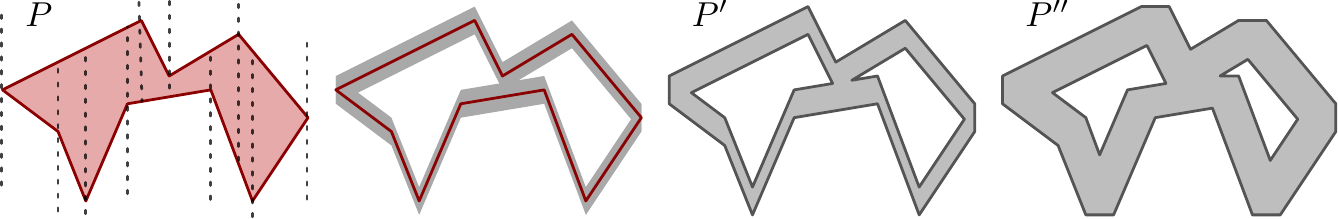}
	\caption{A simple polygon $P$ with its vertical decomposition, and the construction of $P'$ and $P''$.}
	\label{fig:thickening}
\end{figure}

\begin{lemma}
	For any cell $c$, at most four edges of $P''$ intersect its boundary twice.
\end{lemma}
\begin{proof}
	For any edge of $P''$, by construction, the whole part vertically above or below it over distance
	at least 2 is inside $P''$, and the same is true for left or right. For any edge $e$ that intersects
	the boundary of $c$ twice, one side of that edge is fully in the interior of $P''$, and hence,
	cannot contain other edges of $P''$. Hence, $e$ can be charged uniquely to a corner of $c$.
\end{proof}

\begin{cor}
	The number of polygon boundary-to-cell crossings of $P''$ is $O(n+b)$, where $b$
	is the number of grid cells intersecting $\partial P$.
\end{cor}

By tracing the boundary of $P''$, we can identify all cells that intersect it. 
Then we can determine all cells that intersect the boundary of $P$, because these are the
cells that lie fully inside $P''$. The modifications needed to find all cells whose module lies
inside $P$ are straightforward. In particular, we can find all cells whose module lies inside $P$,
but have a neighbor for which this is not the case in $O(n+b)$ time. 
This allows us to find the $O(B)$ cells selected in step $Q_1$ in $O(n+B)$ time. Steps $Q_2$ and $Q_3$ are now straightforward as well.

We now have a number of connected components of chosen grid cells.
No component has holes, and if there are $k$ components, we can connect
them into one with only $k-1$ extra grid cells.
We walk around the perimeter of some component and
mark all non-chosen cells adjacent to it. If a cell is marked twice, it is immediately
removed from consideration. Cells that are marked once but are adjacent to
two chosen cells will merge two different components.
We choose one of them, then walk around the perimeter of the new part and mark the
adjacent cells. Again, cells that are marked twice (possibly, both times from the new part,
or once from the old and once from the new part) are removed from consideration.
Continuing this process unites all components without creating holes.

\begin{theorem}
	For any simple polygon $P$ with $n$ edges, we can determine a set of $B$ cells
	that together form a grid polygon $Q$ in $O(n+B)$ time, such that
	$d_H(P, Q),\,d_H(\partial P, \partial Q) \leq \frac{1}{2}\sqrt{2}$ and
	$d_H(Q, P),\,d_H(\partial Q, \partial P) \leq \frac{3}{2}\sqrt{2}$.
\end{theorem}

\subsection{Heuristic improvements}
\label{sec:hausdorff-heuristics}

The grid polygon $Q$ constructed in Section~\ref{sec:hausdorff-construction} does not follow the shape of $P$ closely (see Fig.~\ref{fig:hausdorff-example}). Although the boundary of $Q$ remains close to the boundary of $P$, it tends to zigzag around it due to the way it is constructed. As a result, the symmetric difference between $P$ and $Q$ is relatively high. 
We consider two modifications of our algorithm to reduce the symmetric difference between $P$ and $Q$ while maintaining a small Hausdorff distance:
\begin{enumerate}
	\item We construct $Q_4$ with symmetric difference in mind.
	\item We post-process the resulting polygon $Q$ by adding, removing, or shifting cells.
\end{enumerate}

\subparagraph{Construction of $\boldsymbol{Q_4}$}
Instead of picking cells arbitrarily when constructing $Q_4$ we improve the construction with two goals in mind: (1) to directly reduce the symmetric difference between $P$ and $Q$, and (2) to enable the post-processing to be more effective. To that end, we construct $Q_4$ by repeatedly adding the cell $c$ (not introducing holes) that has the largest overlap with $P$. These cells hence reduce the symmetric difference between $P$ and $Q$ the most. 

\subparagraph{Post-processing}
After computing the grid polygon $Q$, we allow three operations to reduce the symmetric difference: (1) adding a cell, (2) removing a cell, and (3) shifting a cell to a neighboring position. These operations are applied iteratively until there is no operation that can reduce the symmetric difference. Every operation must maintain the following conditions: (1) $Q$ is simply connected, and (2) the Hausdorff distance between $P$ ($\partial P$) and $Q$ ($\partial Q$) is small. For the second condition we allow a slight relaxation with regard to the bounds of Lemma~\ref{lem:upperboundPtoQ}: $d_H(P, Q)$ and $d_H(\partial P, \partial Q)$ can be at most $\frac{3}{2}\sqrt{2}$ (like $d_H(Q, P)$ and $d_H(\partial Q, \partial P)$). This relaxation gives the post-processing more room to reduce the symmetric difference.

\section{Fr\'echet distance}
\label{sec:frechet}

The Fr\'echet distance $d_F$ between two curves is generally considered a better measure for similarity
than the Hausdorff distance.
For an input polygon $P$, we consider computing a grid polygon $Q$ such that $d_F(\partial P, \partial Q)$ is bounded by a small constant.
We study under what conditions on $\partial P$ this is possible and prove an upper and lower bound.
However, if $\partial P$ zigzags back and forth within a single row of grid cells, any grid polygon must have a large Fr\'echet distance: the grid is too coarse to follow $\partial P$ closely.
To account for this in our analysis, we introduce a realistic input model, as explained below.

\subparagraph{Narrow polygons}
For $a, b \in \partial P$, we use $|ab|_{\partial P}$ to denote the perimeter distance, i.e., the shortest distance from $a$ to $b$ along $\partial P$.
We define \emph{narrowness} as follows.
\begin{definition}
	A polygon $P$ is \emph{$(\alpha, \beta)$-narrow}, if for any two points $a, b \in \partial P$ with $|ab| \leq \alpha$, $|ab|_{\partial P} \leq \beta$.
\end{definition}

Given a value for $\alpha$, we refer to the minimal $\beta$ as the $\alpha$-narrowness of a polygon. We assume $\alpha < \beta$, to avoid degenerately small polygons.
We note that narrowness is a more forgiving model than straightness~\cite{alt3}.
A polygon $P$ is \emph{$\kappa$-straight} if for any two points $a, b \in \partial P$, $|ab|_{\partial P} \leq \kappa \cdot \|a - b\|$.
A $\kappa$-straight polygon is $(\alpha, \kappa\alpha)$-narrow for any $\alpha$, but not the other way around.
In particular, a finite polygon that intersects itself (or comes infinitesimally close to doing so) has a bounded narrowness, whereas its straightness becomes unbounded.

\subparagraph{An upper bound}
With our realistic input model in place, we can bound the Fr\'echet distance needed for a grid polygon from above.
In particular, we prove the following theorem.

\begin{theorem}
	\label{thm-frechet-upper}
	Given a $(\sqrt{2}, \beta)$-narrow polygon $P$ with $\beta\geq\sqrt{2}$, there exists a grid polygon $Q$ such that $d_F(\partial P, \partial Q) \leq (\beta + \sqrt{2}) / 2$.
\end{theorem}
\begin{proof}
	To prove the claimed upper bound, we construct $Q$ via a grid cycle $C$ that defines $\partial Q$.
	The construction is illustrated in Fig.~\ref{fig:example-fr-new}.
	\begin{figure}[t]
		\centering
		\includegraphics{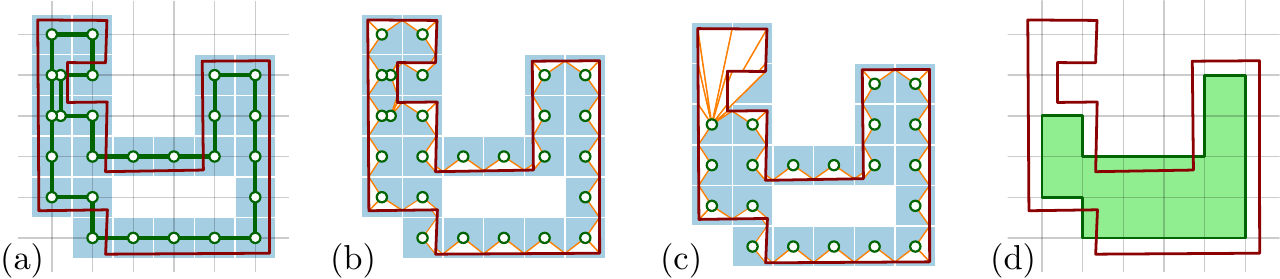}
		\caption{\label{fig:example-fr-new}Constructing $Q$ for the upper bound on the Fr\'echet distance. (a) Input polygon on the grid and the squares it visits (shaded); initial state of $C$ with revisited vertices slightly offset for legibility. (b) Initial mapping $\mu$ (white triangles) between the vertices of $C$ and $\partial P$. (c) Removal of duplicate vertices in $C$, and its effect on $\mu$. (d) Resulting cycle represents a grid polygon.}
	\end{figure}
	We define the \emph{square} of a grid-graph vertex $v$ to be the $1 \times 1$-square centered on $v$.
	Let $C$ be the cyclic chain of vertices whose square is intersected by $\partial P$, in the order in which $\partial P$ visits them.
	We define a mapping $\mu$ between the vertices of $C$ and $\partial P$.
	In particular, for each $c \in C$, let $\mu(c)$ be the ``visit'' of $\partial P$ that led to $c$'s existence in $C$, that is, the part of $\partial P$ within the square of $c$.
	By construction, we have that $\| c - p_c \| \leq \sqrt{2}/2$ for all $c \in C$ and $p_c \in \mu(c)$.
	The visits $\mu(c)$ and $\mu(c')$ for two consecutive vertices, $c$ and $c'$, in $C$ intersect in a point (or, in degenerate cases, in a line segment) that lies on the common boundary of the squares of $c$ and $c'$; let $p$ denote such a point.
	For any point $\sigma$ on the line segment between $c$ and $c'$, we have that $\| \sigma - p \| \leq \max \{ \| c - p \| , \| c' - p \| \} \leq \sqrt{2}/2$, as the Euclidean distance is convex (i.e., its unit disk is a convex set).
	Hence, $\mu$ describes a continuous mapping on $\partial P$ and acts as a witness for $d_F(\partial P, C) \leq \sqrt{2}/2$.

	However, $C$ may contain duplicates and thus not describe a grid polygon $Q$.
	We argue here that we can remove the duplicates and maintain $\mu$ in such a way that it remains a witness to prove that $d_F(\partial P, C) \leq (\beta + \sqrt{2})/2$.
	Let $c$ and $c'$ be two occurrences in $C$ of the same vertex $v$.
	Let $p \in \mu(c)$ and $p' \in \mu(c')$, both in the square of $v$.
	As they lie within the same square, $\|p-p'\| \leq \sqrt{2}$ and hence we know that $|pp'|_{\partial P} \leq \beta$.
	Hence, at least one of the two subsequences of $C$ strictly in between $c$ and $c'$ maps via $\mu$ to a part of $\partial P$ that has length at most $\beta$.
	We pick one such subsequence and remove it as well as $c'$ from $C$.
	We concatenate to $\mu(c)$ the mapped parts of $\partial P$ from the removed vertices.
	As the length of the mapped parts is bounded by $\beta$, the maximal distance between any point on these mapped parts is $\beta/2 + \sqrt{2}/2$.
	Hence, after removing all duplicates, we are left with a cycle $C$, with $\mu$ as a witness to testify that $d_F(\partial P, C) \leq (\beta + \sqrt{2})/2$.

	If $C$ contains at least three vertices, it describes a grid polygon and we are done.
	However, if $C$ consists of at most two vertices, then it does not describe a grid polygon.
	We can extend $C$ easily into a 4-cycle for which the bound still holds (Lemma~\ref{lem-frechet-upper-degenerate}, Appendix~\ref{app-frechet}).
\end{proof}

The proof of the theorem readily leads to a straightforward algorithm to compute such a grid polygon.
The construction poses no restrictions on the order in which to remove duplicates and the decisions are based solely on the lengths of $\mu(v)$.
Hence, the algorithm runs in linear time by walking over $P$ to find $C$ and handling duplicates as they arise.

\subparagraph{Lower bound}
To show a lower bound, we construct a $(\sqrt{2}, \beta)$-narrow polygon $P$ for which there is no grid polygon with Fr\'echet distance smaller than $\frac{1}{4} \sqrt{\beta^2-2}$ to $P$, for any $\beta > \sqrt{2}$.
First, construct a polygonal line $L = (p_1, \ldots, p_n)$, where $n = 2\,\big\lceil \frac{1}{4} \sqrt{\beta^2-2} \,\big\rceil + 1$.
Vertex $p_i$ is $(0, i/2)$ if $i$ is odd and $(\frac{1}{2} \sqrt{\beta^2-2},\ i/\sqrt{2})$ otherwise.
Now, consider a regular $k$-gon with side length $(n - 1) / \sqrt{2}$ and $k \geq 4$ such that its interior angle is not smaller than $\varphi= \arccos{(1-4/\beta^2)}$.
Assume the $k$-gon has a vertical edge on the right-hand side.
We replace this edge by $L$ to construct our polygon $P$.
Fig.~\ref{fig:comb-graph} shows a polygon for $k = 4$ ($\beta \geq 2$) and for $k = 7$ ($\beta < 2$).

The two lemmas below readily imply our lower bound on the Fr\'echet distance. The proof of the first can be found in Appendix~\ref{app-frechet}.

%ARXIV: bh->t
\begin{figure}[t]
	\centering
	\includegraphics[page=1]{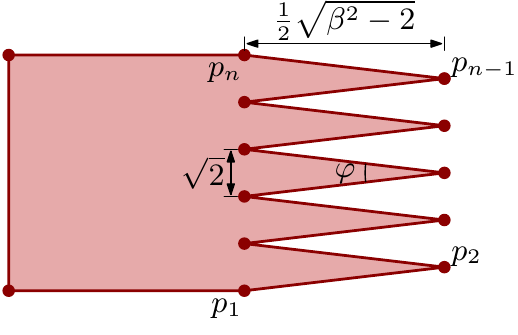}
	\hfill
	\includegraphics[page=2]{frechet-lower-bound}
	\hfill
	\includegraphics{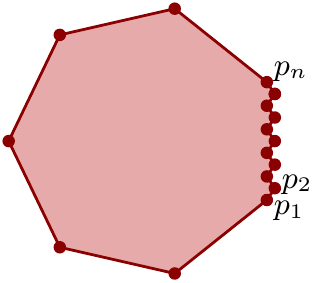}
 	\caption{Polygon $P$ (left) for which any grid polygon will have high Fr\'echet distance (center); polygon $P$ for $\beta<2$ (right).}
	\label{fig:comb-graph}
\end{figure}

\begin{lemma}
	\label{lem:p-obese}
	The polygon $P$ described above is $(\sqrt{2}, \beta)$-narrow.
\end{lemma}

\begin{lemma}
	\label{lem:p-frechet}
	For constructed polygon $P$ and any grid polygon $Q$, $d_F(\partial P, \partial Q) \geq \frac{1}{4} \sqrt{\beta^2-2}$.
\end{lemma}
\begin{proof}
We show this by contradiction: assume that a grid polygon $Q$ exists with $d_F(\partial P, \partial Q) = \varepsilon < \frac{1}{4} \sqrt{\beta^2-2}$.
For any vertex $p_i$ of $P$, there must be a point $q_i \in \partial Q$ (not necessarily a vertex) such that $\|p_i - q_i\| < \varepsilon$.
Moreover, these points $q_1, \ldots, q_n$ need to appear on $\partial Q$ in order.
Equivalently, if we draw disks with radius $\varepsilon$ centered at $p_1, \ldots, p_n$, curve $\partial Q$ needs to visit these disks in order.

The disks centered at $p_1, p_3, \ldots, p_n$ never intersect the disks centered at $p_2, p_4, \ldots, p_{n-1}$. In particular, the disks centered at $p_1, p_3, \ldots, p_n$ are all to the left of the vertical line $v \colon x = \frac{1}{4} \sqrt{\beta^2-2}$, and all disks centered at $p_2, p_4, \ldots, p_{n-1}$ are all to the right of this line. Hence, between $q_1$ and $q_2$, $\partial Q$ must contain at least one horizontal line segment crossing line $v$ to the right, and between $q_2$ and $q_3$ there must be at least one horizontal segment crossing $v$ to the left, and so on until we reach $q_{n}$.
Since $Q$ is simple, this requires that the difference between the maximum and the minimum $y$-coordinate of the these horizontal segments on $\partial Q$ is at least $n - 1$.
The $y$-difference between $p_1$ and $p_n$ is only $(n-1) / \sqrt{2}$.
This implies $d_F(\partial P, \partial Q) \geq n - 1 - (n-1) / \sqrt{2} > \frac{1}{4} \sqrt{\beta^2-2}$ and thus contradicts our assumption.
\end{proof}

\begin{theorem}
	\label{thm:frechet-lower}
	For any $\beta > \sqrt{2}$, there exists a $(\sqrt{2}, \beta)$-narrow polygon $P$ such that $d_F(\partial P, \partial Q) \geq \frac{1}{4} \sqrt{\beta^2-2}$ holds for any grid polygon $Q$.
\end{theorem}

\section{Experiments}
\label{sec:experimental-results}

%ARXIV set wraplines to 12 instead of 10
\begin{wrapfigure}[12]{r}{0.35\textwidth} %
	\raggedleft
	\centering
	\vspace{-1em}
	%\vspace{-2em}%ARXIV
	\includegraphics{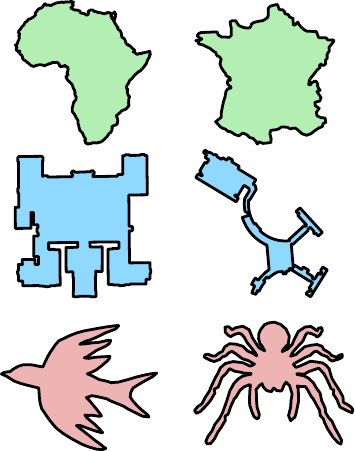}
	\vspace{-.1em}
	\caption{The input categories.}
	\label{fig:input}
\end{wrapfigure}
Here, we apply our algorithms to a set of polygons that can be encountered in practice.
We investigate the performance of the Hausdorff algorithm and its heuristcs as well as the Fr\'echet algorithm.
Moreover, we consider the effects of grid resolution and the placement of the input.

\subparagraph{Data set}
We use a set of $34$ polygons: $14$~territorial outlines (countries, provinces, islands), $11$~building footprints and $9$~animal silhouettes (see Fig.~\ref{fig:input} for six examples; a full list is given in Appendix~\ref{sec:app-inputs}). We scale all input polygons such that their bounding box has area $r$; we call $r$ the resolution.
Unless stated otherwise, we use $r = 100$. This scaling is used to eliminate any bias introduced from comparing different resolutions.

\subsection{Symmetric difference}
\label{sec:symmetric-difference}

We start our investigation by measuring the symmetric difference between the input and output polygon.
If the symmetric difference is small, this indicates that the output is similar to the input.
We normalize the symmetric difference by dividing it by the area of the input polygon.
The results of our algorithms depend on the position of the input polygon relative to the grid. Hence, for every input polygon we computed the average normalized symmetric difference over $20$~random placements.

Computing a (simply connected) grid polygon that minimizes symmetric difference is NP-hard~\cite{meulemans2016esa}.
Hence, as a baseline for our comparison, we compute the set of cells with the best possible symmetric difference by simply taking all cells that are covered by the input polygon for at least $50\,\%$. This set of cells is optimal with respect to symmetric difference but may not be simply connected. It can hence be thought of as a lower bound. 

\subparagraph{Overview}
In Table~\ref{tab:mst-methods}, we compare the Fr\'echet algorithm and the various instantiations of the Hausdorff algorithm in terms of the (normalized) symmetric difference.
The second column lists the average symmetric difference of the symmetric-difference optimal solution, calculated as described above.
The other columns are hence given as a percentage representing the increase with respect to the optimal value.
We aggregated the results per input type; the full results are in Appendix~\ref{app:experiments}.

\begin{table}[t]
	\centering
	\caption{\label{tab:mst-methods} Normalized symmetric difference, as an increase percentage w.r.t.\ optimal, of the algorithms. Note that ``optimal'' here means optimal for the symmetric difference when not insisting on a connected set of cells. For the Hausdorff algorithm, results for the various heuristic improvements are shown. In the second row, \textit{None} means that no postprocessing heuristic was used; \textit{A}, \textit{R} and \textit{S} mean additions, removals and shifts, respectively. In the third row, \cmark\ and \xmark\ indicate whether $Q_4$ was chosen arbitrarily (\xmark) or using the symmetric difference heuristic (\xmark).}
\medskip %ARXIV : added some space
	\begin{tabular}{rccccccccccc}
		\toprule
		& \textbf{Optimal} & \multicolumn{6}{c}{\textbf{Hausdorff}} & \textbf{Fr\'echet} \\
		\cmidrule(lr){3-8}
		\textit{postproc.} & & \multicolumn{2}{c}{\textit{None}} &  \multicolumn{2}{c}{\textit{A\,/\,R}} & \multicolumn{2}{c}{\textit{A\,/\,R\,/\,S}} \\
		\cmidrule(lr){3-4}
		\cmidrule(lr){5-6}
		\cmidrule(lr){7-8}
		\textit{$Q_4$ heur.} & & \xmark & \cmark & \xmark & \cmark &\xmark & \cmark \\
		\midrule
		Maps & $0.223$ & $+\,316\,\%$ & $+\,238\,\%$ & $+\,39\,\%$ & $+\,3\,\%$ & $+\,11\,\%$ & $+\,3\,\%$ & $+\,23\,\%$ \\
		Buildings & $0.257$ & $+\,270\,\%$ & $+\,197\,\%$ & $+\,47\,\%$ & $+\,9\,\%$ & $+\,21\,\%$ & $+\,8\,\%$ & $+\,17\,\%$ \\
		Animals & $0.333$ & $+\,246\,\%$ & $+\,188\,\%$ & $+\,60\,\%$ & $+\,12\,\%$ & $+\,29\,\%$ & $+\,11\,\%$ & $+\,8\,\%$ \\
		\bottomrule
	\end{tabular}
\end{table}

The table tells us that, with the use of heuristics, the Hausdorff algorithm gets quite close to the optimal symmetric difference, while still bounding the Hausdorff distance as well as guaranteeing a grid polygon.
The Fr\'echet algorithm is performing more poorly in comparison, though interestingly performs \emph{better} on the animal contours.

Fig.~\ref{fig:ostrich} shows three solutions for one of the input polygons: symmetric-difference optimal, Fr\'echet algorithm and Hausdorff algorithm with heuristics.
The symmetric-difference optimal solution looks like the input, but consists of multiple disconnected polygons.
The result of the Fr\'echet algorithm is a single grid polygon, but the algorithm cuts off at narrow parts.
The result of the Hausdorff algorithm is also a single grid polygon, but does not have to cut off parts when input is narrow.

\begin{figure}[b]
	\centering
	\vspace{-0.6\baselineskip}%ARXIV: smaller space
	\includegraphics{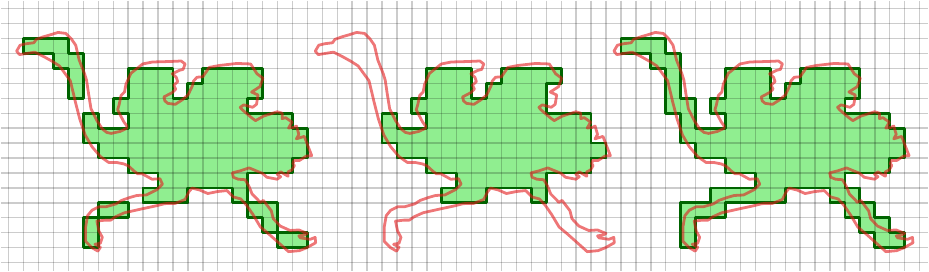}
	\vspace{-0.4\baselineskip}%ARXIV: smaller space
    \caption{Example outputs for the symmetric-difference optimal algorithm  (left), the Fr\'echet algorithm (center) and the Hausdorff algorithm (right). Note that the first does not yield a grid polygon.}
    \label{fig:ostrich}
\end{figure}

Below, we examine the effect of the different heuristics for the Hausdorff algorithm to explain its success.
Moreover, we show that the performance of the Fr\'echet algorithm is highly dependent on the grid resolution.

\subparagraph{Hausdorff heuristics}
Table~\ref{tab:mst-methods} shows that using the heuristic for $Q_4$ makes a tremendous difference, especially if a postprocessing heuristic is used as well.
Fig.~\ref{fig:msts} illustrates this finding with four results on the same input.
In (a--b) $Q_4$ is chosen arbitrarily and the resulting shape does not look like the input---even after postprocessing.
In particular, the postprocessing heuristic cannot progress further: the cell marked $c$ cannot be added to $Q$ since that would increase the symmetric difference.
In (c--d) $Q_4$ is chosen using the heuristic; it provides a better initial solution which allows the postprocessing to create a nice result.

\begin{figure}[t]
	\centering
	\includegraphics{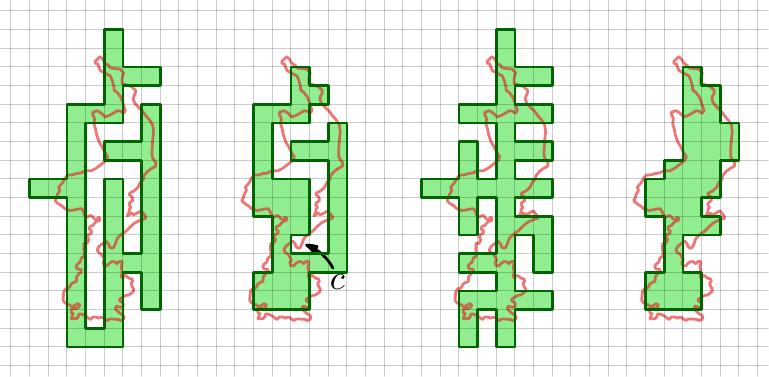}
    \caption{Without the heuristic for the $Q_4$ construction (a), the algorithm gets stuck in the post processing phase (b). The smart $Q_4$ construction gives a better starting point (c) resulting in the desired shape (d).}
	\label{fig:msts}
\end{figure}

In the postprocessing heuristic, allowing or disallowing shifts can influence the result. See for example Fig.~\ref{fig:shifts-are-useful}.
Without shifts, the heuristic cannot move the connection between the two ends of the input polygon to the correct location as it would first need to increase the symmetric difference. With a series of diagonal shifts this can be achieved. Our experiments show that in practice allowing shifts indeed decreases the symmetric difference.
However, the effect is only marginal if we use the heuristic for the $Q_4$ construction.
Hence, we conclude that shifts only significantly improve the result if $Q_4$ is chosen badly.

\begin{figure}[t]
	\centering
	\includegraphics{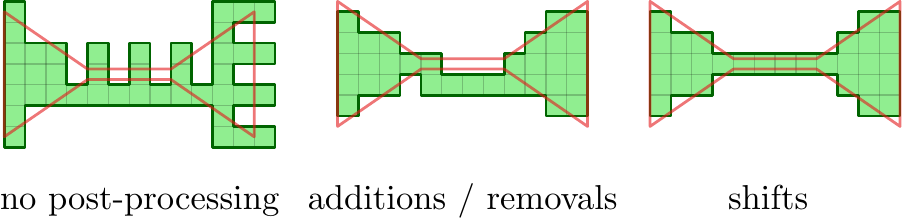}
	\caption{Without allowing shifts, the post-processing phase cannot move the cells in the middle to coincide with the input polygon. With shifts, this is possible.}
	\label{fig:shifts-are-useful}
\end{figure}

\subparagraph{Resolution and placement}
While developing our algorithm we noticed that not just the grid resolution but also the placement of the input polygon effected the symmetric difference. Hence we set up experiments to investigate these factors. First we tested how much the resolution influences the symmetric differences. 
In Table~\ref{tab:resolution}, the results are shown, averaged over all 34 inputs. 
As expected, for all algorithms, the normalized symmetric difference decreases when the resolution increases.

\begin{table}[h]
	\caption{Normalized symmetric difference for the various algorithms on five resolutions.}
	\label{tab:resolution}
	\centering
	\begin{tabular}{rccccc}
		\toprule
		& $\boldsymbol{r = 100}$ & $\boldsymbol{r = 225}$ & $\boldsymbol{r = 400}$ & $\boldsymbol{r = 625}$ & $\boldsymbol{r = 900}$ \\
		\midrule
		Optimal & $0.263$ & $0.188$ & $0.147$ & $0.119$ & $0.101$ \\
		Hausdorff & $0.282$ & $0.201$ & $0.155$ & $0.123$ & $0.103$ \\
		Fr\'echet & $0.306$ & $0.227$ & $0.184$ & $0.148$ & $0.122$ \\
		\bottomrule
	\end{tabular}
\end{table}

To investigate how much the results of our algorithms depend on the placement of the input polygon, we compared the minimal, maximal and average symmetric difference over 20 runs of the algorithms. The polygons were placed randomly for each run, but per polygon the same $20$~positions were used for all three algorithms.
The results of this experiment are shown in Appendix~\ref{app:experiments},~Table~\ref{tab:input-placement}. For every input polygon and algorithm, the minimum, average and maximum symmetric difference of all $20$~runs is shown. The difference between the minimum and the maximum symmetric difference for each algorithm\,/\,polygon combination is rather large.
We found that placement can have a significant effect on the achieved symmetric difference.
Hence, if the application permits us to choose the placement, it is advisable to do so to obtain the best possible result.
This leads to an interesting open question of whether we can algorithmically optimize the placement, to avoid the need to find a good placement with trial and error.
In the upcoming analysis, we also consider the effect of resolution and placement, with respect to the Fr\'echet distance.

\subsection{Fr\'echet analysis}

Theorem~\ref{thm-frechet-upper} predicts an upper bound on the Fr\'echet distance based on $\sqrt{2}$-narrowness.
However, if the points defining the narrowness lie within different squares of grid vertices, this bound may be naive.
Moreover, it assumes a worst-case detour, going away in a thin triangle to maximize the distance between the detour and a doubly-visited cell.
Hence, the algorithm has the potential to perform better, depending on the actual geometry and its placement with respect to the grid.
Here, we discuss our investigation of these effects; refer to Appendix~\ref{app-exp-frechet} for more details.

\subparagraph{Procedure}
We use the 34~polygons described in Appendix~\ref{sec:app-inputs}.
As we may expect the grid resolution to significantly affect results, we used 20 different resolutions.
In particular, we use resolutions varying from $10\,000$ to $25$, using $(100/s)^2$ with scale $s \in \{ 1, \ldots, 20 \}$.

For each resolution-polygon combination (case), we measure its $\sqrt{2}$-narrowness (using the algorithm described in Appendix~\ref{app-measure-narrow}) and derive the predicted upper bound.
Then, we run the Fr\'echet algorithm, using the 25 possible offsets in $\{0, 0.2, 0.4, 0.6, 0.8 \}^2$, and measure the precise Fr\'echet distance between input and output.
We keep track of three summary statistics for each case: the minimum (best), average (``expected'') and maximum (worst) measured Fr\'echet distance.

\subparagraph{Effect of placement}
We consider placement with respect to the grid (offset) to have a significant effect on the result computed for a polygon, if the difference between the maximal and minimal Fr\'echet distance over the 25 offsets is at least $2$.
Almost $30\,\%$ of cases exhibit such a significant effect, with the animal contours being particularly affected ($35\,\%$ significant).
Again, this raises the question of whether we can algorithmically determine a good placement.

\subparagraph{Upper bound quality}
We define the performance as the measured Fr\'echet distance as a percentage of the upper bound.
We consider the algorithm's performance significantly better than the upper bound, if it is less than $40\,\%$.
Using the best placement, over $95\,\%$ of cases perform significantly better.
Averaging performance over placement, we still find such a majority (over $81\,\%$).
Interestingly, this drop is mostly due to the animal contours, of which only $63\,\%$ now perform significantly better.
Thus, although we have a provable upper bound, we may typically expect our simple algorithm to perform significantly better than the upper bound. 
This holds even without any postprocessing to further optimize the result and when taking a random offset.

\subparagraph{Effect of resolution}
The influence of the resolution on the above results does not seem to exhibit a clear pattern.
Nonetheless, resolution likely plays an important role in these results, but not as straightforward as either low or high resolution being more problematic.
Instead, it is likely the most problematic resolutions are those at which the $\sqrt{2}$-narrowness of the polygon jumps as a new pair of edges comes within distance $\sqrt{2}$ of each other. 
However, an in-depth investigation of this is beyond the scope of this paper.

\subparagraph{Heuristic improvement}
In contrast to the Hausdorff algorithm, the Fr\'echet algorithm needs no heuristic improvement on inputs that are not too narrow.
However, badly placed narrow polygons can be problematic: large parts of the polygon may be cut, greatly diminishing similarity.
A solution may be to select an appropriate resolution (if our application permits us to).
In our experiments the algorithm tends to perform well at resolutions where the symmetric-difference optimal solution is a single grid polygon.
The advantage of our Fr\'echet algorithm is that it guarantees a grid polygon on all outputs and bounds the Fr\'echet distance.

\begin{wrapfigure}[8]{r}{0.27\linewidth}
	\vspace{-\baselineskip}
	\centering
	\includegraphics{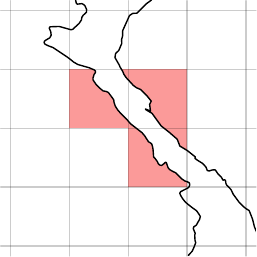}
	\caption{Red cells cause a cut-off and have high symmetric difference.}
	\label{fig-frechet-sdheur-problem}
\end{wrapfigure}
Nonetheless, we may want to consider heuristic postprocessing to obtain a locally-optimal result.
If we want to do this in terms of the symmetric difference, we may use similar techniques as for the Hausdorff algorithm.
However, this does not perform well: the narrow strip that causes the Fr\'echet algorithm to perform badly tends to effect a high symmetric difference for the nearby grid cells (Fig.~\ref{fig-frechet-sdheur-problem}).
As such, the result is already (close to) a local optimum in terms of the symmetric difference.

\section{Conclusion}
\label{sec:conclusion}

We presented two algorithms to map simple polygons to grid polygons that capture the shape of the polygon well. For measuring the distance between the input and the output, we considered the Hausdorff and the Fr\'echet distance. We achieved a constant bound on the Hausdorff distance; for the Fr\'echet distance we require a realistic input assumption to achieve a constant bound. We also evaluated our algorithms in practice. Although the Hausdorff algorithm does not produce great results directly, the algorithm achieves good results when combined with heuristic improvements. The Fr\'echet algorithm, on the other hand, struggles with narrow polygons, and it is not clear how to improve the results using heuristics. Designing an algorithm for the Fr\'echet distance that also works well in practice remains an interesting open problem. Another interesting open problem is to algorithmically optimize the placement of the input polygon, for the best results of both the Hausdorff and the Fr\'echet algorithm.

\newpage

\bibliography{p197-bouts}

\begin{thebibliography}{10}

\bibitem{alt1}
Helmut Alt, Bernd Behrends, and Johannes Bl{\"{o}}mer.
\newblock Approximate matching of polygonal shapes.
\newblock {\em Annals of Mathematics {\&} Artificial Intelligence},
  13(3-4):251--265, 1995.

\bibitem{alt2}
Helmut Alt and Michael Godau.
\newblock Computing the {F}r{\'{e}}chet distance between two polygonal curves.
\newblock {\em International Journal of Computational Geometry {\&}
  Applications}, 5:75--91, 1995.

\bibitem{alt3}
Helmut Alt, Christian Knauer, and Carola Wenk.
\newblock Comparison of distance measures for planar curves.
\newblock {\em Algorithmica}, 38(1):45--58, 2003.

\bibitem{Althaus04}
Ernst Althaus, Friedrich Eisenbrand, Stefan Funke, and Kurt Mehlhorn.
\newblock Point containment in the integer hull of a polyhedron.
\newblock In {\em Proceedings of the 15th Annual {ACM-SIAM} Symposium on
  Discrete Algorithms (SODA)}, pages 929--933, 2004.

\bibitem{batenburg2009}
K.~Joost Batenburg, Sjoerd Henstra, Walter~A. Kosters, and Willem~Jan
  Palenstijn.
\newblock Constructing simple nonograms of varying difficulty.
\newblock {\em Pure Mathematics and Applications (Pu.~MA)}, 20:1--15, 2009.

\bibitem{Berend14}
Daniel Berend, Dolev Pomeranz, Ronen Rabani, and Ben Raziel.
\newblock Nonograms: Combinatorial questions and algorithms.
\newblock {\em Discrete Applied Mathematics}, 169:30--42, 2014.

\bibitem{cano}
Rafael~G. Cano, Kevin Buchin, Thom Castermans, Astrid Pieterse, Willem Sonke,
  and Bettina Speckmann.
\newblock Mosaic drawings and cartograms.
\newblock {\em Computer Graphics Forum}, 34(3):361--370, 2015.

\bibitem{Chazelle1991}
Bernard Chazelle.
\newblock Triangulating a simple polygon in linear time.
\newblock {\em Discrete {\&} Computational Geometry}, 6:485--524, 1991.

\bibitem{chun}
Jinhee Chun, Matias Korman, Martin N{\"{o}}llenburg, and Takeshi Tokuyama.
\newblock Consistent digital rays.
\newblock {\em Discrete {\&} Computational Geometry}, 42(3):359--378, 2009.

\bibitem{berg}
Mark de~Berg, Dan Halperin, and Mark~H. Overmars.
\newblock An intersection-sensitive algorithm for snap rounding.
\newblock {\em Computational Geometry}, 36(3):159--165, 2007.

\bibitem{Devillers2006}
Olivier Devillers and Philippe Guigue.
\newblock Inner and outer rounding of boolean operations on lattice polygonal
  regions.
\newblock {\em Computational Geometry}, 33(1-2):3--17, 2006.

\bibitem{guibas}
Michael~T. Goodrich, Leonidas~J. Guibas, John Hershberger, and Paul~J.
  Tanenbaum.
\newblock Snap rounding line segments efficiently in two and three dimensions.
\newblock In {\em Proceedings of the 13th Annual Symposium on Computational
  Geometry (SoCG)}, pages 284--293, 1997.

\bibitem{yao}
Daniel~H. Greene and F.~Frances Yao.
\newblock Finite-resolution computational geometry.
\newblock In {\em Proceedings of the 27th Annual Symposium on Foundations of
  Computer Science (FOCS)}, pages 143--152, 1986.

\bibitem{Harvey99}
Warwick Harvey.
\newblock Computing two-dimensional integer hulls.
\newblock {\em {SIAM} Journal on Computing}, 28(6):2285--2299, 1999.

\bibitem{hershberger}
John Hershberger.
\newblock Stable snap rounding.
\newblock {\em Computational Geometry}, 46(4):403--416, 2013.

\bibitem{ham-cycle-grid-graph}
Alon Itai, Christos~H. Papadimitriou, and Jayme~Luiz Szwarcfiter.
\newblock Hamilton paths in grid graphs.
\newblock {\em SIAM Journal on Computing}, 11(4):676--686, 1982.

\bibitem{klette2}
Reinhard Klette and Azriel Rosenfeld.
\newblock {\em Digital Geometry -- geometric methods for digital picture
  analysis}.
\newblock Morgan Kaufmann, 2004.

\bibitem{klette}
Reinhard Klette and Azriel Rosenfeld.
\newblock Digital straightness -- a review.
\newblock {\em Discrete Applied Mathematics}, 139(1-3):197--230, 2004.

\bibitem{kopf2013}
Johannes Kopf, Ariel Shamir, and Pieter Peers.
\newblock Content-adaptive image downscaling.
\newblock {\em ACM Transactions on Graphics}, 32(6):Article No. 173, 2013.

\bibitem{meulemans}
Wouter Meulemans.
\newblock {\em Similarity Measures and Algorithms for Cartographic
  Schematization}.
\newblock PhD thesis, Technische Universiteit Eindhoven, 2014.

\bibitem{meulemans2016esa}
Wouter Meulemans.
\newblock Discretized approaches to schematization.
\newblock {\em CoRR}, 2016.

\bibitem{Ortiz-Garcia2007}
Emilio~G. Ort{\'{\i}}z{-}Garc{\'{\i}}a, Sancho Salcedo{-}Sanz, Jos{\'{e}}~M.
  Leiva{-}Murillo, {\'{A}}ngel~M. P{\'{e}}rez{-}Bellido, and Jos{\'{e}}~Antonio
  Portilla{-}Figueras.
\newblock Automated generation and visualization of picture-logic puzzles.
\newblock {\em Computers {\&} Graphics}, 31(5):750--760, 2007.

\end{thebibliography}

\newpage

\appendix

\section{Hausdorff distance}
\label{app-hausdorff}

\begin{cloneclaim}{Theorem~\ref{thm:hard}}
	Given a polygon $P$, it is NP-hard to decide whether there exists a grid polygon $Q$ such that both $d_H(\partial P, \partial Q) \leq 1/2$ and $d_H(\partial Q, \partial P)\leq 1/2$.
\end{cloneclaim}
\begin{proof}
	We reduce from the problem of finding a Hamiltonian cycle in a (``partial'') grid graph~\cite{ham-cycle-grid-graph}. 
	Given such a grid graph $G$, we construct a polygon $P$ such that a Hamiltonian cycle in $G$ exists if and only if there exists a grid polygon $Q$ such that both $d_H(\partial P, \partial Q) \leq 1/2$ and $d_H(\partial Q, \partial P)\leq 1/2$.

	\begin{figure}[b]
		\begin{minipage}[t]{0.475\textwidth}
			\centering
			\includegraphics[page=3]{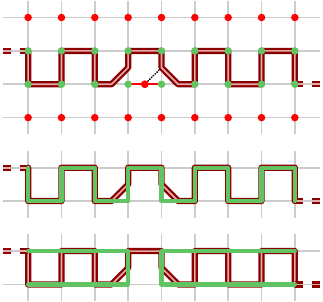}
			\caption{Edge gadget (top): the red segment cannot lie on $\partial Q$ because the Hausdorff distance from the middle point to the chain is $1/\sqrt{2}>1/2$. There are two ways of covering the edge construction with a grid chain with the Hausdorff distance $1/2$ (middle and bottom).}
			\label{fig:hausdorff-hard-edge}
		\end{minipage}\hfill
		\begin{minipage}[t]{0.475\textwidth}
			\includegraphics[page=4]{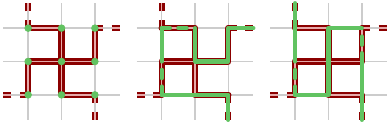}
			\caption{Vertex gadget (left) can be connected with up to four edges. There are two ways of covering the vertex gadget with a polygonal chain within the Hausdorff distance $1/2$ (middle and right). Dashed edges can be extended to cover a half of an adjacent edge.}
			\label{fig:hausdorff-hard-vertex}
		\end{minipage}
	\end{figure}

	To construct $P$ from $G$, we replace the edges of $G$ with \emph{edge gadgets} (see Fig.~\ref{fig:hausdorff-hard-edge}) and vertices of $G$ with \emph{vertex gadgets} (see Fig.~\ref{fig:hausdorff-hard-vertex}).

	\subparagraph{Edge gadget}
	The edge gadget consists of two ``zigzag'' polygonal chains following the grid that are placed $\epsilon$-close to each other for a small constant $\epsilon$. The area between these chains will eventually become the interior of polygon $P$. In the middle of the zigzag, we cut off two corners, as illustrated. 
	As a result, a closest point on the zigzag chain to the middle of the red segment in Fig.~\ref{fig:hausdorff-hard-edge} (top) is $1/\sqrt{2}$ distance away, and thus the boundary of the grid polygon $Q$ cannot pass through this segment. Moreover, the distance from any of the red points in the figure to the chain is larger than $1/2$. By construction there will be no other parts of $P$ closer than $1/2$ to any of the red points. Therefore, $\partial Q$ cannot pass through them. Furthermore, $\partial Q$ must pass through all the green points, otherwise the Hausdorff distance from $\partial P$ to $\partial Q$ would be too large. Therefore we can conclude that there are only two ways of covering the edge with $\partial Q$ shown in the Fig.~\ref{fig:hausdorff-hard-edge} (middle and bottom).

	\subparagraph{Vertex gadget}
	The vertex gadget consists of four pieces of polygonal chains forming a cross-like pattern such that along every branch of the cross there are two polygonal chains that are placed $\epsilon$-close to each other. This is schematically illustrated in Fig.~\ref{fig:hausdorff-hard-vertex} (left); observe that to obtain the $1/2$ bound, we must be careful with this construction, moving the edges inward appropriately to ensure that all connections are possible within the desired Hausdorff distance, this is shown in Fig.~\ref{fig:hausdorff-hard-vertex-detail}. 
	The area between these parts of chains will also eventually become the interior of polygon $P$. There are nine grid points in the vertex gadget that need to be traversed by the boundary of the grid polygon $Q$ to achieve a Hausdorff distance of $1/2$. Consider an example of a degree-3 vertex gadget connected to three edge gadgets in Fig.~\ref{fig:hausdorff-hard-vertex-2}.  Similarly to the edge gadget case we can argue that $\partial Q$ must pass through all green points and cannot pass through any of the red points to achieve the Hausdorff distance of $1/2$. Considering this case, and also the cases of 1-, 2-, and 4-degree vertices, we can observe the following: a chain of $\partial Q$ must enter the vertex along one of the adjacent edges following the ``zigzag'' pattern from the Fig.~\ref{fig:hausdorff-hard-edge} (middle), cover the nine grid points, and leave along another adjacent edge following the same ``zigzag'' pattern. If the vertex has more adjacent edges, those edges can only be partially covered with the U-turn pattern from the Fig.~\ref{fig:hausdorff-hard-edge} (bottom); the cut-off corners ensure that the U-turn cannot pass the middle of the edge-construction. A vertex gadget can be entered by $\partial Q$ only once. The two ways for $\partial Q$ to enter and exit the vertex gadget are shown in Fig.~\ref{fig:hausdorff-hard-vertex} (center and right).

	\begin{figure}[t]
	\begin{minipage}[b]{0.475\textwidth}
			\centering
			\includegraphics[page=1,width=0.8\linewidth]{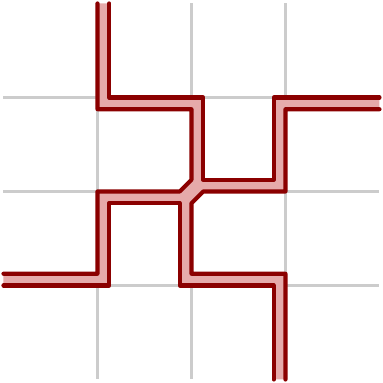}
			\caption{The precise construction of the vertex gadget, to ensure a Hausdorff distance of $1/2$ with constant $\epsilon$.\newline }
			\label{fig:hausdorff-hard-vertex-detail}
		\end{minipage}\hfill
		\begin{minipage}[b]{0.475\textwidth}
		\centering
		\includegraphics[page=7]{appendix/hausdorff-hard}
		\caption{An example of a degree-3 vertex gadget connected to three edge gadgets. $\partial Q$ must pass through all green points and cannot cannot pass through red points.}
		\label{fig:hausdorff-hard-vertex-2}
		\end{minipage}
	\end{figure}

	\subparagraph{Putting the gadgets together} Now, given a grid graph $G$, replace its vertices with vertex gadgets, and replace its edges with edge gadgets, and connect their corresponding pairs of polygonal chains. 
	This construction leads to a polygon $P'$ that may have holes. To make a simple polygon $P$ out of $P'$, we cut some of the edges of $G$ to break all the cycles (see Fig.~\ref{fig:hausdorff-hard} (top left)). We introduce small cuts in the corresponding edge gadgets and reconnect the pairs of polygonal chains (see Fig.~\ref{fig:hausdorff-hard} (top right)). Thus, all the polygonal chains are connected in a cycle now forming a simple polygon $P$. The small cuts in the edge gadgets do not affect the way that it can be covered with $\partial Q$.

	\begin{figure}
		\begin{minipage}[c]{0.4\textwidth}
			\centering
			\includegraphics[page=1, scale=1.5]{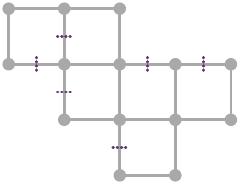}
		\end{minipage}
		\hfil
		\begin{minipage}[c]{0.4\textwidth}
			\centering
			\includegraphics[page=5, scale=0.5]{appendix/hausdorff-hard}
		\end{minipage}\\
		\begin{minipage}[c]{0.4\textwidth}
			\centering
			\includegraphics[page=2, scale=1.5]{appendix/hausdorff-hard}
		\end{minipage}
		\hfil
		\begin{minipage}[c]{0.4\textwidth}
			\centering
			\includegraphics[page=6, scale=0.5]{appendix/hausdorff-hard}
		\end{minipage}
		\caption{Given a grid graph $G$ (top left), we construct a simple polygon $P$ (top right), such that if and only if there exists a Hamiltonian cycle in $G$ (bottom left), there exists a grid polygon $Q$ with both $d_H(\partial P, \partial Q)\leq \frac{1}{2}$ and $d_H(\partial Q, \partial P)\leq \frac{1}{2}$ (bottom right).}
		\label{fig:hausdorff-hard}
	\end{figure}

	Suppose there is a Hamiltonian cycle in $G$ (see Fig.~\ref{fig:hausdorff-hard} (bottom left)). It is straightforward to construct a grid polygon $Q$ with the Hausdorff distance $1/2$ from its boundary to $\partial P$ and vice versa. The edge gadgets corresponding to the edges in $G$ that belong to the Hamiltonian cycle are covered by $\partial Q$ following a ``zigzag'' pattern, and the edge gadgets corresponding to the edges in $G$ that do not belong to the Hamiltonian cycle are covered by two pieces of $\partial Q$ following a U-turn pattern coming from the two adjacent vertices (see Fig.~\ref{fig:hausdorff-hard} (bottom right)).

	Now, suppose there is a grid polygon $Q$ such that $d_H(\partial P, \partial Q)\leq 1/2$ and $d_H(\partial Q, \partial P)\leq 1/2$. $\partial Q$ can only pass through the grid points marked with green dots in the Fig.~\ref{fig:hausdorff-hard-edge} (top) and Fig.~\ref{fig:hausdorff-hard-vertex} (left), otherwise the Hausdorff distance to the input polygon is too large. Thus, every vertex gadget must have two adjacent edge gadgets covered by $\partial Q$ following the ``zigzag'' pattern. The edges of the grid graph $G$ corresponding to these edge gadgets form a Hamiltonian cycle, because the grid polygon $Q$ is simple.

	Therefore, a Hamiltonian cycle in $G$ exists if and only if there exists a grid polygon $Q$ such that both $d_H(\partial P, \partial Q)\leq 1/2$ and $d_H(\partial Q, \partial P)\leq 1/2$.
\end{proof}

\clearpage

\section{Upper and lower bounds on the Fr\'echet distance}
\label{app-frechet}

\begin{lemma}
	\label{lem-frechet-upper-degenerate}
	Let $P$ be a $(\sqrt{2}, \beta)$-narrow polygon.
	Let $C$ be the cycle in the construction of Theorem~\ref{thm-frechet-upper} for $P$ after removing all duplicates.
	If $C$ consists of at most two vertices, a 4-cycle $C'$ exists with $d_F(\partial P, C') \leq (\beta + \sqrt{2})/2$.
\end{lemma}
\begin{proof}
	If $C$ is one vertex $v$, we make it into a simple 4-cycle, picking the direction of extension such that at least one point of $\partial P$ lies in that direction from $v$ as well. 
	We can map the entire extension onto this point, which has distance at most $\sqrt{2} \leq (\beta + \sqrt{2})/2$ by our assumption on $\beta$; mapping $v$ to the entire $\partial P$ then proves the necessary bound.

	If $C$ contains two vertices, $u$ and $v$, we first observe that $\mu(v)$ starts and ends on the boundary of the square of $v$, as this holds initially and is maintained during the removal of duplicates.
	In particular, that means that there is a point on the common boundary between the squares of $u$ and $v$ and $\mu(u)$ and $\mu(v)$ end there. 
	As for the single-vertex case, we extend into a simple 4-cycle in the direction of this point, and map the extension onto it (causing distance at most $\sqrt{2}$). 
	Together with $\mu(u)$ and $\mu(v)$, this acts as a witness for the necessary bound on the Fr\'echet-distance.
\end{proof}

\begin{cloneclaim}{Lemma~\ref{lem:p-obese}}
	The polygon $P$ described above is $(\sqrt{2}, \beta)$-narrow.
\end{cloneclaim}
\begin{proof}
	We must show that $|ab|_{\partial P} \leq \beta$ holds for any two points $a,b \in \partial P$ with $\|a - b\| \leq \sqrt{2}$.
	Note that $a$ and $b$ either lie on the same boundary segment of $\partial P$, on two consecutive segments, or on two consecutive parallel segments of $L$ (i.e., on $p_{i-1} p_{i}$ and $p_{i} p_{i+1}$ for some $1<i<n$): otherwise the Euclidean distance between $a$ and $b$ would be larger than $\sqrt{2}$.
	If the two points lie on the same segment, then the distance between $a$ and $b$ along the polygon boundary $|ab|_{\partial P}=||a - b|| \leq \sqrt{2} < \beta$.

	Now assume that $a$ and $b$ lie on two consecutive segments of the boundary of $P$.
	Denote the common point of the two segments to which $a$ and $b$ belong as $p$.
	By construction, $\angle apb \in [\varphi, \pi)$ where $\varphi = \arccos{(1-\frac{4}{\beta^2})}$. 
	As the Euclidean distance is convex, symmetry implies that the value of $\|a-p\|+\|p-b\|$ is maximized when $\|a-p\|=\|p-b\|$ and $\|a - b\| = \sqrt{2}$ (see also Lemma~\ref{lem-obesity}, Appendix~\ref{app-measure-narrow}). 
	Therefore, using the law of cosines, we can conclude that $|ab|_{\partial P}=\|a-p\|+\|p-b\| \leq \frac{2}{\sqrt{1-\cos{\varphi}}}=\beta$.

	Finally, assume that $a$ and $b$ lie on two consecutive parallel segments $p_{i-1} p_{i}$ and $p_{i} p_{i+1}$ for some $1<i<n$.
	W.l.o.g., let $a$ lie on segment $p_{i-1} p_{i}$ and let $i$ be even.
	Then observe that the $x$-coordinate of $b$ is not greater than the $x$-coordinate of $a$, otherwise the Euclidean distance between $a$ and $b$ would be larger than $\sqrt{2}$.
	Therefore, the distance between $a$ and $b$ along $\partial P$ is not greater than the length of two boundary edges forming a spike, i.e., $|ab|_{\partial P} \leq \beta$.
\end{proof}

\clearpage

\section{Measuring narrowness}
\label{app-measure-narrow}

We bounded the Fr\'echet distance needed for a grid polygon, depending on the $\sqrt{2}$-narrowness of the input polygon.
An important question to ask is then how narrow realistic inputs really are.
We briefly describe how to compute the $\alpha$-narrowness of a polygon $P$ in quadratic time.
To this end, we must find the pair of points $p,q$ on $\partial P$ with $\| p - q \| \leq \alpha$ such that $|pq|_{\partial P}$ is maximized; we call such a pair an \emph{($\alpha$-narrow) witness}.
If we consider $|pq|_{\partial P}$ for some fixed $p$, then we see that this function has a single maximum with value $|P|/2$.
In particular, this implies the observation below, and inspires the two subsequent lemmas.
From these statements, the algorithm readily follows.

\begin{obs}
	\label{obs-obesity}
	An $\alpha$-narrow witness $(p,q)$ of $P$ satisfies $\|p-q\| = \alpha$ or $|pq|_{\partial P} = |P|/2$.
\end{obs}

\begin{lemma}
	\label{lem-obesity-maxbeta}
	Let $e = (t,u)$ and $e' = (v,w)$ be two edges of $P$, and assume $|uv|_{\partial P}$ is known.
	We can determine in constant time whether there is a $p \in e$ and $q \in e'$ such that $\|p-q\| \leq \alpha$ and $|pq|_{\partial P} = |P|/2$.
\end{lemma}
\begin{proof}
	We express $p \in e$ as a function of $\lambda_p \in [0,1]$: $p = t (1 - \lambda_p) + u \lambda_p$.
	Analogously, for $\lambda_q \in [0,1]$ we have $q =  v (1 - \lambda_q) + w \lambda_q$.

	The condition $|pq|_{\partial P} = |P|/2$ can now be written as $\|p - u\| + |uv|_{\partial P} + \| v - q \| = |P|/2$, which, using the above expressions leads us to $(1 - \lambda_p) \| t - u \| + |uv|_{\partial P} + \lambda_q \| v - w \| = |P|/2$.
	We may rewrite this to $\lambda_q = (|P|/2 - (1 - \lambda_p) \| t - u \| - |uv|_{\partial P}) /  \| v - w \| = (|P|/2 - \| t - u \| - |uv|_{\partial P}) / \| v - w \| + \lambda_p \| t - u \| / \| v - w \|$.
	To make this expression simpler, we introduce $R = \| t - u \| / \| v - w \|$ and $C = (|P|/2 - \| t - u \| - |uv|_{\partial P}) / \| v - w \|$.
	Then, we find $\lambda_q = C + R \lambda_p$.

	As above, we can write $\|p-q\| \leq \alpha$ to $\| t (1 - \lambda_p) + u \lambda_p - (v (1 - \lambda_q) + w \lambda_q)  \| \leq \alpha$. Substituting the expression for $\lambda_q$, we obtain $\| t (1 - \lambda_p) + u \lambda_p - (v (1 - C - R \lambda_p) + w (C + R \lambda_p)) \| = \| (t - v + v C - w C) + (u + v R - t - w R) \lambda_p \| \leq \alpha$.
	Introducing $c = t - v + v C - w C$ and $r = u + v R - t - w R$, we get $\| c + r \lambda_p \| \leq \alpha$.
	Writing out the Euclidean distance, this becomes $(c_x + r_x \lambda_p)^2 + (c_y + r_y \lambda_p)^2 \leq \alpha^2$ which is a quadratic equation: $(r_x^2 + r_y^2) \lambda_p^2 + 2 (c_x r_x + c_y r_y) \lambda_p + (c_x^2 + c_y^2 - \alpha^2) \leq 0$. 

	Solving this quadratic equation, yields us an interval describing the points on the line spanned by $e$ that satisfy the two conditions.
	If this interval does not overlap $[0,1]$, then there are no solutions. 
	Otherwise, let $[\lambda_{p,1}, \lambda_{p,2}]$ denote the intersection of the computed interval with $[0,1]$.
	Via $\lambda_q = C + R \lambda_p$, this projects an interval on the line spanned by $e'$ containing the points corresponding to the points described by $[\lambda_{p,1}, \lambda_{p,2}]$.
	There is now a solution satisfying both equations if this projected interval intersects $[0,1]$.
\end{proof}

\begin{lemma}
	\label{lem-obesity}
	Let $P$ be a $(\alpha, \beta)$-narrow polygon with $\beta < |P|/2$.
	$P$ has an $\alpha$-narrow witness $(p,q)$ such that $\|p-q\| = \alpha$ and at least one of the following holds:
	(1) $p$ is a vertex of $P$; (2) $p$ and $q$ lie on nonparallel edges of $P$ and are equidistant from the intersection of the lines spanned by these edges.
\end{lemma}
%\begin{proof}[Proof sketch]
%Assume we have a $p$ and $q$ that do not adhere to (1) or (2).
%We can always shift $p$ and $q$ in such a way that their distance does not increase and their perimeter distance does not decrease.
%Either we find a vertex or $p$ and $q$ become equidistant.
%\end{proof} 
\begin{proof}
	Observation~\ref{obs-obesity} implies that a witness satisfies $\|p-q\| = \alpha$.
	Consider such a witness $(p,q)$ that does not adhere to either (1) or (2). 
	Without reducing $|pq|_{\partial P}$ we transform the witness into one that does (refer to Fig.~\ref{fig:obesity-sliding}).
	Let $\ell_p$ and $\ell_q$ denote the lines spanned by the edges containing $p$ and $q$ respectively; let $c$ denote their intersection.
	Both $p$ and $q$ have a direction along their respective lines for which $|pq|_{\partial P}$ increases if we keep the other fixed.
	Note that if the direction would change at any point, we would have a maximum value and thus a witness showing that $P$ is $(\alpha, |P|/2)$-narrow, contradicting the assumption.

	\begin{figure}[h]
	\centering
	\includegraphics{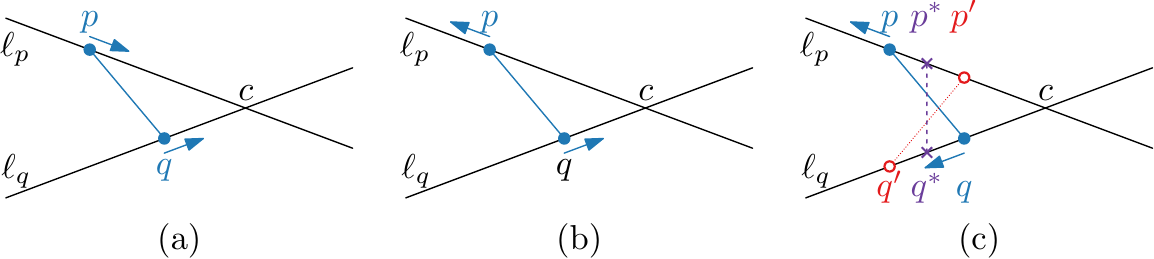}
	\caption{(a-b) Moving $p$ and $q$ towards $c$ does not decrease $|pq|_{\partial P}$ nor increase $\|p-q\|$; we can do so until we hit a vertex. (c) Inverse case $(p',q')$ has distance at most $\alpha$; the midpoint $(p^*,q^*)$ thus defines an equidistant witness.}
	\label{fig:obesity-sliding}
	\end{figure}

	If at least one direction points to $c$ or if the lines are parallel, we simultaneously slide $p$ and $q$ (at the same speed) into that direction without increasing $\|p-q\|$ or decreasing $|pq|_{\partial P}$. At some point, a vertex of $P$ must be found and thus we have a witness that adheres to (1).

	If the lines are not parallel and both directions point away from $c$, we argue as follows. 
	Consider the inverse case $(p',q')$, where we place $p'$ on $\ell_p$ such that $\|p' - c\| = \|q - c\|$ and $q'$ on $\ell_q$ such that $\|q' - c\| = \|p - c\|$. 
	Since $\|p' - p\| = \|q' - q\|$, $\|p' - q'\| \leq \alpha$; note that $p'$ and $q'$ need not lie on the edges containing $p$ and $q$. 
	Convexity of the Euclidean distance implies that, as we move $p$ to $p'$ and $q$ to $q'$ simultaneously at the same speed the points remain within distance $\alpha$ and $|pq|_{\partial P}$ does not change. 
	In particular, halfway, we find a witness $(p^*,q^*)$ that are equidistant to $c$.
	If either $p$ or $q$ hits before reaching $(p^*,q^*)$, we are done and found a witness adhering to (1).
	Otherwise, if $\|p^* - q^*\| < \alpha$, we may move them again simultaneously away from $c$ (thus increasing $|p^*q^*|_{\partial P}$), until we either hit a vertex---adhering to (1)---or until the distance is exactly $\alpha$---adhering to (2).
\end{proof}

\clearpage
\section{Experiments}
\label{app:experiments}

In this appendix we present more details for our experiments.

\subsection{Experimental inputs}
\label{sec:app-inputs}

We used a total of 34 simple polygons to fuel our experiments.
They can be categorized as territorial outlines, building footprints and animal contours.
The former two target schematization purposes and the latter the construction of nonograms.

\subparagraph{Territorial outlines}
We used 14 territorial outlines (e.g. countries, provinces, islands), as illustrated in Fig.~\ref{fig:maps}.
We tried to vary the geographic extent and nature of the outlines.
In particular, we have two continents (Africa and Antarctica), three provinces (Noord Brabant, Limburg and Languedoc-Roussillon), and the largest contiguous part of 9 countries (Australia, Brazil, China, France, Greece, Italy, Switzerland, Great Britain, Vietnam).
With this choice, the outlines contain a mix of regions defined by shorelines, territorial borders or a combination of these. Moreover, the geographic extent varies from small to large-scale regions.

\subparagraph{Building footprints}
Our data set contains 11 buildings (Fig.~\ref{fig:buildings}).
It contains a mix of complex buildings (Bld 1, terminal of Logan Airport in Boston (MA); Bld 5, high school in Chicago; Bld 6, stadium in Chicago; Bld 11, university building at TU Eindhoven), castles (Bld 2--4) and low complexity buildings (Bld 7--10).

\subparagraph{Animal contours}
Finally, our experiments also consider 9 animal contours (Fig.~\ref{fig:animals}).
In particular, we chose these inputs to achieve an expected variety of narrowness. 
For example, the spider intuitively is more narrow than the turtle.

%\afterpage{\clearpage}

\begin{figure}[p]
\vspace{-4\baselineskip}
%ARXIV: moving this up a lot
	\centering
	\begin{tabular}{ccc}
		\includegraphics[width=0.21\linewidth]{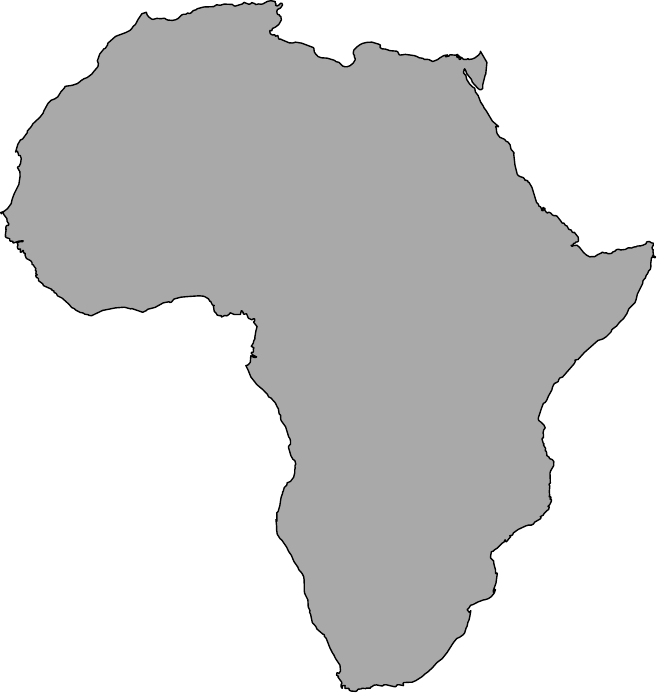}
		&
		\includegraphics[width=0.25\linewidth]{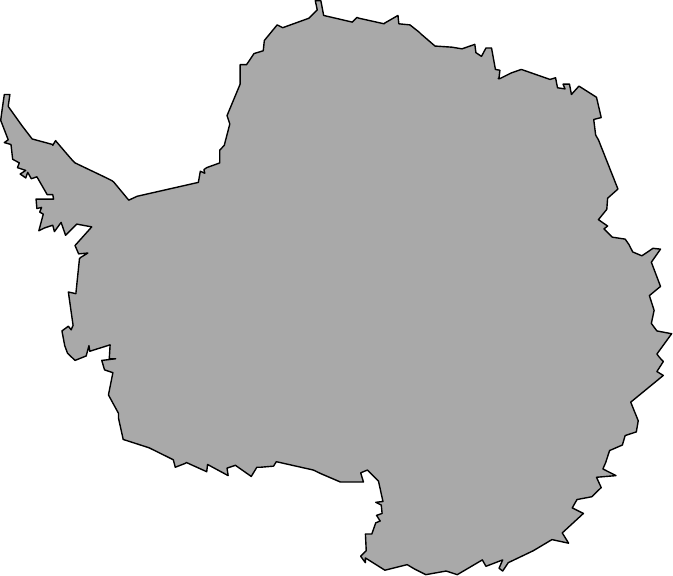}
		&
		\includegraphics[width=0.25\linewidth]{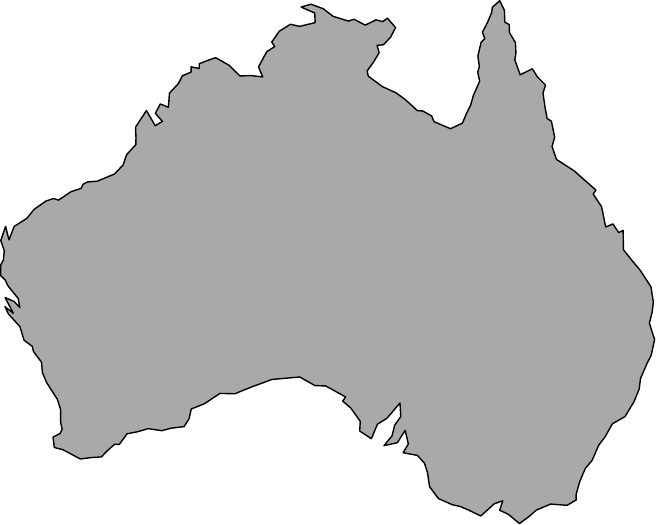}
		\\
		Africa & Antarctica & Australia \\

		\includegraphics[width=0.25\linewidth]{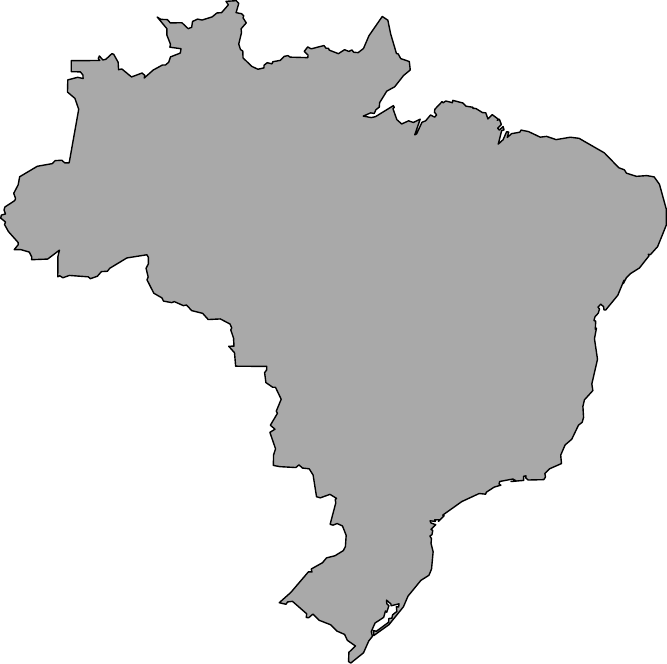}
		&
		\includegraphics[width=0.25\linewidth]{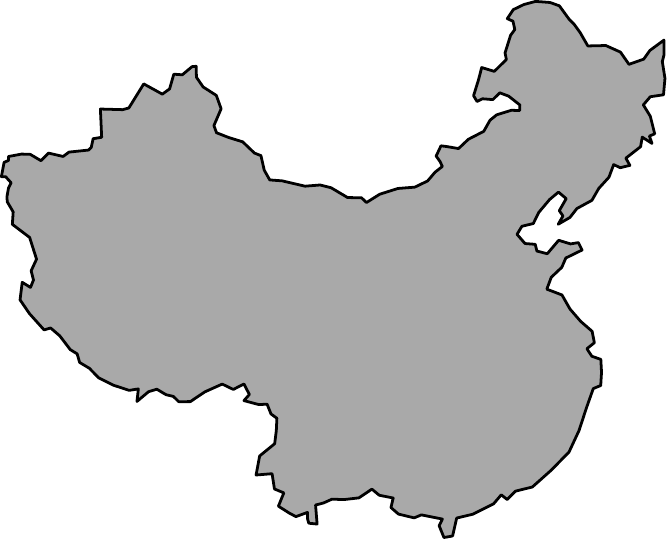}
		&
		\includegraphics[width=0.25\linewidth]{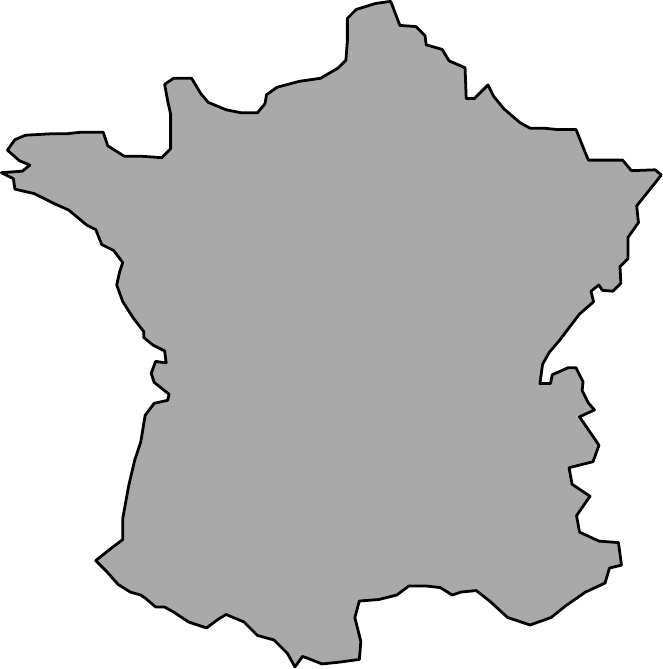}
		\\
		Brazil & China & France \\

		\includegraphics[width=0.21\linewidth]{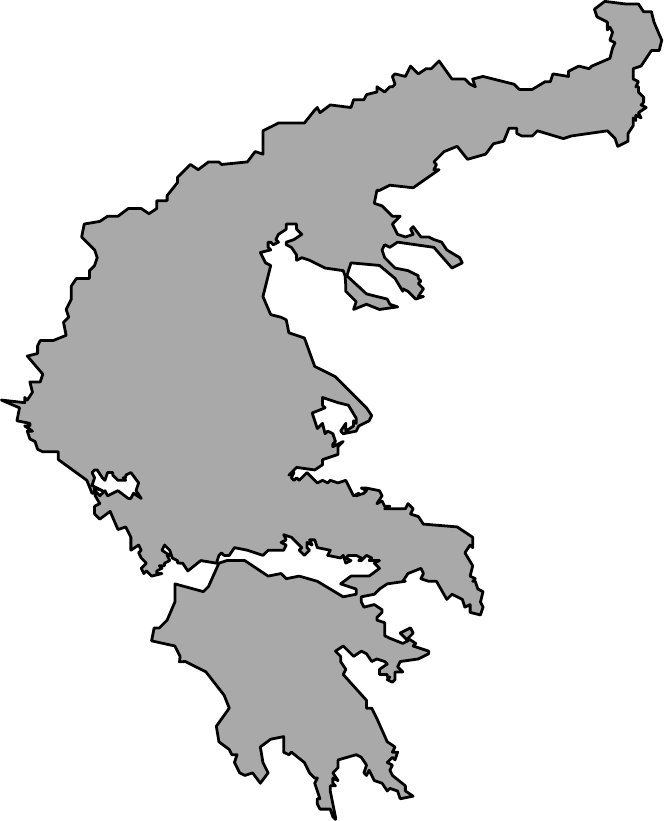}
		&
		\includegraphics[width=0.25\linewidth]{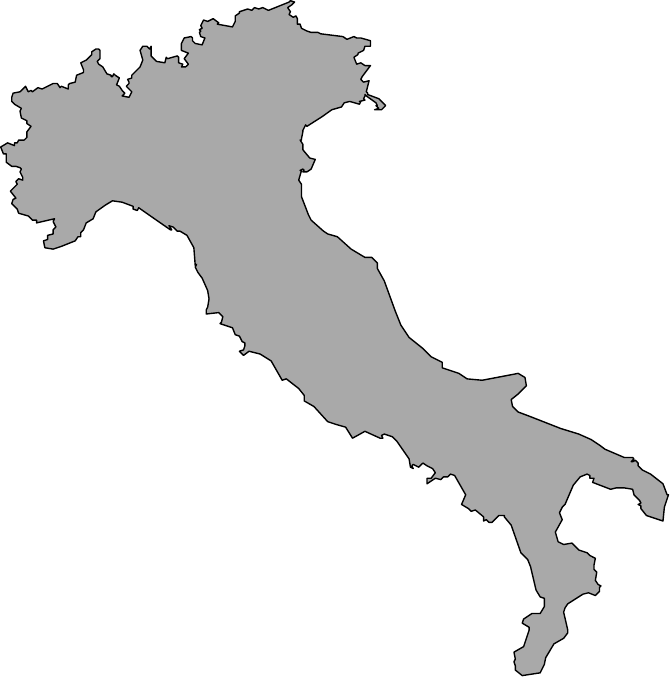}
		&
		\includegraphics[width=0.23\linewidth]{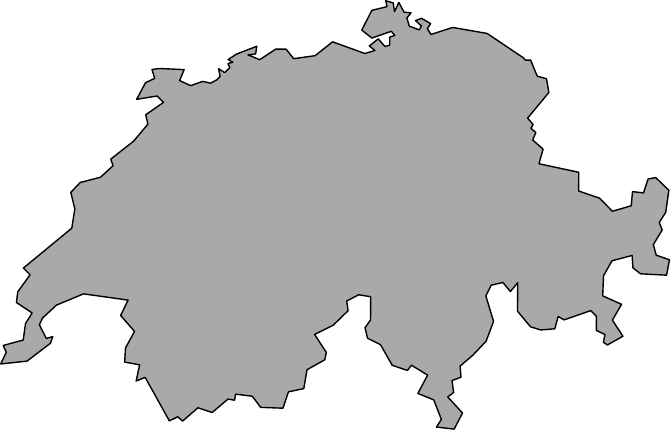}
		\\
		Greece & Italy & Switzerland \\

		\includegraphics[width=0.2\linewidth]{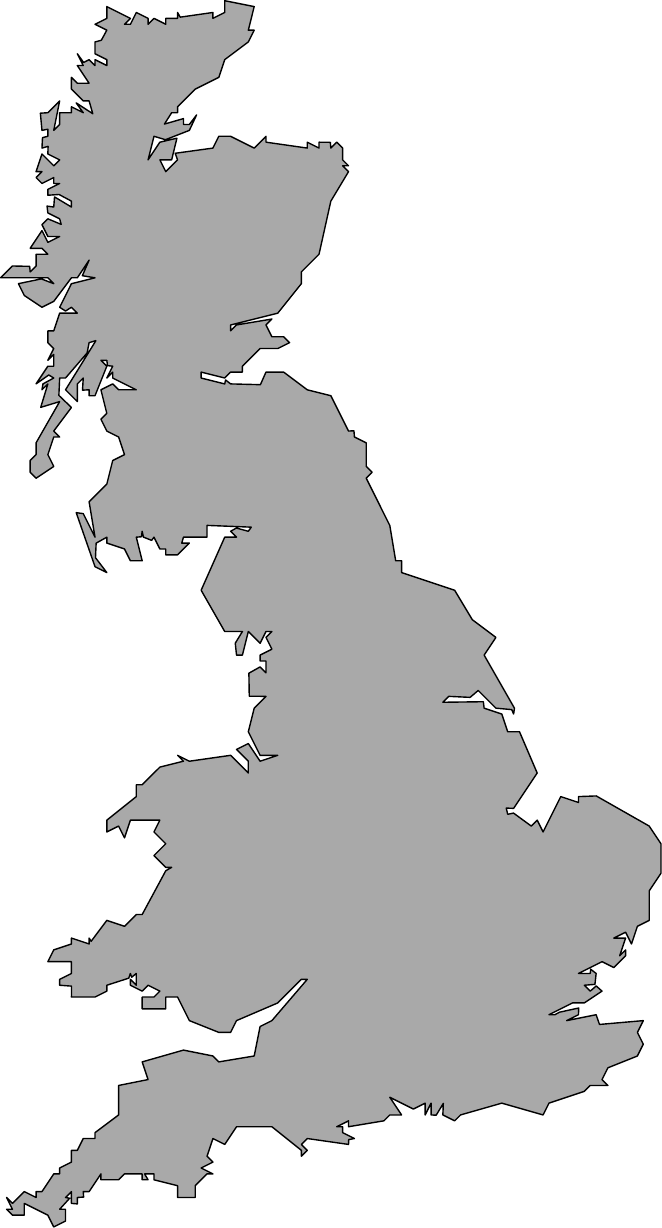}
		&
		\includegraphics[width=0.18\linewidth]{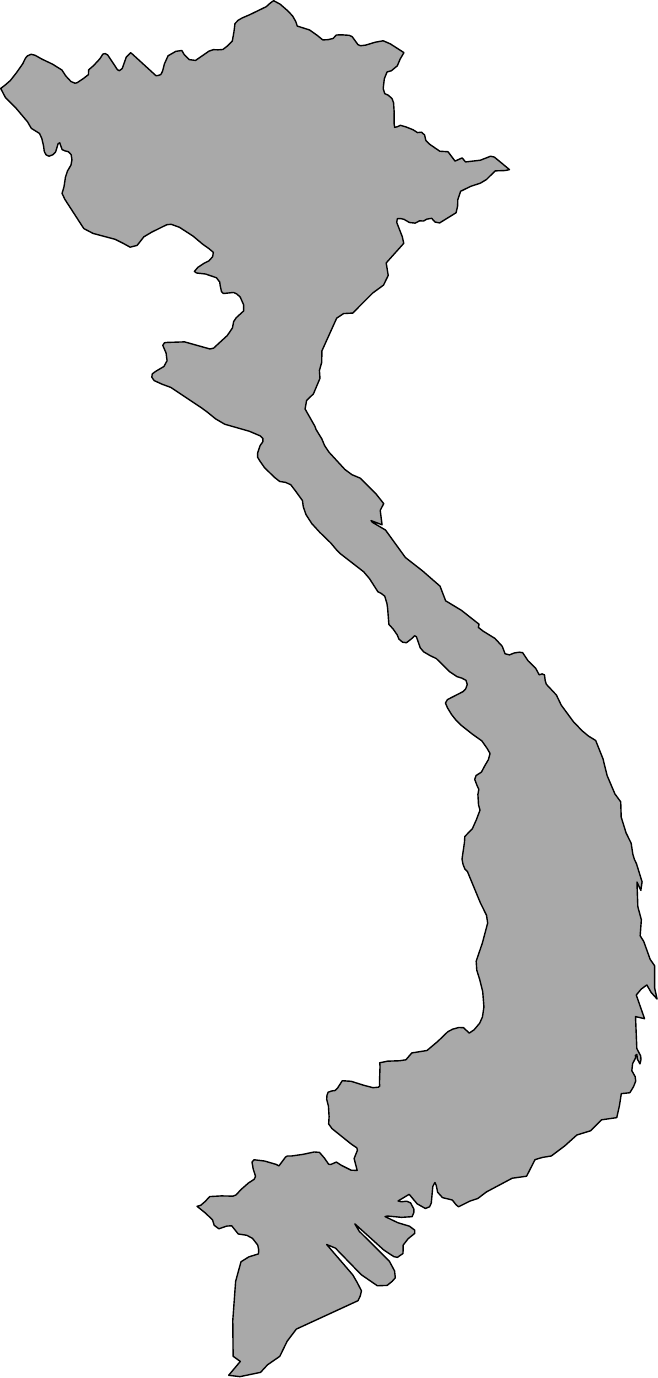}
		&
		\includegraphics[width=0.15\linewidth]{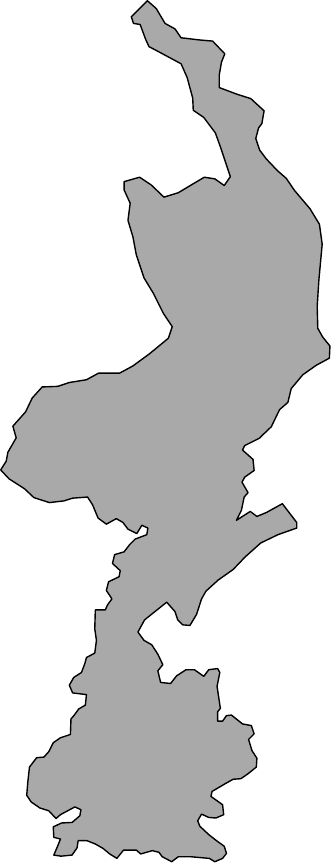}
		\\
		Great Britain & Vietnam & Limburg (NL) \\

		\includegraphics[width=0.25\linewidth]{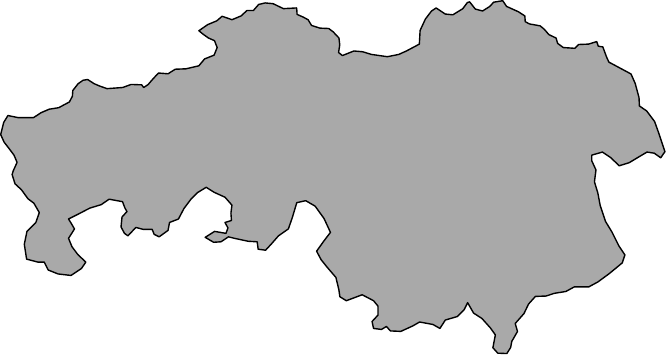}
		&
		\includegraphics[width=0.21\linewidth]{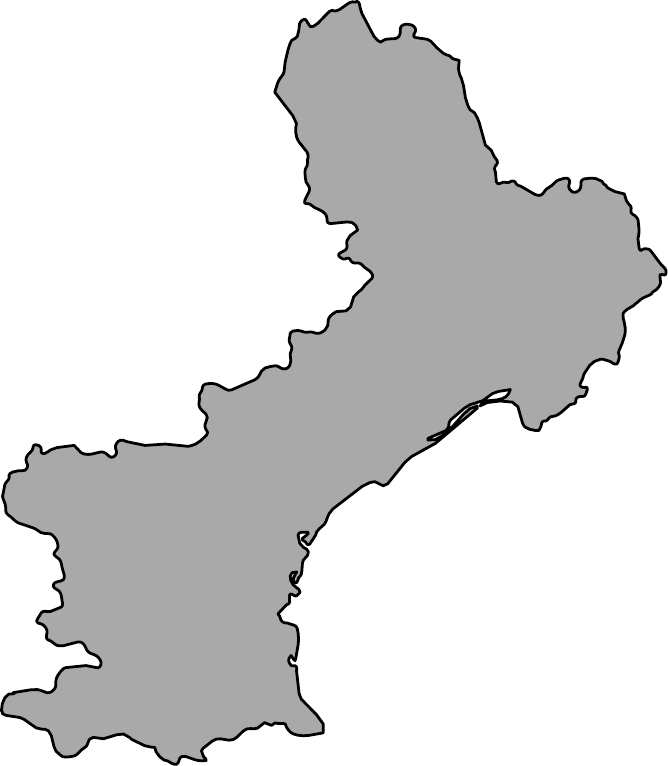}
		\\
		Noord Brabant (NL) & Languedoc-Roussillon (FR) \\
	\end{tabular}
	\caption{The 14 territorial outlines used in our experiments.}
	\label{fig:maps}
\end{figure}

\begin{figure}[p]
	\centering
	\begin{tabular}{ccc}
		\includegraphics[width=0.21\linewidth]{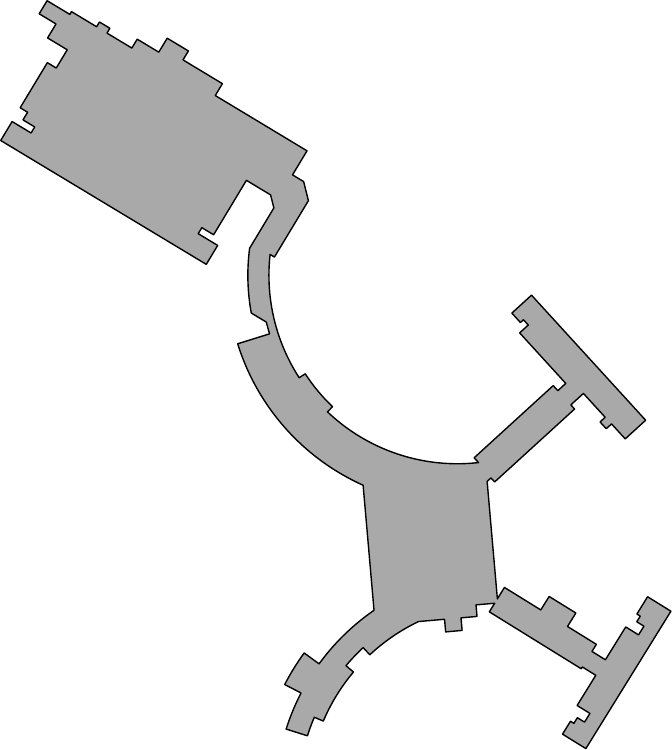}
		&
		\includegraphics[width=0.25\linewidth]{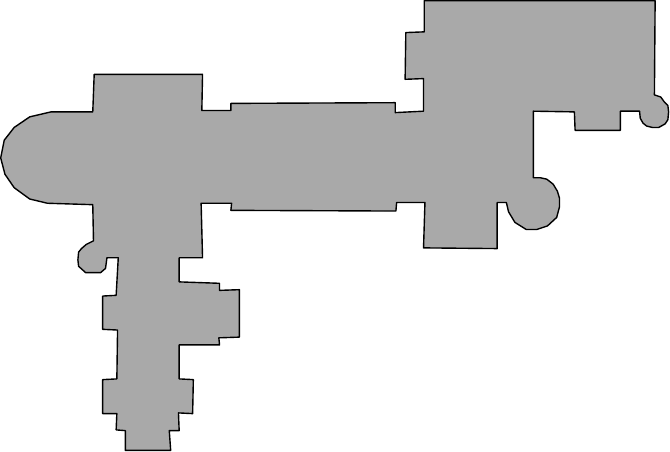}
		&
		\includegraphics[width=0.25\linewidth]{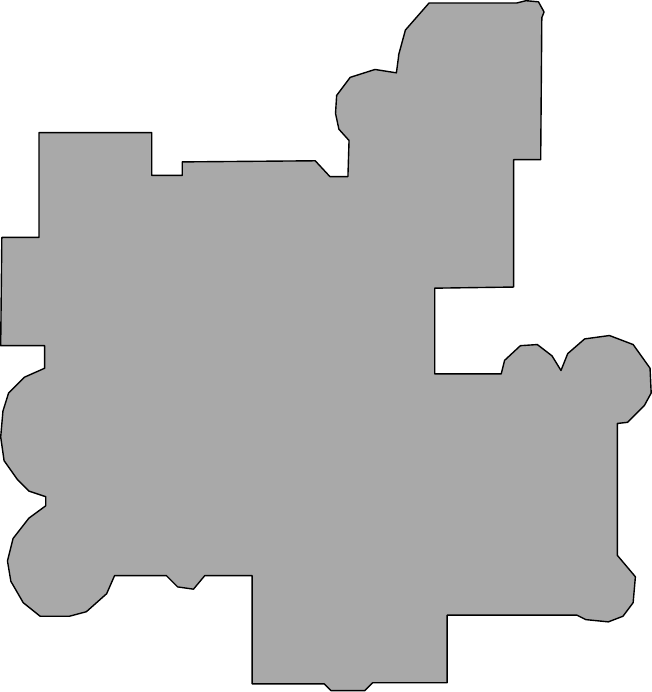}
		\\
		Bld 1 & Bld 2 & Bld 3 \\

		\includegraphics[width=0.21\linewidth]{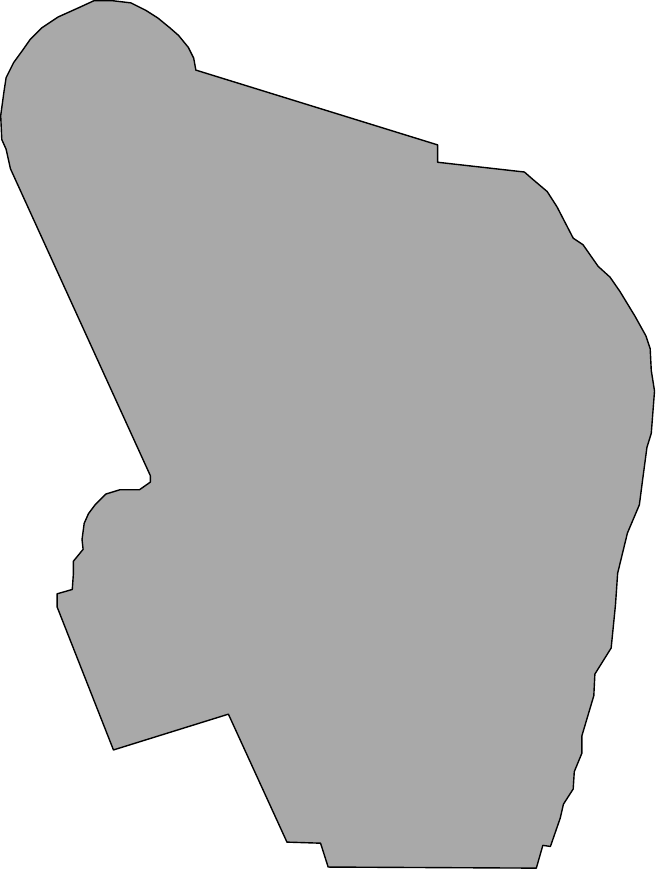}
		&
		\includegraphics[width=0.25\linewidth]{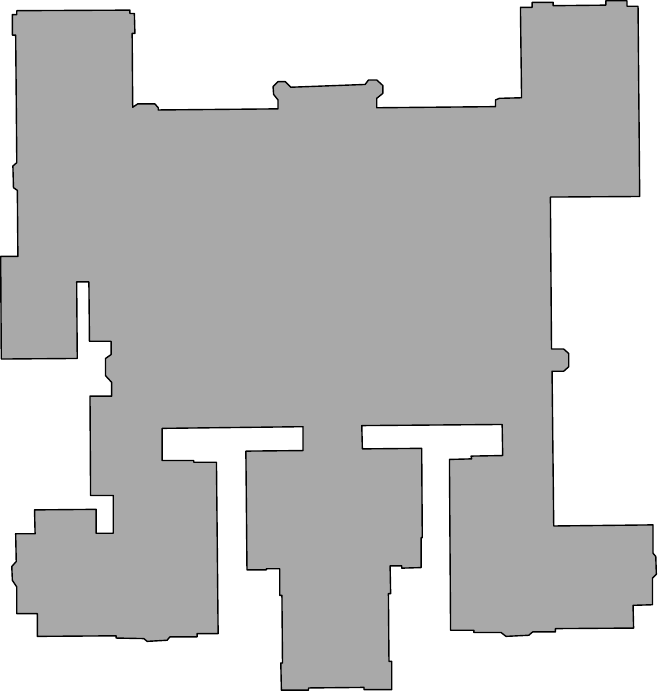}
		&
		\includegraphics[width=0.23\linewidth]{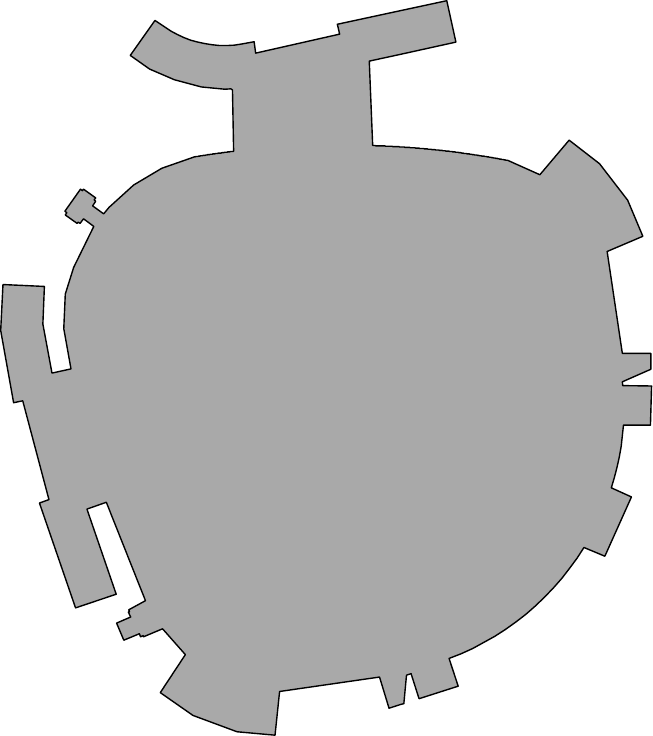}
		\\
		Bld 4 & Bld 5 & Bld 6 \\

		\includegraphics[width=0.21\linewidth]{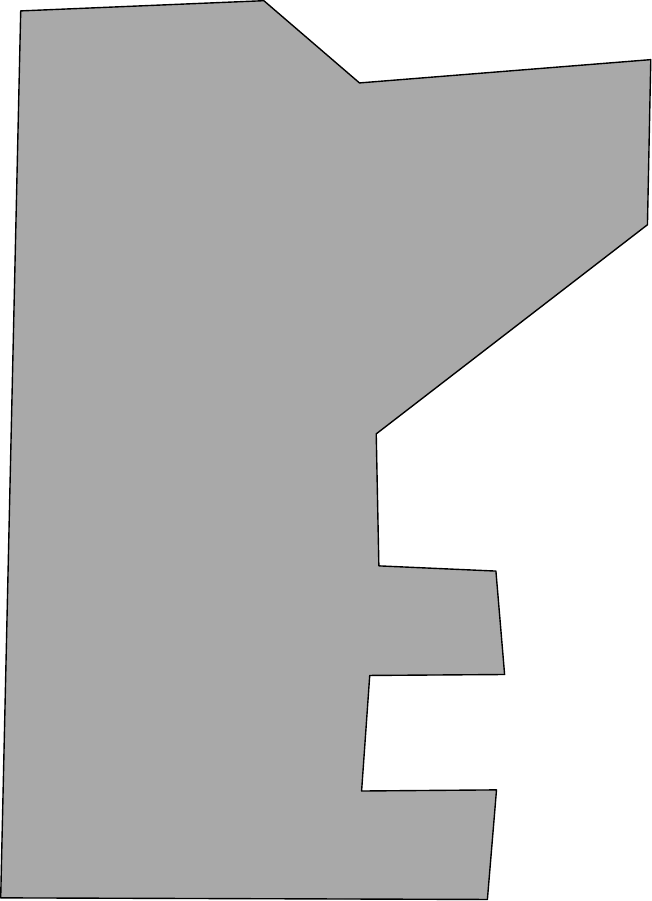}
		&
		\includegraphics[width=0.25\linewidth]{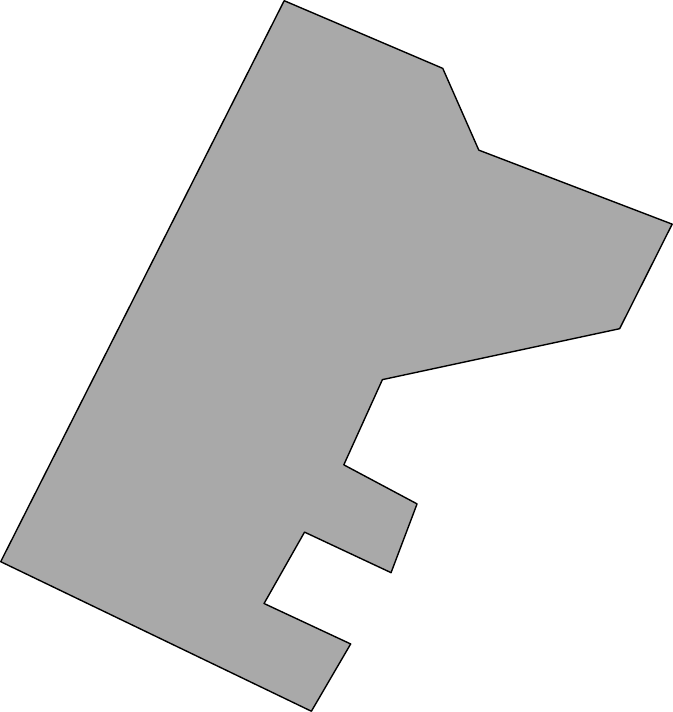}
		&
		\includegraphics[width=0.25\linewidth]{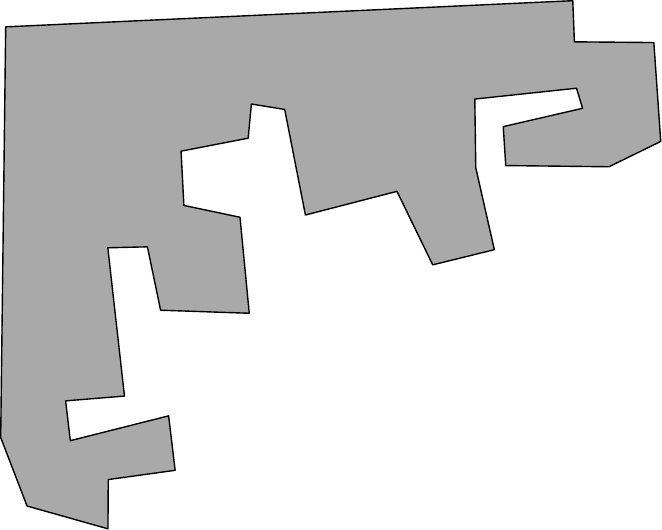}
		\\
		Bld 7 & Bld 8 & Bld 9 \\

		\includegraphics[width=0.21\linewidth]{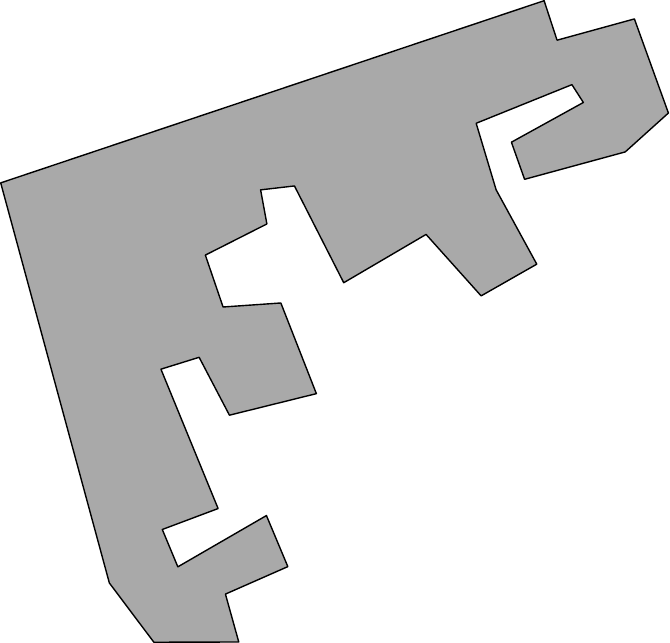}
		&
		\makebox[4em][l]{\includegraphics[width=0.45\linewidth]{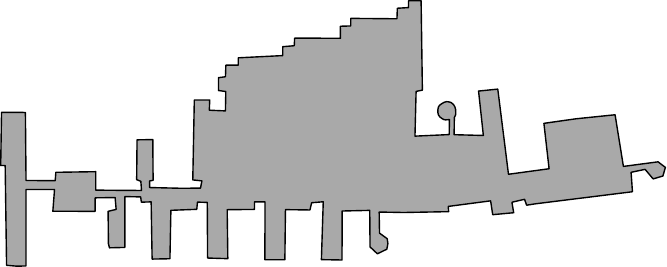}}
		\\
		Bld 10 & \multicolumn{2}{c}{Bld 11}  \\
		\end{tabular}
	\caption{The 11 building footprints used in our experiments.}
	\label{fig:buildings}
\end{figure}

\begin{figure}[p]
	\centering
	\begin{tabular}{ccc}
		\includegraphics[width=0.25\linewidth]{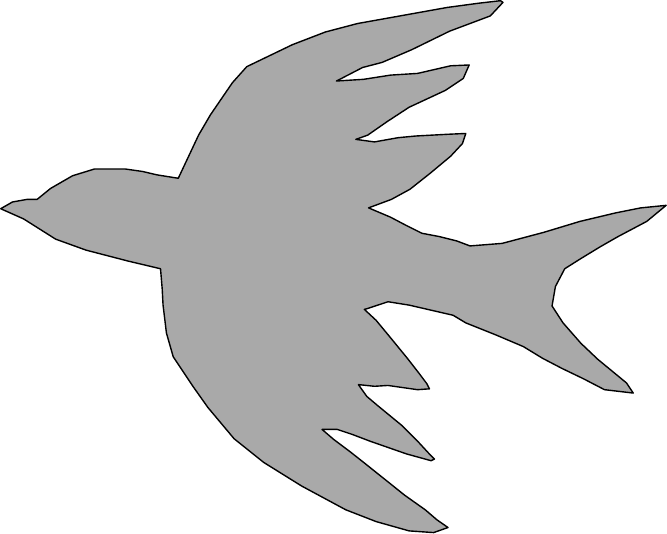}
		&
		\includegraphics[width=0.25\linewidth]{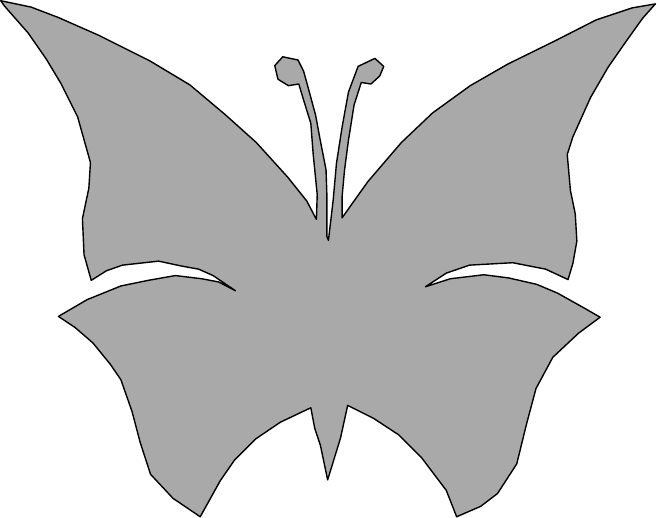}
		&
		\includegraphics[width=0.25\linewidth]{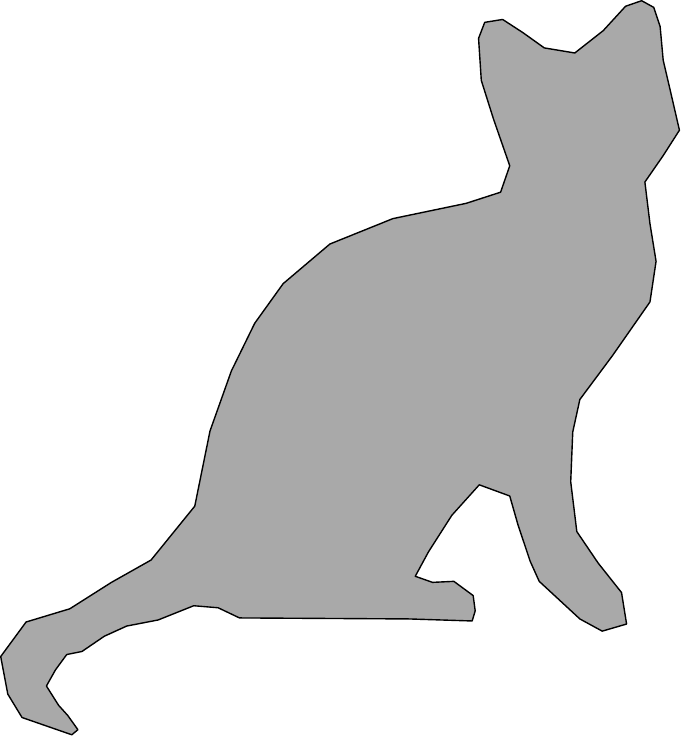}
		\\
		bird & butterfly & cat \\

		\includegraphics[width=0.25\linewidth]{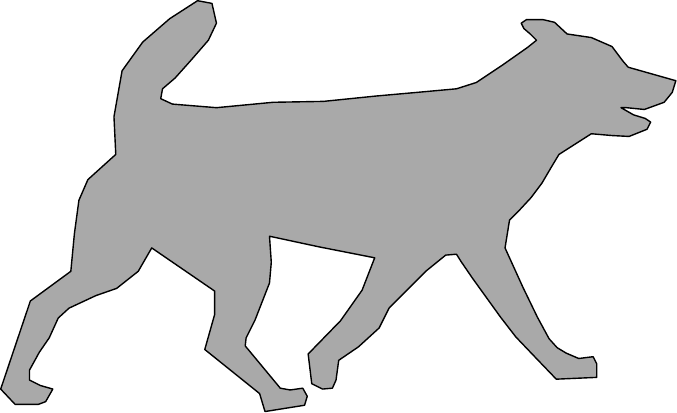}
		&
		\includegraphics[width=0.25\linewidth]{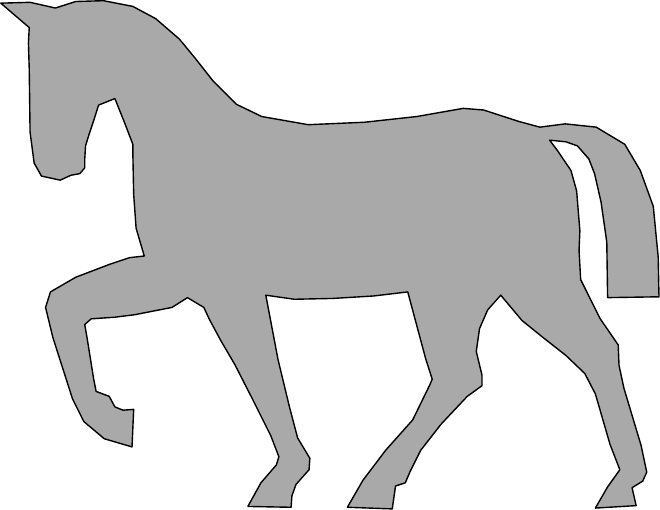}
		&
		\includegraphics[width=0.23\linewidth]{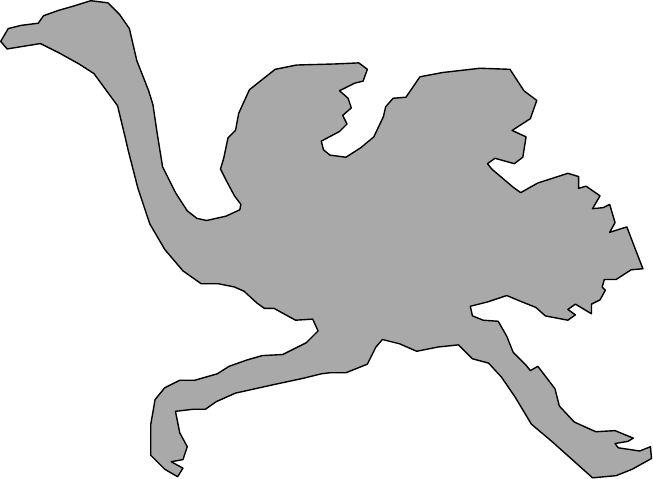}
		\\
		dog & horse & ostrich \\

		\includegraphics[width=0.25\linewidth]{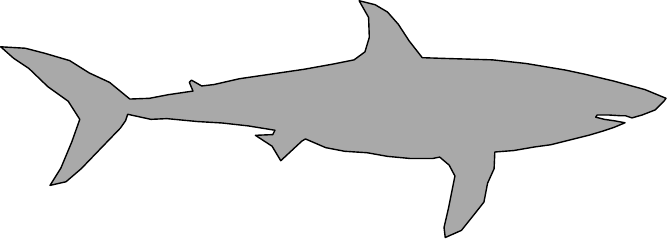}
		&
		\includegraphics[width=0.25\linewidth]{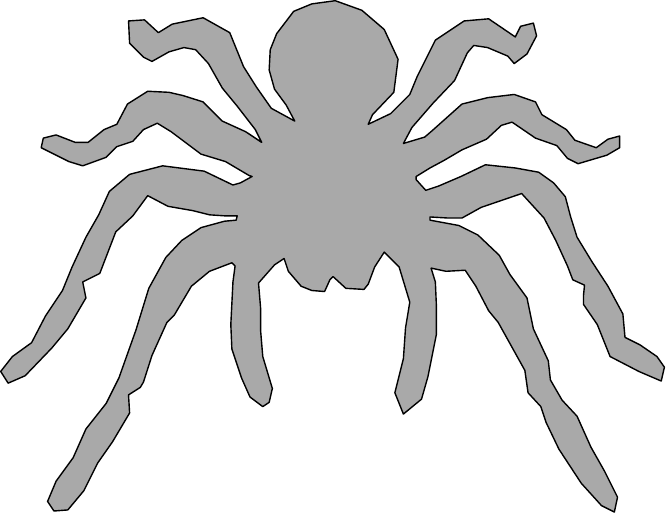}
		&
		\includegraphics[width=0.25\linewidth]{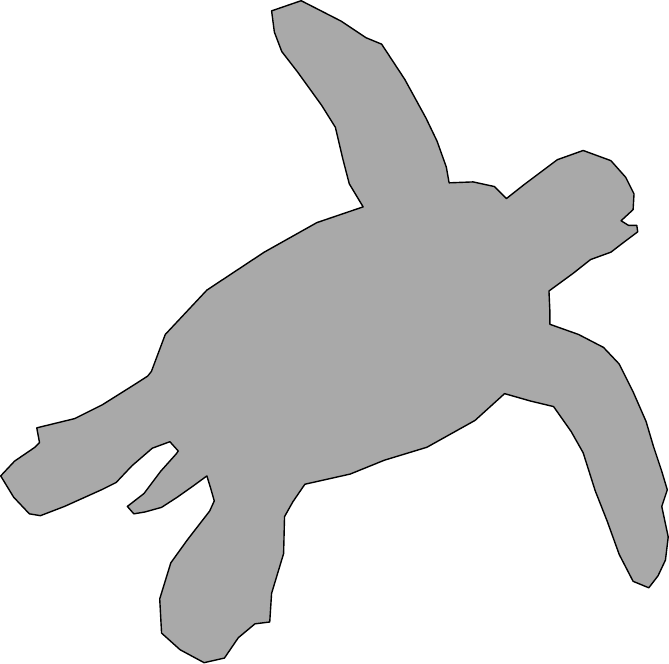}
		\\
		shark & spider & turtle \\
	\end{tabular}
	\caption{The 9 animal contours used in our experiments.}
	\label{fig:animals}
\end{figure}

\subsection{Fr\'echet analysis}
\label{app-exp-frechet}

\subparagraph{Procedure}
We use the 34 polygon described in Appendix~\ref{sec:app-inputs}.
As we may expect the grid resolution to significantly affect results, we used 20 different resolutions.
In particular, we use resolutions varying from $10\,000$ to $25$, using $(100/s)^2$ with scale $s \in \{ 1, \ldots, 20 \}$.

For each resolution-polygon combination (case), we measure its $\sqrt{2}$-narrowness (see Appendix~\ref{app-measure-narrow} for the algorithm) and derive the predicted upper bound.
Then, we run the Fr\'echet algorithm, using the 25 possible offsets in $\{0, 0.2, 0.4, 0.6, 0.8 \}^2$, and measure the precise Fr\'echet distance between input and output.
We keep track of three summary statistics for each case: the minimum (best), average (``expected'') and maximum (worst)measured Fr\'echet distance.

\subparagraph{Effect of placement}
We define the placement effect as the difference between between the maximal and minimal Fr\'echet distance.
We consider the effect significant if it is at least $2$.
These differences are listed in Table~\ref{tbl-frechet-placement}.
Almost $30\%$ of cases exhibit such a significant effect, with the animal contours being particularly affected ($35\%$ significant).
Eight cases even have a difference over $10$, with the maximal difference slightly under $55$.

If we look at the difference with the average Fr\'echet distance, we find a much less pronounced effect (see Table~\ref{tbl-frechet-placement-avg}).
Only $6.5\%$ of cases have a significant difference, most of these buildings ($9.1\%$) and the least with territorial outlines ($3.6\%$).

We conclude that a bad placement is quite likely to have a drastic effect on the algorithm's performance, but the potential gain when optimizing the placement instead of picking a random one is limited.

\subparagraph{Upper bound quality}
Let us now turn to the relation between the predicted upper bound and the actual Fr\'echet distance achieved by the algorithm.
We define the performance as the measured Fr\'echet distance expressed as a percentage of the upper bound.
We consider the algorithm's performance significantly better than the upper bound, if its performance is less than a very conservative $40\%$.

Using the best placement (Table~\ref{tbl-frechet-min-ub}), over $95\%$ of cases perform significantly better than the upper bound predicts.
The worst case of the best placements is even only only $63$ percent.

Averaging performance over placement (Table~\ref{tbl-frechet-avg-ub}), we still find such a majority (over $81\%$) to have a significantly better performance.
Interestingly, this drop is mostly due to the animal contours, of which only $63\%$ now perform significantly better.
The average performance over all cases is around 30 percent.

Only when we look at the worst performing offset for each case (Table~\ref{tbl-frechet-max-ub}), we see that a large number of cases start to fail to drop below the 40-percent significant threshold.
The average performance of these cases is around $44$ percent. 
We also note that we never achieve the actual upper bound. 
Though we can likely construct a polygon that may come close to the upper bound, this suggests that still this might not be realistic for the types of inputs considered here.

Thus, although we have a provable upper bound, we may typically expect our simple algorithm to perform significantly better than the upper bound. 
This holds even without any postprocessing to further optimize the result and when taking a random offset.

\subparagraph{Effect of resolution}
Resolution does not exhibit a clear pattern in the two analyses above. 
For placement, we may see that effects may occur not all (e.g. Australia), at high resolution (e.g. Switzerland), at noncontiguous resolution (e.g. Bld 5) or at most resolution (e.g. Bld 11).
Similar patterns can be seen for the performance with respect to the upper bound.

Nonetheless, resolution likely plays an important role in these results, but not as straightforward as either low or high resolution being more problematic.
Instead, it is likely the most problematic resolutions are those at which the $\sqrt{2}$-narrowness of the polygon jumps as a new pair of edges comes within distance $\sqrt{2}$ of each other. 
However, an in-depth investigation of this is beyond the scope of this paper.

\begin{sidewaystable}
	\setlength{\tabcolsep}{2.5pt}
	\caption{Difference between maximum and minimum Fr\'echet distance over the 25 runs for all cases. Significant differences (at least 2) are marked in a bold font.}
	\label{tbl-frechet-placement}
	\begin{tabular}{l|rrrrrrrrrrrrrrrrrrrr}
		Scale & 1 & 2 & 3 & 4 & 5 & 6 & 7 & 8 & 9 & 10 & 11 & 12 & 13 & 14 & 15 & 16 & 17 & 18 & 19 & 20 \\

		\hline
		Africa & 0.68 & 0.85 & 0.62 & 0.55 & 0.46 & 0.41 & 0.51 & 0.30 & 0.23 & 0.23 & 0.25 & 0.32 & 0.29 & 0.42 & 0.30 & 0.50 & 0.56 & 0.45 & 0.61 & 0.30 \\
		Antarctica & 1.21 & 1.15 & 1.87 & 1.87 & \textbf{2.28} & \textbf{3.41} & \textbf{2.85} & \textbf{2.45} & \textbf{2.14} & \textbf{2.00} & 1.81 & 1.62 & 1.43 & 1.27 & 1.22 & 1.22 & 1.19 & 1.05 & 1.03 & 0.90 \\
		Australia & 1.68 & 1.12 & 0.99 & 0.94 & 0.99 & 0.83 & 0.93 & 1.33 & 1.10 & 1.09 & 0.89 & 0.84 & 0.88 & 0.85 & 0.94 & 0.92 & 0.94 & 0.90 & 0.86 & 0.85 \\
		Brazil & 0.40 & 0.55 & 0.95 & 0.78 & 1.34 & 0.95 & 0.89 & 0.97 & 0.75 & 0.86 & 0.80 & 0.68 & 0.61 & 1.23 & 0.54 & 1.26 & 1.02 & 1.03 & 1.06 & 0.80 \\
		China & 0.74 & 1.04 & 1.08 & 1.04 & 1.36 & 1.27 & 1.10 & 0.87 & 0.93 & 0.69 & 0.50 & \textbf{2.06} & \textbf{2.16} & \textbf{2.14} & 1.76 & 1.92 & 1.70 & 1.56 & 1.46 & 1.36 \\
		France & 1.30 & 0.77 & 0.77 & 0.59 & 0.56 & 0.40 & 0.41 & 0.42 & 0.78 & 0.77 & 1.00 & 0.98 & 0.81 & 0.69 & 0.61 & 0.63 & 0.65 & 0.79 & 0.62 & 0.68 \\
		Greece & 1.21 & 1.41 & \textbf{2.20} & 1.60 & 1.61 & 1.31 & \textbf{2.81} & \textbf{2.37} & \textbf{2.08} & \textbf{3.50} & \textbf{2.56} & \textbf{3.26} & \textbf{3.12} & \textbf{2.79} & \textbf{2.92} & \textbf{2.96} & \textbf{2.54} & \textbf{2.49} & \textbf{2.65} & \textbf{2.03} \\
		Italy & 1.45 & 0.91 & 1.76 & \textbf{2.24} & \textbf{2.83} & \textbf{2.52} & \textbf{2.71} & \textbf{2.50} & \textbf{2.30} & \textbf{2.15} & \textbf{3.26} & \textbf{3.16} & \textbf{3.14} & \textbf{3.02} & \textbf{3.05} & \textbf{3.67} & \textbf{3.05} & \textbf{2.19} & \textbf{2.57} & \textbf{2.09} \\
		Switzerland & \textbf{3.02} & \textbf{4.47} & \textbf{2.71} & \textbf{2.23} & 1.52 & 1.61 & 1.38 & 1.77 & 1.34 & 1.21 & 1.25 & 1.16 & 1.10 & 1.11 & 0.92 & 0.72 & 0.71 & 0.74 & 0.61 & 0.48 \\
		Great Britain & \textbf{3.40} & \textbf{2.18} & \textbf{2.23} & 1.72 & 1.75 & 1.77 & 1.58 & 1.34 & 1.41 & 1.60 & 1.43 & 1.68 & 1.60 & \textbf{2.72} & \textbf{4.00} & \textbf{3.33} & \textbf{3.25} & \textbf{3.09} & \textbf{2.92} & \textbf{2.66} \\
		Vietnam & 1.66 & 1.27 & 0.72 & \textbf{17.49} & \textbf{15.15} & \textbf{4.63} & \textbf{2.68} & \textbf{2.44} & 1.49 & 1.49 & 0.81 & 0.78 & 0.95 & 0.82 & 0.78 & 0.78 & 0.85 & 0.93 & 0.88 & 0.91 \\
		Limburg & 1.70 & \textbf{2.24} & 1.37 & 1.80 & \textbf{3.42} & \textbf{4.95} & \textbf{5.77} & \textbf{5.06} & \textbf{4.56} & \textbf{4.53} & \textbf{4.21} & \textbf{2.92} & \textbf{3.58} & \textbf{3.29} & \textbf{2.97} & \textbf{2.73} & \textbf{2.64} & \textbf{2.46} & \textbf{3.18} & \textbf{2.70} \\
		Noord Brabant & \textbf{3.90} & 1.99 & 1.23 & 1.03 & 0.89 & 0.69 & 0.78 & 1.33 & 1.20 & 1.07 & 0.92 & 0.82 & 0.96 & 0.84 & 0.69 & 0.76 & 0.62 & 0.69 & 0.57 & 0.61 \\
		Lang.-Rouss. & 0.52 & 0.18 & 0.23 & 0.65 & 0.58 & 0.89 & 0.74 & 0.68 & 0.51 & 0.54 & 0.63 & 0.58 & 0.73 & 0.81 & 0.39 & \textbf{3.41} & \textbf{3.05} & \textbf{2.70} & \textbf{2.85} & \textbf{2.59} \\

		\hline
		%Logan Airport
		Bld 1 & \textbf{54.52} & \textbf{26.57} & \textbf{14.30} & \textbf{7.84} & \textbf{10.93} & \textbf{9.93} & \textbf{4.22} & \textbf{6.72} & \textbf{2.67} & 1.70 & 1.60 & 1.44 & 1.49 & 1.28 & 1.28 & 1.32 & 1.28 & 1.28 & 1.60 & 1.77 \\
		%castle
		Bld 2 & 0.06 & 0.09 & 1.00 & 0.80 & 0.67 & 0.82 & 0.55 & 0.50 & 0.46 & 0.39 & 0.34 & 0.51 & 0.49 & \textbf{2.49} & \textbf{2.25} & \textbf{2.08} & 1.90 & 1.93 & \textbf{2.72} & \textbf{2.52} \\
		%castle Boldt
		Bld 3 & 0.53 & 0.63 & 0.59 & 0.44 & 0.45 & 0.48 & 0.41 & 0.43 & 0.92 & 0.73 & 1.23 & 1.08 & 0.93 & 0.88 & 1.02 & 0.99 & 0.83 & 1.74 & 1.58 & 1.48 \\
		%castle Dracula
		Bld 4 & 0.05 & 0.05 & 0.06 & 0.09 & 0.07 & 0.10 & 0.12 & 0.17 & 0.22 & 0.18 & 0.22 & 0.14 & 0.14 & 0.26 & 0.17 & 0.31 & 0.20 & 0.17 & 0.12 & 0.23 \\
		%Chicago Lane Technical High School
		Bld 5 & 0.32 & 0.35 & \textbf{2.92} & 1.97 & \textbf{3.82} & \textbf{3.63} & \textbf{2.82} & \textbf{3.33} & 1.88 & \textbf{3.02} & 1.67 & 1.41 & 1.52 & \textbf{2.06} & 1.87 & 1.72 & 0.95 & 1.02 & 0.98 & 1.19 \\
		%Chicago Stadium
		Bld 6 & 0.62 & 1.75 & \textbf{3.14} & 1.97 & \textbf{2.05} & 1.67 & 1.41 & 0.97 & 0.80 & 0.67 & 0.98 & 0.73 & 0.66 & 0.76 & 0.66 & 0.78 & 0.63 & 0.62 & 0.45 & 0.52 \\
		%Gaffuri
		Bld 7 & 0.02 & 0.07 & 0.03 & 0.05 & 0.12 & 0.09 & 0.08 & 0.12 & 0.11 & 0.19 & 0.15 & 0.19 & 0.25 & 0.22 & 0.66 & 0.98 & 0.90 & 0.84 & 0.90 & 0.80 \\
		%Gaffuri2
		Bld 8 & 0.14 & 0.16 & 0.14 & 1.04 & \textbf{2.32} & \textbf{2.03} & 1.60 & 1.70 & 1.78 & 1.25 & 1.14 & 1.23 & 1.24 & 0.88 & 1.04 & 1.02 & \textbf{3.82} & \textbf{3.54} & \textbf{3.37} & \textbf{3.38} \\
		%Gaffuri_rotated
		Bld 9 & 0.20 & 0.27 & 0.23 & 0.23 & 0.21 & 0.22 & 0.26 & 0.24 & 1.01 & 1.05 & 0.80 & 0.81 & 0.66 & 0.66 & 0.53 & 0.62 & 0.60 & 0.46 & 0.47 & 0.48 \\
		%Gaffuri2_rotated
		Bld 10 & 0.21 & 0.24 & 1.24 & \textbf{2.82} & 1.94 & 1.46 & 1.11 & 0.97 & 1.34 & 1.13 & 1.07 & 1.80 & \textbf{4.31} & \textbf{4.19} & \textbf{3.87} & \textbf{3.66} & \textbf{3.24} & \textbf{3.01} & \textbf{2.69} & \textbf{2.80} \\
		%TU/e Paviljoen
		Bld 11 & 0.10 & \textbf{18.31} & \textbf{12.58} & \textbf{7.51} & \textbf{5.95} & \textbf{5.11} & \textbf{4.54} & \textbf{4.74} & \textbf{3.98} & \textbf{4.06} & \textbf{3.54} & \textbf{3.38} & \textbf{3.44} & \textbf{2.35} & \textbf{2.36} & \textbf{2.55} & \textbf{2.67} & \textbf{2.02} & \textbf{2.31} & \textbf{2.20} \\

		\hline
		bird & 1.86 & \textbf{2.54} & \textbf{2.51} & 1.13 & 1.46 & 1.47 & 1.83 & 1.37 & 0.97 & 1.21 & 1.22 & 1.04 & 0.78 & \textbf{2.15} & \textbf{2.09} & \textbf{2.22} & \textbf{2.05} & 1.78 & 1.67 & 1.40 \\
		butterfly & \textbf{8.69} & 1.58 & \textbf{2.59} & \textbf{2.23} & 1.81 & 1.47 & 1.45 & 1.03 & 0.87 & 0.83 & 0.75 & 0.75 & 0.59 & 0.65 & 0.72 & 0.82 & \textbf{2.31} & \textbf{2.15} & 0.51 & 1.69 \\
		cat & 0.18 & 0.49 & 0.97 & 1.21 & \textbf{3.39} & \textbf{3.20} & \textbf{2.73} & \textbf{2.55} & 1.86 & 1.62 & 1.73 & 1.48 & 1.50 & 1.58 & 1.64 & 1.21 & 1.41 & 1.43 & 1.34 & 1.30 \\
		dog & 1.58 & \textbf{12.12} & \textbf{7.66} & \textbf{5.12} & \textbf{3.30} & \textbf{2.58} & 1.54 & 1.70 & 1.27 & 1.53 & 1.44 & 1.06 & 1.37 & 1.25 & 1.19 & 1.13 & 0.88 & 1.17 & 1.09 & 1.02 \\
		horse & 0.78 & 0.81 & \textbf{5.96} & \textbf{3.93} & \textbf{2.55} & 1.83 & 1.81 & 1.34 & 1.23 & 1.65 & 1.26 & 1.39 & 0.72 & 0.59 & 1.09 & 1.12 & 1.26 & 1.22 & 1.17 & 0.98 \\
		ostrich & 0.88 & 0.65 & 0.71 & \textbf{6.94} & \textbf{5.21} & \textbf{2.79} & \textbf{2.22} & 1.99 & 1.20 & 1.30 & 1.27 & 1.60 & 1.77 & 1.83 & 1.37 & 1.50 & 1.46 & 1.38 & 1.20 & 1.36 \\
		shark & 1.27 & 1.42 & 0.67 & 1.20 & \textbf{6.79} & \textbf{6.47} & \textbf{5.58} & \textbf{5.05} & \textbf{5.07} & \textbf{5.09} & \textbf{4.57} & \textbf{4.04} & \textbf{4.01} & \textbf{4.09} & \textbf{3.52} & \textbf{3.88} & \textbf{3.09} & \textbf{3.22} & \textbf{3.18} & \textbf{3.21} \\
		spider & \textbf{2.83} & \textbf{2.44} & \textbf{5.59} & \textbf{3.31} & \textbf{4.41} & \textbf{2.06} & \textbf{2.08} & \textbf{3.72} & \textbf{2.68} & \textbf{2.38} & \textbf{2.05} & \textbf{2.03} & \textbf{2.40} & 1.60 & 1.27 & \textbf{2.10} & 1.03 & 1.53 & 1.61 & 1.56 \\
		turtle & 0.39 & 0.42 & \textbf{2.06} & 1.55 & 1.15 & 0.65 & 0.68 & 1.87 & \textbf{2.80} & \textbf{2.59} & \textbf{2.51} & \textbf{2.33} & \textbf{2.07} & 1.55 & 1.31 & 1.51 & 1.20 & 1.29 & 1.29 & 1.00 \\

		\hline

	\end{tabular}
\end{sidewaystable}

\begin{sidewaystable}
	\setlength{\tabcolsep}{2.5pt}
	\caption{Difference between average and minimum Fr\'echet distance over the 25 runs for all cases. Significant differences (at least 2) are marked in a bold font.}
	\label{tbl-frechet-placement-avg}
	\begin{tabular}{l|rrrrrrrrrrrrrrrrrrrr}
		Scale & 1 & 2 & 3 & 4 & 5 & 6 & 7 & 8 & 9 & 10 & 11 & 12 & 13 & 14 & 15 & 16 & 17 & 18 & 19 & 20 \\

		\hline
		%Logan Airport
		Bld 1 & \textbf{4.83} & \textbf{20.33} & \textbf{12.42} & \textbf{6.55} & \textbf{7.49} & \textbf{8.00} & \textbf{3.07} & \textbf{5.65} & 1.87 & 0.95 & 0.78 & 0.73 & 0.76 & 0.59 & 0.63 & 0.72 & 0.67 & 0.66 & 0.93 & 1.05 \\
		%castle
		Bld 2 & 0.03 & 0.06 & 0.27 & 0.39 & 0.36 & 0.33 & 0.23 & 0.27 & 0.19 & 0.14 & 0.16 & 0.22 & 0.18 & 0.39 & 0.34 & 0.38 & 0.31 & 0.39 & 0.69 & 0.57 \\
		%castle Boldt
		Bld 3 & 0.09 & 0.10 & 0.09 & 0.12 & 0.16 & 0.17 & 0.17 & 0.16 & 0.17 & 0.22 & 0.35 & 0.31 & 0.29 & 0.26 & 0.36 & 0.37 & 0.28 & 0.29 & 0.35 & 0.29 \\
		%castle Dracula
		Bld 4 & 0.03 & 0.03 & 0.03 & 0.05 & 0.04 & 0.05 & 0.05 & 0.08 & 0.06 & 0.06 & 0.07 & 0.08 & 0.05 & 0.08 & 0.09 & 0.10 & 0.08 & 0.08 & 0.06 & 0.08 \\
		%Chicago Lane Technical High School
		Bld 5 & 0.07 & 0.08 & 1.37 & 1.28 & 1.24 & 1.60 & \textbf{2.04} & 1.66 & 0.61 & 1.83 & 0.85 & 0.67 & 0.57 & 1.27 & 1.23 & 1.16 & 0.52 & 0.52 & 0.48 & 0.58 \\
		%Chicago Stadium
		Bld 6 & 0.12 & 0.53 & 1.47 & 1.34 & 1.08 & 0.89 & 0.73 & 0.50 & 0.39 & 0.36 & 0.52 & 0.36 & 0.31 & 0.35 & 0.30 & 0.37 & 0.33 & 0.37 & 0.20 & 0.27 \\
		%Gaffuri
		Bld 7 & 0.01 & 0.02 & 0.02 & 0.03 & 0.05 & 0.06 & 0.04 & 0.08 & 0.08 & 0.14 & 0.09 & 0.09 & 0.13 & 0.13 & 0.14 & 0.28 & 0.24 & 0.26 & 0.36 & 0.33 \\
		%Gaffuri_rotated
		Bld 8 & 0.05 & 0.05 & 0.06 & 0.07 & 0.08 & 0.08 & 0.09 & 0.09 & 0.16 & 0.23 & 0.23 & 0.25 & 0.18 & 0.25 & 0.19 & 0.23 & 0.24 & 0.17 & 0.19 & 0.18 \\
		%Gaffuri2
		Bld 9 & 0.06 & 0.05 & 0.05 & 0.16 & 0.98 & 1.08 & 0.90 & 0.78 & 0.87 & 0.56 & 0.65 & 0.71 & 0.58 & 0.30 & 0.27 & 0.48 & 0.71 & 0.56 & 0.63 & 0.68 \\
		%Gaffuri2_rotated
		Bld 10 & 0.06 & 0.06 & 0.17 & 1.31 & 1.14 & 0.83 & 0.55 & 0.48 & 0.65 & 0.53 & 0.51 & 0.69 & 0.63 & 0.91 & 0.77 & 0.86 & 0.68 & 0.62 & 0.52 & 0.84 \\
		%TU/e Paviljoen
		Bld 11 & 0.06 & \textbf{5.10} & \textbf{8.74} & \textbf{3.18} & \textbf{2.88} & \textbf{3.07} & \textbf{2.83} & \textbf{3.08} & \textbf{2.78} & \textbf{3.05} & \textbf{2.49} & \textbf{2.10} & 1.93 & 1.43 & 1.40 & 1.51 & 1.90 & 1.48 & 1.37 & 1.22 \\
		 
		\hline
		bird & 0.62 & 1.00 & 0.98 & 0.60 & 0.64 & 0.65 & 0.99 & 0.60 & 0.43 & 0.60 & 0.62 & 0.52 & 0.34 & 0.69 & 0.68 & 0.72 & 0.63 & 0.60 & 0.61 & 0.53 \\
		butterfly & 1.89 & 1.09 & 1.36 & 1.28 & 1.15 & 0.91 & 0.97 & 0.60 & 0.56 & 0.46 & 0.43 & 0.41 & 0.31 & 0.34 & 0.41 & 0.43 & 0.50 & 0.48 & 0.24 & 0.47 \\
		cat & 0.02 & 0.08 & 0.19 & 0.28 & 0.73 & 1.56 & 1.36 & 1.48 & 0.87 & 0.80 & 0.83 & 0.74 & 0.89 & 0.78 & 0.83 & 0.56 & 0.63 & 0.61 & 0.72 & 0.65 \\
		dog & 0.50 & \textbf{3.78} & \textbf{3.98} & 1.89 & 1.67 & 1.49 & 0.68 & 0.85 & 0.68 & 0.94 & 0.69 & 0.50 & 0.89 & 0.67 & 0.67 & 0.61 & 0.42 & 0.55 & 0.51 & 0.48 \\
		horse & 0.23 & 0.23 & 1.76 & \textbf{2.53} & 1.39 & 1.06 & 1.03 & 0.54 & 0.57 & 1.04 & 0.62 & 0.62 & 0.43 & 0.27 & 0.37 & 0.44 & 0.53 & 0.47 & 0.51 & 0.50 \\
		ostrich & 0.40 & 0.32 & 0.32 & \textbf{3.04} & \textbf{2.88} & 1.42 & 1.51 & 1.44 & 0.73 & 0.86 & 0.75 & 1.08 & 1.14 & 1.13 & 0.74 & 0.87 & 0.85 & 0.69 & 0.56 & 0.52 \\
		shark & 0.50 & 0.85 & 0.31 & 0.48 & 1.05 & \textbf{2.07} & \textbf{2.29} & \textbf{2.27} & \textbf{2.72} & \textbf{2.58} & \textbf{2.62} & 1.79 & 1.75 & \textbf{2.04} & 1.82 & \textbf{2.04} & 1.69 & 1.53 & 1.88 & 1.69 \\
		spider & 1.30 & 1.45 & \textbf{3.06} & 1.58 & 1.84 & 0.66 & 0.72 & 1.93 & 1.40 & 0.96 & 0.62 & 0.92 & 1.30 & 0.77 & 0.73 & 0.96 & 0.52 & 0.61 & 0.68 & 0.58 \\
		turtle & 0.08 & 0.12 & 0.47 & 0.68 & 0.40 & 0.35 & 0.31 & 0.40 & 1.05 & 1.20 & 1.45 & 1.43 & 1.15 & 0.90 & 0.65 & 0.85 & 0.70 & 0.76 & 0.81 & 0.60 \\

		\hline
		Africa & 0.44 & 0.41 & 0.23 & 0.19 & 0.11 & 0.11 & 0.12 & 0.09 & 0.06 & 0.07 & 0.08 & 0.08 & 0.07 & 0.13 & 0.07 & 0.10 & 0.13 & 0.11 & 0.18 & 0.09 \\
		Antarctica & 0.41 & 0.44 & 0.71 & 0.82 & 0.82 & 1.10 & 0.91 & 0.88 & 0.89 & 0.93 & 0.98 & 0.87 & 0.68 & 0.55 & 0.57 & 0.55 & 0.57 & 0.47 & 0.47 & 0.39 \\
		Australia & 0.52 & 0.53 & 0.36 & 0.35 & 0.35 & 0.30 & 0.29 & 0.34 & 0.29 & 0.32 & 0.24 & 0.25 & 0.23 & 0.28 & 0.35 & 0.30 & 0.31 & 0.34 & 0.31 & 0.29 \\
		Brazil & 0.24 & 0.19 & 0.23 & 0.29 & 0.35 & 0.32 & 0.36 & 0.30 & 0.28 & 0.27 & 0.23 & 0.20 & 0.25 & 0.29 & 0.26 & 0.33 & 0.22 & 0.22 & 0.27 & 0.19 \\
		China & 0.35 & 0.36 & 0.47 & 0.44 & 0.39 & 0.39 & 0.32 & 0.28 & 0.31 & 0.18 & 0.22 & 0.24 & 0.27 & 0.43 & 0.37 & 0.42 & 0.41 & 0.46 & 0.42 & 0.35 \\
		France & 0.47 & 0.35 & 0.31 & 0.21 & 0.22 & 0.17 & 0.16 & 0.13 & 0.17 & 0.19 & 0.24 & 0.30 & 0.23 & 0.24 & 0.18 & 0.20 & 0.21 & 0.20 & 0.19 & 0.26 \\
		Greece & 0.57 & 0.76 & 0.93 & 0.71 & 0.72 & 0.69 & \textbf{2.21} & 0.73 & 0.68 & 1.38 & 0.80 & 1.02 & 1.00 & 0.87 & 0.99 & 1.13 & 1.07 & 1.27 & 1.35 & 0.93 \\
		Italy & 0.85 & 0.43 & 0.53 & 0.74 & 1.16 & 1.22 & 1.10 & 1.04 & 1.06 & 0.97 & 1.06 & 1.25 & 1.39 & 1.56 & 1.63 & 1.88 & 1.47 & 0.89 & 1.40 & 1.07 \\
		Switzerland & 1.54 & 1.05 & 0.87 & 1.11 & 0.71 & 0.73 & 0.61 & 0.74 & 0.52 & 0.53 & 0.53 & 0.56 & 0.55 & 0.50 & 0.45 & 0.35 & 0.39 & 0.42 & 0.34 & 0.23 \\
		Great Britain & 1.24 & 1.33 & 1.05 & 0.77 & 0.82 & 0.59 & 0.71 & 0.62 & 0.62 & 0.84 & 0.65 & 0.59 & 0.70 & 0.55 & 0.97 & 0.68 & 0.64 & 0.93 & 0.98 & 1.05 \\
		Vietnam & 0.65 & 0.77 & 0.41 & \textbf{7.00} & \textbf{11.86} & \textbf{3.34} & 1.77 & 1.71 & 0.90 & 0.95 & 0.38 & 0.37 & 0.45 & 0.43 & 0.40 & 0.39 & 0.43 & 0.43 & 0.48 & 0.56 \\
		Limburg & 0.70 & 0.88 & 0.60 & 0.59 & 0.77 & 1.73 & 1.98 & \textbf{2.12} & \textbf{2.24} & \textbf{2.38} & \textbf{2.34} & 1.44 & \textbf{2.18} & \textbf{2.05} & 1.72 & 1.63 & 1.37 & 1.40 & 1.17 & 1.18 \\
		Noord Brabant & 0.48 & 0.83 & 0.54 & 0.46 & 0.35 & 0.32 & 0.33 & 0.28 & 0.28 & 0.30 & 0.26 & 0.24 & 0.31 & 0.32 & 0.25 & 0.24 & 0.21 & 0.20 & 0.19 & 0.18 \\
		Lang.-Rouss. & 0.36 & 0.05 & 0.05 & 0.11 & 0.21 & 0.26 & 0.30 & 0.24 & 0.21 & 0.21 & 0.21 & 0.20 & 0.16 & 0.20 & 0.13 & 0.41 & 0.72 & 0.55 & 0.92 & 1.08 \\

		\hline

	\end{tabular}
\end{sidewaystable}

\begin{sidewaystable}
	\setlength{\tabcolsep}{2.5pt}
	\caption{The minimum Fr\'echet distance over the 25 runs for all cases, as a percentage of the predicted upper bound. Significantly better performance (less than 40 percent) are marked in a bold font.}
	\label{tbl-frechet-min-ub}
	\begin{tabular}{l|rrrrrrrrrrrrrrrrrrrr}
		Scale & 1 & 2 & 3 & 4 & 5 & 6 & 7 & 8 & 9 & 10 & 11 & 12 & 13 & 14 & 15 & 16 & 17 & 18 & 19 & 20 \\
		\hline
		Africa & 45.70 & \textbf{18.50} & \textbf{16.71} & \textbf{17.97} & \textbf{20.10} & \textbf{22.21} & \textbf{23.02} & \textbf{23.64} & \textbf{25.32} & \textbf{25.81} & \textbf{26.07} & \textbf{25.59} & \textbf{26.34} & \textbf{24.12} & \textbf{27.17} & \textbf{27.02} & \textbf{26.06} & \textbf{26.14} & \textbf{24.33} & \textbf{25.89} \\
		Antarctica & \textbf{30.80} & \textbf{24.96} & \textbf{20.44} & \textbf{15.56} & \textbf{13.59} & \textbf{9.69} & \textbf{10.58} & \textbf{11.41} & \textbf{12.02} & \textbf{12.53} & \textbf{13.21} & \textbf{14.84} & \textbf{15.58} & \textbf{15.88} & \textbf{15.61} & \textbf{16.49} & \textbf{14.47} & \textbf{14.19} & \textbf{14.67} & \textbf{15.35} \\
		Australia & \textbf{31.53} & \textbf{20.78} & \textbf{25.06} & \textbf{22.76} & \textbf{19.27} & \textbf{20.50} & \textbf{20.92} & \textbf{21.43} & \textbf{20.40} & \textbf{20.99} & \textbf{21.32} & \textbf{20.09} & \textbf{20.62} & \textbf{20.84} & \textbf{17.99} & \textbf{21.55} & \textbf{21.73} & \textbf{20.74} & \textbf{22.34} & \textbf{24.07} \\
		Brazil & \textbf{29.61} & \textbf{27.71} & \textbf{20.08} & \textbf{17.71} & \textbf{19.07} & \textbf{19.16} & \textbf{16.81} & \textbf{18.42} & \textbf{13.52} & \textbf{13.23} & \textbf{13.65} & \textbf{14.55} & \textbf{14.23} & \textbf{14.00} & \textbf{13.81} & \textbf{14.23} & \textbf{15.94} & \textbf{15.94} & \textbf{15.62} & \textbf{17.64} \\
		China & \textbf{28.47} & \textbf{10.05} & \textbf{11.85} & \textbf{13.96} & \textbf{15.70} & \textbf{16.80} & \textbf{20.40} & \textbf{22.98} & \textbf{20.50} & \textbf{21.77} & \textbf{10.67} & \textbf{9.41} & \textbf{10.01} & \textbf{10.20} & \textbf{10.74} & \textbf{10.83} & \textbf{11.18} & \textbf{11.60} & \textbf{11.02} & \textbf{11.53} \\
		France & \textbf{23.10} & \textbf{24.88} & \textbf{24.62} & \textbf{28.84} & \textbf{25.93} & \textbf{27.99} & \textbf{24.47} & \textbf{23.70} & \textbf{18.30} & \textbf{19.31} & \textbf{18.83} & \textbf{20.02} & \textbf{19.73} & \textbf{19.71} & \textbf{21.20} & \textbf{20.87} & \textbf{20.52} & \textbf{23.00} & \textbf{20.80} & \textbf{18.71} \\
		Greece & \textbf{22.16} & \textbf{19.86} & \textbf{19.47} & \textbf{19.26} & \textbf{15.48} & \textbf{12.97} & \textbf{5.36} & \textbf{12.10} & \textbf{11.67} & \textbf{7.22} & \textbf{10.66} & \textbf{11.55} & \textbf{10.60} & \textbf{12.07} & \textbf{10.34} & \textbf{9.26} & \textbf{10.25} & \textbf{9.32} & \textbf{8.04} & \textbf{11.64} \\
		Italy & \textbf{32.11} & \textbf{25.22} & \textbf{14.00} & \textbf{10.48} & \textbf{13.90} & \textbf{10.72} & \textbf{12.23} & \textbf{13.42} & \textbf{6.78} & \textbf{6.75} & \textbf{6.70} & \textbf{7.04} & \textbf{8.19} & \textbf{7.54} & \textbf{7.82} & \textbf{7.95} & \textbf{13.58} & \textbf{21.77} & \textbf{14.79} & \textbf{20.94} \\
		Switzerland & \textbf{17.64} & \textbf{10.54} & \textbf{13.58} & \textbf{16.38} & \textbf{19.77} & \textbf{19.29} & \textbf{21.22} & \textbf{16.70} & \textbf{23.29} & \textbf{23.40} & \textbf{17.49} & \textbf{18.45} & \textbf{19.17} & \textbf{15.87} & \textbf{14.77} & \textbf{17.11} & \textbf{14.61} & \textbf{15.10} & \textbf{15.20} & \textbf{18.06} \\
		Great Britain & \textbf{26.95} & \textbf{22.09} & \textbf{20.56} & \textbf{19.63} & \textbf{17.07} & \textbf{5.04} & \textbf{4.64} & \textbf{4.91} & \textbf{5.68} & \textbf{4.32} & \textbf{4.83} & \textbf{5.89} & \textbf{5.30} & \textbf{6.91} & \textbf{4.88} & \textbf{6.95} & \textbf{6.89} & \textbf{7.02} & \textbf{7.15} & \textbf{7.79} \\
		Vietnam & 62.52 & 42.62 & \textbf{26.46} & \textbf{2.55} & \textbf{3.11} & \textbf{37.15} & 43.46 & 42.91 & 47.52 & 46.00 & 50.08 & 49.39 & 47.77 & 47.35 & 47.21 & 46.50 & 45.93 & 44.89 & 43.28 & 41.97 \\
		Limburg & \textbf{10.81} & \textbf{16.64} & \textbf{18.34} & \textbf{12.73} & \textbf{10.60} & \textbf{4.90} & \textbf{6.15} & \textbf{10.72} & \textbf{9.15} & \textbf{6.67} & \textbf{6.27} & \textbf{13.19} & \textbf{7.46} & \textbf{6.67} & \textbf{8.55} & \textbf{9.92} & \textbf{7.77} & \textbf{8.52} & \textbf{8.85} & \textbf{12.42} \\
		Noord Brabant & \textbf{11.14} & \textbf{17.78} & \textbf{20.39} & \textbf{22.47} & \textbf{14.99} & \textbf{16.74} & \textbf{17.00} & \textbf{19.21} & \textbf{19.07} & \textbf{21.49} & \textbf{21.32} & \textbf{22.14} & \textbf{20.98} & \textbf{21.86} & \textbf{23.18} & \textbf{18.79} & \textbf{16.69} & \textbf{12.68} & \textbf{12.37} & \textbf{12.83} \\
		Lang.-Rouss. & \textbf{24.91} & \textbf{25.00} & \textbf{23.00} & \textbf{21.20} & \textbf{19.04} & \textbf{17.69} & \textbf{17.92} & \textbf{19.51} & \textbf{21.33} & \textbf{21.75} & \textbf{21.15} & \textbf{21.17} & \textbf{23.67} & \textbf{6.65} & \textbf{7.09} & \textbf{6.40} & \textbf{7.98} & \textbf{7.46} & \textbf{8.08} & \textbf{8.27} \\

		\hline
		%Logan Airport
		Bld 1 & \textbf{0.88} & \textbf{2.40} & \textbf{13.50} & \textbf{24.60} & \textbf{19.77} & \textbf{17.16} & \textbf{39.96} & \textbf{21.18} & 43.37 & 48.94 & 48.83 & 48.23 & 47.21 & 48.34 & 46.54 & 43.73 & 43.57 & 43.07 & \textbf{36.69} & \textbf{34.09} \\
		%castle
		Bld 2 & \textbf{37.22} & \textbf{22.28} & \textbf{20.30} & \textbf{18.80} & \textbf{20.21} & \textbf{21.86} & \textbf{22.41} & \textbf{7.04} & \textbf{8.12} & \textbf{9.50} & \textbf{9.27} & \textbf{4.87} & \textbf{5.29} & \textbf{5.18} & \textbf{5.96} & \textbf{5.98} & \textbf{6.13} & \textbf{7.31} & \textbf{5.57} & \textbf{5.24} \\
		%castle Boldt
		Bld 3 & \textbf{29.90} & \textbf{31.05} & \textbf{33.52} & \textbf{26.17} & \textbf{28.13} & \textbf{27.70} & \textbf{14.97} & \textbf{16.83} & \textbf{18.65} & \textbf{15.61} & \textbf{15.89} & \textbf{17.59} & \textbf{16.36} & \textbf{18.38} & \textbf{11.76} & \textbf{10.36} & \textbf{12.57} & \textbf{13.89} & \textbf{13.20} & \textbf{14.34} \\
		%castle Dracula
		Bld 4 & \textbf{39.74} & \textbf{39.11} & \textbf{37.41} & \textbf{35.02} & \textbf{35.04} & \textbf{34.43} & \textbf{34.20} & \textbf{33.30} & \textbf{34.20} & \textbf{35.44} & \textbf{34.38} & \textbf{33.38} & \textbf{35.26} & \textbf{33.75} & \textbf{33.17} & \textbf{32.03} & \textbf{32.71} & \textbf{32.06} & \textbf{32.30} & \textbf{28.76} \\
		%Chicago Lane Technical High School
		Bld 5 & \textbf{31.99} & \textbf{10.10} & \textbf{4.12} & \textbf{5.73} & \textbf{13.53} & \textbf{6.39} & \textbf{8.83} & \textbf{11.42} & \textbf{25.33} & \textbf{8.67} & \textbf{24.37} & \textbf{19.09} & \textbf{18.10} & \textbf{8.02} & \textbf{8.45} & \textbf{8.77} & \textbf{23.40} & \textbf{21.13} & \textbf{20.96} & \textbf{19.69} \\
		%Chicago Stadium
		Bld 6 & \textbf{23.16} & \textbf{13.92} & \textbf{9.46} & \textbf{16.91} & \textbf{17.86} & \textbf{22.36} & \textbf{26.70} & \textbf{30.36} & \textbf{31.32} & \textbf{32.04} & \textbf{24.19} & \textbf{26.54} & \textbf{27.68} & \textbf{23.25} & \textbf{17.14} & \textbf{15.65} & \textbf{16.79} & \textbf{14.56} & \textbf{20.55} & \textbf{17.06} \\
		%Gaffuri
		Bld 7 & \textbf{39.53} & \textbf{38.24} & \textbf{38.08} & \textbf{36.91} & \textbf{35.53} & \textbf{34.79} & \textbf{35.54} & \textbf{33.23} & \textbf{33.32} & \textbf{16.21} & \textbf{17.53} & \textbf{17.66} & \textbf{16.82} & \textbf{16.74} & \textbf{18.66} & \textbf{14.94} & \textbf{19.03} & \textbf{21.41} & \textbf{19.78} & \textbf{19.40} \\
		%Gaffuri_rotated
		Bld 8 & \textbf{38.39} & \textbf{37.64} & \textbf{37.46} & \textbf{36.12} & \textbf{35.84} & \textbf{35.46} & \textbf{35.50} & \textbf{18.75} & \textbf{19.64} & \textbf{20.65} & \textbf{20.76} & \textbf{22.00} & \textbf{23.26} & \textbf{23.06} & \textbf{25.32} & \textbf{23.62} & \textbf{23.32} & \textbf{25.43} & \textbf{24.67} & \textbf{25.76} \\
		%Gaffuri2
		Bld 9 & \textbf{32.66} & \textbf{32.65} & \textbf{18.75} & \textbf{7.72} & \textbf{9.17} & \textbf{7.86} & \textbf{9.59} & \textbf{14.66} & \textbf{10.90} & \textbf{18.10} & \textbf{8.04} & \textbf{7.22} & \textbf{7.06} & \textbf{11.04} & \textbf{11.13} & \textbf{8.28} & \textbf{8.39} & \textbf{7.90} & \textbf{8.86} & \textbf{7.87} \\
		%Gaffuri2_rotated
		Bld 10 & \textbf{36.53} & \textbf{35.92} & \textbf{15.96} & \textbf{8.93} & \textbf{16.43} & \textbf{14.49} & \textbf{17.61} & \textbf{16.25} & \textbf{14.79} & \textbf{7.46} & \textbf{6.89} & \textbf{6.54} & \textbf{7.55} & \textbf{7.68} & \textbf{8.09} & \textbf{8.28} & \textbf{9.29} & \textbf{8.62} & \textbf{9.56} & \textbf{8.62} \\
		%TU/e Paviljoen
		Bld 11 & \textbf{5.69} & \textbf{1.37} & \textbf{1.39} & \textbf{14.43} & \textbf{14.20} & \textbf{12.79} & \textbf{12.35} & \textbf{8.19} & \textbf{9.12} & \textbf{7.02} & \textbf{8.30} & \textbf{7.39} & \textbf{5.47} & \textbf{11.90} & \textbf{9.66} & \textbf{6.96} & \textbf{6.91} & \textbf{9.58} & \textbf{8.34} & \textbf{7.19} \\

		\hline
		bird & \textbf{24.49} & \textbf{16.75} & \textbf{25.22} & \textbf{36.18} & \textbf{35.56} & \textbf{29.57} & \textbf{28.27} & \textbf{35.72} & \textbf{20.98} & \textbf{18.48} & \textbf{18.14} & \textbf{16.28} & \textbf{18.49} & \textbf{16.88} & \textbf{16.95} & \textbf{17.17} & \textbf{17.08} & \textbf{19.76} & \textbf{20.16} & \textbf{18.54} \\
		butterfly & \textbf{19.51} & \textbf{35.99} & \textbf{38.37} & \textbf{21.74} & \textbf{21.02} & \textbf{21.08} & \textbf{18.91} & \textbf{20.58} & \textbf{19.71} & \textbf{20.19} & \textbf{20.16} & \textbf{20.04} & \textbf{15.44} & \textbf{15.37} & \textbf{13.43} & \textbf{10.56} & \textbf{10.63} & \textbf{10.88} & \textbf{13.35} & \textbf{13.05} \\
		cat & \textbf{38.82} & \textbf{27.07} & \textbf{24.72} & \textbf{18.96} & \textbf{9.40} & \textbf{9.75} & \textbf{19.47} & \textbf{15.94} & \textbf{25.61} & \textbf{27.28} & \textbf{24.83} & \textbf{26.24} & \textbf{21.03} & \textbf{20.46} & \textbf{20.31} & \textbf{25.27} & \textbf{21.06} & \textbf{21.63} & \textbf{17.48} & \textbf{19.09} \\
		dog & \textbf{16.25} & \textbf{3.11} & \textbf{5.80} & \textbf{10.25} & \textbf{18.36} & \textbf{23.92} & \textbf{35.60} & \textbf{31.63} & \textbf{33.54} & \textbf{27.42} & \textbf{29.52} & \textbf{33.45} & \textbf{24.59} & \textbf{28.22} & \textbf{26.14} & \textbf{28.39} & \textbf{32.46} & \textbf{13.36} & \textbf{13.20} & \textbf{13.63} \\
		horse & \textbf{26.18} & \textbf{3.18} & \textbf{4.62} & \textbf{15.85} & \textbf{27.75} & \textbf{31.76} & \textbf{30.33} & \textbf{37.48} & \textbf{30.34} & \textbf{22.66} & \textbf{27.40} & \textbf{27.32} & \textbf{28.27} & \textbf{31.85} & \textbf{31.74} & \textbf{25.45} & \textbf{23.19} & \textbf{22.02} & \textbf{21.13} & \textbf{21.75} \\
		ostrich & \textbf{23.07} & \textbf{30.02} & \textbf{11.40} & \textbf{5.17} & \textbf{23.01} & 42.87 & 41.34 & \textbf{38.80} & 43.93 & 40.29 & \textbf{39.00} & \textbf{33.42} & \textbf{29.92} & \textbf{26.53} & \textbf{31.99} & \textbf{27.34} & \textbf{25.27} & \textbf{27.71} & \textbf{26.49} & \textbf{22.33} \\
		shark & \textbf{13.84} & \textbf{16.37} & \textbf{6.40} & \textbf{5.92} & \textbf{7.07} & \textbf{5.55} & \textbf{6.20} & \textbf{8.96} & \textbf{7.34} & \textbf{7.21} & \textbf{6.81} & \textbf{13.24} & \textbf{12.91} & \textbf{10.36} & \textbf{12.81} & \textbf{8.87} & \textbf{14.39} & \textbf{14.13} & \textbf{9.81} & \textbf{9.52} \\
		spider & \textbf{14.79} & \textbf{9.17} & \textbf{18.12} & \textbf{37.17} & \textbf{39.61} & 50.32 & \textbf{38.07} & \textbf{28.64} & \textbf{27.43} & \textbf{18.44} & \textbf{20.15} & \textbf{17.97} & \textbf{15.41} & \textbf{15.06} & \textbf{14.22} & \textbf{10.54} & \textbf{13.74} & \textbf{12.61} & \textbf{11.29} & \textbf{11.33} \\
		turtle & \textbf{35.87} & \textbf{26.25} & \textbf{13.78} & \textbf{24.91} & \textbf{23.11} & \textbf{21.26} & \textbf{14.12} & \textbf{15.14} & \textbf{12.61} & \textbf{13.18} & \textbf{13.48} & \textbf{14.34} & \textbf{16.51} & \textbf{25.31} & \textbf{29.32} & \textbf{22.59} & \textbf{26.09} & \textbf{22.56} & \textbf{19.78} & \textbf{27.00} \\

		\hline

	\end{tabular}
\end{sidewaystable}

\begin{sidewaystable}
	\setlength{\tabcolsep}{2.5pt}
	\caption{The average Fr\'echet distance over the 25 runs for all cases, as a percentage of the predicted upper bound. Significantly better performance (less than 40 percent) are marked in a bold font.}
	\label{tbl-frechet-avg-ub}
		\begin{tabular}{l|rrrrrrrrrrrrrrrrrrrr}
		Scale & 1 & 2 & 3 & 4 & 5 & 6 & 7 & 8 & 9 & 10 & 11 & 12 & 13 & 14 & 15 & 16 & 17 & 18 & 19 & 20 \\

		\hline
		Africa & 53.77 & \textbf{26.55} & \textbf{22.08} & \textbf{23.12} & \textbf{23.43} & \textbf{25.89} & \textbf{27.64} & \textbf{27.28} & \textbf{27.79} & \textbf{28.70} & \textbf{29.50} & \textbf{29.18} & \textbf{29.36} & \textbf{30.11} & \textbf{30.29} & \textbf{31.67} & \textbf{32.10} & \textbf{31.25} & \textbf{32.61} & \textbf{30.00} \\
		Antarctica & 41.52 & \textbf{36.45} & \textbf{37.58} & \textbf{33.87} & \textbf{29.17} & \textbf{24.95} & \textbf{24.80} & \textbf{26.05} & \textbf{28.32} & \textbf{30.74} & \textbf{33.70} & \textbf{34.13} & \textbf{31.46} & \textbf{29.24} & \textbf{30.22} & \textbf{30.96} & \textbf{28.38} & \textbf{25.99} & \textbf{26.69} & \textbf{25.56} \\
		Australia & 45.21 & \textbf{31.71} & \textbf{34.51} & \textbf{33.16} & \textbf{28.85} & \textbf{29.64} & \textbf{30.29} & \textbf{32.43} & \textbf{29.87} & \textbf{31.44} & \textbf{29.34} & \textbf{28.35} & \textbf{28.51} & \textbf{30.84} & \textbf{30.11} & \textbf{32.33} & \textbf{33.17} & \textbf{33.64} & \textbf{34.14} & \textbf{35.31} \\
		Brazil & \textbf{32.13} & \textbf{31.13} & \textbf{24.39} & \textbf{24.41} & \textbf{27.75} & \textbf{27.87} & \textbf{26.19} & \textbf{26.72} & \textbf{19.34} & \textbf{18.91} & \textbf{18.64} & \textbf{19.09} & \textbf{20.02} & \textbf{21.07} & \textbf{20.47} & \textbf{23.00} & \textbf{21.90} & \textbf{22.13} & \textbf{23.57} & \textbf{23.30} \\
		China & 40.92 & \textbf{14.88} & \textbf{19.79} & \textbf{22.82} & \textbf{24.40} & \textbf{26.82} & \textbf{29.37} & \textbf{31.24} & \textbf{30.23} & \textbf{27.62} & \textbf{14.31} & \textbf{12.98} & \textbf{14.24} & \textbf{17.21} & \textbf{17.10} & \textbf{18.46} & \textbf{18.81} & \textbf{20.55} & \textbf{19.07} & \textbf{18.51} \\
		France & \textbf{38.38} & \textbf{37.84} & \textbf{35.57} & \textbf{36.69} & \textbf{34.93} & \textbf{35.14} & \textbf{31.02} & \textbf{28.66} & \textbf{23.22} & \textbf{25.09} & \textbf{26.29} & \textbf{29.59} & \textbf{27.40} & \textbf{27.87} & \textbf{27.57} & \textbf{28.15} & \textbf{28.45} & \textbf{30.74} & \textbf{27.53} & \textbf{28.05} \\
		Greece & \textbf{22.74} & \textbf{21.33} & \textbf{22.08} & \textbf{21.87} & \textbf{18.21} & \textbf{15.62} & \textbf{15.09} & \textbf{15.76} & \textbf{15.40} & \textbf{15.59} & \textbf{15.74} & \textbf{18.42} & \textbf{17.80} & \textbf{18.74} & \textbf{17.97} & \textbf{18.45} & \textbf{19.42} & \textbf{20.35} & \textbf{20.34} & \textbf{20.49} \\
		Italy & 45.99 & \textbf{35.65} & \textbf{22.99} & \textbf{20.06} & \textbf{30.52} & \textbf{29.95} & \textbf{31.57} & \textbf{33.48} & \textbf{17.07} & \textbf{16.77} & \textbf{18.36} & \textbf{21.44} & \textbf{23.00} & \textbf{24.47} & \textbf{26.65} & \textbf{30.96} & \textbf{32.57} & \textbf{33.81} & \textbf{34.78} & \textbf{36.92} \\
		Switzerland & 40.21 & \textbf{22.78} & \textbf{27.38} & \textbf{38.15} & \textbf{35.05} & \textbf{35.57} & \textbf{35.47} & \textbf{33.41} & \textbf{35.68} & \textbf{36.49} & \textbf{31.42} & \textbf{33.70} & \textbf{34.66} & \textbf{30.13} & \textbf{25.25} & \textbf{25.11} & \textbf{23.54} & \textbf{25.06} & \textbf{23.60} & \textbf{23.88} \\
		Great Britain & \textbf{33.29} & \textbf{32.10} & \textbf{30.84} & \textbf{27.71} & \textbf{26.60} & \textbf{7.01} & \textbf{7.37} & \textbf{7.58} & \textbf{8.62} & \textbf{7.94} & \textbf{7.89} & \textbf{8.85} & \textbf{9.07} & \textbf{10.03} & \textbf{10.80} & \textbf{11.37} & \textbf{11.29} & \textbf{13.76} & \textbf{14.66} & \textbf{16.21} \\
		Vietnam & 71.93 & 60.45 & \textbf{34.75} & \textbf{22.68} & 45.54 & 51.42 & 52.26 & 52.54 & 53.23 & 52.65 & 52.97 & 52.47 & 51.79 & 51.46 & 51.34 & 50.76 & 50.91 & 50.11 & 49.32 & 49.37 \\
		Limburg & \textbf{21.37} & \textbf{37.73} & \textbf{34.25} & \textbf{21.07} & \textbf{19.02} & \textbf{14.91} & \textbf{18.96} & \textbf{26.06} & \textbf{26.82} & \textbf{26.67} & \textbf{26.79} & \textbf{25.19} & \textbf{26.77} & \textbf{25.83} & \textbf{25.63} & \textbf{26.95} & \textbf{22.73} & \textbf{24.56} & \textbf{22.97} & \textbf{27.41} \\
		Noord Brabant & \textbf{18.96} & \textbf{38.72} & \textbf{36.27} & \textbf{37.75} & \textbf{22.96} & \textbf{24.68} & \textbf{25.94} & \textbf{27.42} & \textbf{28.01} & \textbf{31.36} & \textbf{30.35} & \textbf{30.48} & \textbf{31.83} & \textbf{33.41} & \textbf{32.48} & \textbf{26.03} & \textbf{22.23} & \textbf{16.84} & \textbf{16.51} & \textbf{16.78} \\
		Lang.-Rouss. & \textbf{27.25} & \textbf{25.58} & \textbf{23.84} & \textbf{23.46} & \textbf{23.73} & \textbf{24.11} & \textbf{26.18} & \textbf{26.65} & \textbf{27.86} & \textbf{28.78} & \textbf{27.86} & \textbf{28.10} & \textbf{29.49} & \textbf{8.87} & \textbf{8.61} & \textbf{11.40} & \textbf{17.26} & \textbf{14.92} & \textbf{21.19} & \textbf{24.33} \\

		\hline
		%Logan Airport
		Bld 1 & \textbf{5.82} & 43.28 & 50.34 & 50.07 & 48.49 & 53.46 & 56.00 & 54.42 & 55.56 & 55.75 & 54.92 & 54.32 & 54.00 & 53.98 & 52.86 & 51.42 & 51.07 & 50.80 & 48.05 & 47.49 \\
		%castle
		Bld 2 & \textbf{39.15} & \textbf{24.47} & \textbf{29.13} & \textbf{31.91} & \textbf{31.78} & \textbf{33.51} & \textbf{31.14} & \textbf{10.40} & \textbf{10.65} & \textbf{11.53} & \textbf{11.71} & \textbf{6.88} & \textbf{7.08} & \textbf{9.16} & \textbf{9.59} & \textbf{10.33} & \textbf{9.85} & \textbf{12.28} & \textbf{14.83} & \textbf{13.26} \\
		%castle Boldt
		Bld 3 & \textbf{33.96} & \textbf{35.74} & \textbf{38.13} & \textbf{31.22} & \textbf{34.73} & \textbf{35.45} & \textbf{19.34} & \textbf{21.30} & \textbf{23.79} & \textbf{21.81} & \textbf{26.53} & \textbf{27.19} & \textbf{25.73} & \textbf{27.19} & \textbf{20.03} & \textbf{19.04} & \textbf{19.55} & \textbf{21.28} & \textbf{22.25} & \textbf{22.26} \\
		%castle Dracula
		Bld 4 & 41.55 & 40.88 & \textbf{38.92} & \textbf{37.62} & \textbf{37.48} & \textbf{37.16} & \textbf{37.22} & \textbf{37.63} & \textbf{37.91} & \textbf{38.68} & \textbf{38.43} & \textbf{38.21} & \textbf{38.13} & \textbf{38.58} & \textbf{38.53} & \textbf{38.04} & \textbf{37.23} & \textbf{36.39} & \textbf{35.69} & \textbf{32.77} \\
		%Chicago Lane Technical High School
		Bld 5 & \textbf{35.24} & \textbf{11.42} & \textbf{14.87} & \textbf{17.18} & \textbf{26.88} & \textbf{21.41} & \textbf{28.41} & \textbf{29.28} & \textbf{32.55} & \textbf{32.02} & \textbf{36.01} & \textbf{27.76} & \textbf{25.88} & \textbf{26.47} & \textbf{27.26} & \textbf{27.40} & \textbf{32.26} & \textbf{30.33} & \textbf{29.77} & \textbf{30.58} \\
		%Chicago Stadium
		Bld 6 & \textbf{27.02} & \textbf{24.68} & \textbf{29.32} & \textbf{38.93} & \textbf{38.72} & 41.61 & 44.13 & 43.33 & 41.95 & 42.54 & \textbf{39.37} & \textbf{37.37} & \textbf{37.19} & \textbf{34.37} & \textbf{23.88} & \textbf{24.53} & \textbf{25.04} & \textbf{24.15} & \textbf{25.87} & \textbf{24.39} \\
		%Gaffuri
		Bld 7 & 40.01 & \textbf{39.62} & \textbf{39.19} & \textbf{38.81} & \textbf{38.60} & \textbf{37.95} & \textbf{38.04} & \textbf{37.72} & \textbf{37.65} & \textbf{20.51} & \textbf{20.49} & \textbf{20.56} & \textbf{21.08} & \textbf{21.19} & \textbf{23.68} & \textbf{24.98} & \textbf{27.72} & \textbf{31.24} & \textbf{33.46} & \textbf{32.40} \\
		%Gaffuri_rotated
		Bld 8 & 41.09 & 40.70 & 40.72 & 40.26 & 40.26 & 40.06 & 40.63 & \textbf{21.67} & \textbf{24.60} & \textbf{28.30} & \textbf{28.47} & \textbf{30.92} & \textbf{29.81} & \textbf{32.47} & \textbf{32.56} & \textbf{32.52} & \textbf{32.86} & \textbf{32.29} & \textbf{32.44} & \textbf{33.20} \\
		%Gaffuri2
		Bld 9 & \textbf{35.93} & \textbf{35.14} & \textbf{20.28} & \textbf{9.77} & \textbf{24.07} & \textbf{21.83} & \textbf{22.63} & \textbf{27.08} & \textbf{25.79} & \textbf{28.42} & \textbf{13.28} & \textbf{13.42} & \textbf{12.43} & \textbf{13.96} & \textbf{13.95} & \textbf{13.31} & \textbf{16.29} & \textbf{14.45} & \textbf{16.69} & \textbf{16.62} \\
		%Gaffuri2_rotated
		Bld 10 & \textbf{39.55} & \textbf{39.09} & \textbf{20.15} & \textbf{27.07} & \textbf{34.99} & \textbf{26.00} & \textbf{26.07} & \textbf{24.34} & \textbf{26.62} & \textbf{11.74} & \textbf{11.35} & \textbf{13.02} & \textbf{13.84} & \textbf{17.18} & \textbf{16.44} & \textbf{18.11} & \textbf{17.60} & \textbf{16.54} & \textbf{16.60} & \textbf{20.44} \\
		%TU/e Paviljoen
		Bld 11 & \textbf{6.25} & \textbf{12.46} & \textbf{20.98} & \textbf{23.80} & \textbf{24.67} & \textbf{25.00} & \textbf{25.32} & \textbf{24.15} & \textbf{25.17} & \textbf{26.24} & \textbf{25.36} & \textbf{22.90} & \textbf{20.71} & \textbf{23.95} & \textbf{21.42} & \textbf{20.41} & \textbf{24.67} & \textbf{23.97} & \textbf{22.25} & \textbf{19.50} \\

		\hline
		bird & \textbf{39.38} & \textbf{34.34} & 42.94 & 46.62 & 47.27 & 40.30 & 45.99 & 47.38 & \textbf{26.20} & \textbf{26.34} & \textbf{26.72} & \textbf{23.86} & \textbf{23.71} & \textbf{28.04} & \textbf{28.45} & \textbf{29.89} & \textbf{28.66} & \textbf{31.21} & \textbf{32.16} & \textbf{28.45} \\
		butterfly & \textbf{29.66} & 41.86 & 48.55 & \textbf{29.09} & \textbf{29.05} & \textbf{28.48} & \textbf{27.91} & \textbf{26.78} & \textbf{26.07} & \textbf{25.89} & \textbf{25.85} & \textbf{25.78} & \textbf{18.88} & \textbf{19.40} & \textbf{18.57} & \textbf{15.15} & \textbf{16.22} & \textbf{16.51} & \textbf{16.31} & \textbf{19.05} \\
		cat & \textbf{39.91} & \textbf{30.46} & \textbf{31.83} & \textbf{27.24} & \textbf{19.86} & \textbf{32.13} & 40.72 & 41.08 & 41.58 & 42.88 & 41.92 & 42.46 & 41.39 & \textbf{38.98} & 40.78 & \textbf{39.59} & \textbf{37.52} & \textbf{38.19} & \textbf{37.61} & \textbf{37.65} \\
		dog & \textbf{27.76} & \textbf{19.66} & \textbf{31.00} & \textbf{25.73} & \textbf{34.89} & 41.07 & 44.41 & 43.63 & 44.14 & 43.24 & 41.79 & 42.79 & 42.10 & 42.08 & 40.58 & 42.06 & 42.10 & \textbf{19.83} & \textbf{19.33} & \textbf{19.52} \\
		horse & \textbf{34.92} & \textbf{4.26} & \textbf{16.30} & \textbf{37.60} & 42.27 & 44.65 & 42.75 & 44.71 & \textbf{37.01} & \textbf{35.78} & \textbf{35.76} & \textbf{36.24} & \textbf{34.78} & \textbf{36.25} & \textbf{37.90} & \textbf{32.15} & \textbf{31.66} & \textbf{29.90} & \textbf{29.95} & \textbf{30.76} \\
		ostrich & \textbf{35.41} & 40.52 & \textbf{16.46} & \textbf{29.04} & 46.98 & 56.10 & 57.09 & 55.25 & 52.58 & 51.32 & 49.38 & 49.33 & 47.77 & 45.13 & 44.75 & 42.94 & 41.21 & 41.19 & \textbf{37.14} & \textbf{30.60} \\
		shark & \textbf{22.10} & \textbf{32.91} & \textbf{7.99} & \textbf{8.89} & \textbf{14.32} & \textbf{21.73} & \textbf{25.23} & \textbf{29.30} & \textbf{33.11} & \textbf{32.95} & \textbf{34.61} & \textbf{33.47} & \textbf{31.90} & \textbf{33.63} & \textbf{34.40} & \textbf{33.85} & \textbf{36.33} & \textbf{35.07} & \textbf{36.71} & \textbf{34.91} \\
		spider & \textbf{35.27} & \textbf{20.05} & \textbf{37.63} & 46.80 & 53.24 & 56.06 & 43.48 & 44.97 & 40.53 & \textbf{24.50} & \textbf{24.44} & \textbf{24.77} & \textbf{25.79} & \textbf{20.33} & \textbf{19.54} & \textbf{16.83} & \textbf{17.32} & \textbf{17.02} & \textbf{16.38} & \textbf{15.78} \\
		turtle & 40.14 & \textbf{30.88} & \textbf{23.35} & 41.32 & \textbf{33.95} & \textbf{31.69} & \textbf{20.65} & \textbf{23.97} & \textbf{31.44} & \textbf{36.39} & 43.34 & 45.58 & 43.01 & 46.99 & 45.79 & 44.87 & 45.31 & 44.10 & 43.30 & 45.03 \\

		\hline

	\end{tabular}
\end{sidewaystable}

\begin{sidewaystable}
	\setlength{\tabcolsep}{2.5pt}
	\caption{The maximum Fr\'echet distance over the 25 runs for all cases, as a percentage of the predicted upper bound. Significantly better performance (less than 40 percent) are marked in a bold font.}
	\label{tbl-frechet-max-ub}
	\begin{tabular}{l|rrrrrrrrrrrrrrrrrrrr}
		Scale & 1 & 2 & 3 & 4 & 5 & 6 & 7 & 8 & 9 & 10 & 11 & 12 & 13 & 14 & 15 & 16 & 17 & 18 & 19 & 20 \\

		\hline 
		Africa & 58.20 & \textbf{35.10} & \textbf{31.39} & \textbf{33.25} & \textbf{34.51} & \textbf{36.36} & 41.84 & \textbf{35.55} & \textbf{34.69} & \textbf{35.41} & \textbf{36.92} & \textbf{39.68} & \textbf{39.33} & 42.84 & 40.62 & 49.95 & 52.09 & 46.97 & 52.87 & \textbf{39.33} \\
		Antarctica & 62.35 & 54.79 & 65.29 & 57.31 & 56.62 & 57.09 & 55.25 & 52.44 & 51.07 & 51.80 & 51.22 & 50.95 & 48.92 & 46.78 & 46.55 & 48.50 & 43.69 & 40.65 & 41.19 & \textbf{39.12} \\
		Australia & 75.82 & 43.71 & 51.38 & 50.92 & 46.66 & 45.43 & 51.10 & 64.95 & 55.90 & 56.44 & 50.65 & 47.95 & 50.53 & 50.98 & 50.92 & 54.62 & 56.17 & 54.39 & 55.40 & 56.77 \\
		Brazil & \textbf{33.73} & \textbf{37.60} & \textbf{38.11} & \textbf{35.60} & 52.65 & 45.20 & 40.14 & 45.50 & \textbf{29.00} & \textbf{31.43} & \textbf{30.83} & \textbf{29.79} & \textbf{28.58} & 44.33 & \textbf{27.42} & 47.43 & 43.59 & 45.02 & 46.39 & 41.06 \\
		China & 54.52 & \textbf{24.00} & \textbf{30.11} & \textbf{34.70} & 46.45 & 49.11 & 50.88 & 48.86 & 49.47 & 44.20 & \textbf{19.08} & 40.01 & 43.45 & 45.16 & 41.04 & 45.50 & 43.00 & 41.93 & \textbf{38.88} & \textbf{38.32} \\
		France & 65.49 & 53.33 & 51.68 & 51.00 & 48.27 & 45.11 & 40.91 & \textbf{39.48} & 41.11 & 42.74 & 49.35 & 50.89 & 46.25 & 43.36 & 42.81 & 43.60 & 44.60 & 53.10 & 43.08 & 42.92 \\
		Greece & \textbf{23.40} & \textbf{22.57} & \textbf{25.62} & \textbf{25.11} & \textbf{21.56} & \textbf{17.99} & \textbf{17.76} & \textbf{23.91} & \textbf{23.15} & \textbf{28.39} & \textbf{26.93} & \textbf{33.45} & \textbf{33.10} & \textbf{33.50} & \textbf{32.74} & \textbf{33.26} & \textbf{31.91} & \textbf{31.01} & \textbf{32.19} & \textbf{30.86} \\
		Italy & 55.73 & 47.42 & 43.69 & \textbf{39.47} & 54.30 & 50.61 & 59.87 & 61.72 & \textbf{29.10} & \textbf{29.00} & 42.41 & 43.57 & 41.72 & 40.26 & 42.95 & 52.78 & 52.97 & 51.55 & 51.49 & 52.20 \\
		Switzerland & 61.98 & 62.44 & 56.71 & 59.95 & 52.31 & 55.06 & 53.12 & 56.73 & 55.01 & 53.51 & 50.13 & 50.22 & 50.22 & 47.40 & \textbf{36.05} & \textbf{33.46} & \textbf{30.87} & \textbf{32.77} & \textbf{30.17} & \textbf{30.24} \\
		Great Britain & 44.30 & \textbf{38.58} & 42.51 & \textbf{37.77} & \textbf{37.41} & \textbf{10.93} & \textbf{10.73} & \textbf{10.67} & \textbf{12.40} & \textbf{11.23} & \textbf{11.56} & \textbf{14.34} & \textbf{13.92} & \textbf{22.44} & \textbf{29.24} & \textbf{28.50} & \textbf{29.17} & \textbf{29.39} & \textbf{29.45} & \textbf{29.10} \\
		Vietnam & 86.79 & 72.02 & 41.19 & 52.88 & 57.33 & 56.95 & 56.75 & 56.68 & 56.91 & 56.40 & 56.29 & 55.85 & 56.30 & 55.20 & 55.22 & 55.01 & 55.70 & 56.09 & 54.42 & 54.06 \\
		Limburg & \textbf{36.40} & 70.47 & 54.38 & \textbf{38.15} & 48.17 & \textbf{33.52} & 43.40 & 47.43 & 45.10 & 44.67 & 43.18 & \textbf{37.47} & \textbf{39.19} & \textbf{37.40} & \textbf{38.03} & \textbf{38.45} & \textbf{36.56} & \textbf{36.77} & 47.23 & 46.64 \\
		Noord Brabant & 74.56 & 68.23 & 56.44 & 56.82 & \textbf{34.97} & \textbf{34.20} & \textbf{38.40} & 58.53 & 56.81 & 56.95 & 52.95 & 50.70 & 54.52 & 52.01 & 48.67 & 41.71 & \textbf{33.14} & \textbf{26.93} & \textbf{24.55} & \textbf{26.17} \\
		Lang.-Rouss. & \textbf{28.26} & \textbf{27.16} & \textbf{26.81} & \textbf{34.21} & \textbf{32.13} & 40.11 & \textbf{38.23} & \textbf{39.33} & \textbf{37.33} & \textbf{39.38} & 41.66 & 40.88 & 49.66 & \textbf{15.63} & \textbf{11.54} & 47.76 & 47.08 & 43.97 & 48.47 & 46.75 \\
		 
		\hline 
		%Logan Airport
		Bld 1 & 56.64 & 55.83 & 55.91 & 55.09 & 61.72 & 62.23 & 61.97 & 60.67 & 60.79 & 61.13 & 61.30 & 60.26 & 60.60 & 60.52 & 59.40 & 57.76 & 57.84 & 58.00 & 56.27 & 56.63 \\
		%castle
		Bld 2 & 40.52 & \textbf{25.70} & 52.94 & 45.57 & 41.93 & 50.56 & 42.96 & \textbf{13.16} & \textbf{14.40} & \textbf{15.17} & \textbf{14.60} & \textbf{9.58} & \textbf{10.15} & \textbf{30.50} & \textbf{30.35} & \textbf{29.91} & \textbf{29.19} & \textbf{32.04} & 42.14 & 40.77 \\
		%castle Boldt
		Bld 3 & 53.27 & 59.90 & 62.59 & 44.91 & 47.39 & 49.52 & \textbf{25.77} & \textbf{29.22} & 46.73 & \textbf{36.52} & 52.98 & 51.52 & 46.59 & 47.92 & \textbf{35.03} & \textbf{33.80} & \textbf{32.91} & 58.14 & 54.57 & 54.42 \\
		%castle Dracula
		Bld 4 & 42.73 & 42.12 & 40.76 & \textbf{39.84} & \textbf{39.03} & \textbf{39.79} & 41.19 & 43.20 & 46.66 & 45.92 & 47.43 & 41.77 & 43.52 & 49.00 & 43.38 & 50.35 & 44.22 & 41.65 & \textbf{38.83} & 40.79 \\
		%Chicago Lane Technical High School
		Bld 5 & 47.69 & \textbf{15.79} & \textbf{27.06} & \textbf{23.33} & 54.65 & 40.45 & \textbf{35.94} & 47.21 & 47.52 & 47.24 & 47.24 & \textbf{37.32} & \textbf{39.01} & \textbf{37.99} & \textbf{37.19} & \textbf{36.41} & \textbf{39.42} & \textbf{39.04} & \textbf{38.88} & 41.98 \\
		%Chicago Stadium
		Bld 6 & 43.62 & 49.67 & 51.77 & 49.42 & 57.43 & 58.61 & 60.23 & 55.39 & 53.36 & 51.30 & 52.99 & 48.24 & 48.23 & 47.62 & \textbf{32.19} & \textbf{34.21} & \textbf{32.49} & \textbf{30.55} & \textbf{32.60} & \textbf{31.44} \\
		%Gaffuri
		Bld 7 & 40.61 & 42.32 & 40.04 & \textbf{39.91} & 42.17 & \textbf{39.92} & \textbf{39.97} & \textbf{39.98} & \textbf{39.65} & \textbf{22.20} & \textbf{22.09} & \textbf{23.71} & \textbf{25.22} & \textbf{24.17} & 41.92 & 50.38 & 52.16 & 52.85 & 54.29 & 50.68 \\
		%Gaffuri_rotated
		Bld 8 & 49.70 & 53.18 & 50.62 & 49.34 & 47.90 & 48.02 & 50.43 & \textbf{26.07} & 51.58 & 55.20 & 48.11 & 50.74 & 47.40 & 47.94 & 45.71 & 47.93 & 47.41 & 44.28 & 44.37 & 46.19 \\
		%Gaffuri2
		Bld 9 & 40.33 & 41.40 & \textbf{23.12} & \textbf{21.20} & 44.50 & \textbf{34.17} & \textbf{32.69} & 41.70 & 41.37 & 40.94 & \textbf{17.30} & \textbf{17.86} & \textbf{18.55} & \textbf{19.72} & \textbf{22.00} & \textbf{18.98} & 51.00 & 49.46 & 50.40 & 51.59 \\
		%Gaffuri2_rotated
		Bld 10 & 47.53 & 48.99 & 46.38 & 47.87 & 47.88 & \textbf{34.66} & \textbf{34.78} & \textbf{32.56} & \textbf{39.19} & \textbf{16.61} & \textbf{16.24} & \textbf{23.46} & 50.74 & 51.39 & 49.98 & 50.36 & 48.57 & 47.07 & 45.69 & 47.96 \\
		%TU/e Paviljoen
		Bld 11 & \textbf{6.57} & 41.24 & \textbf{29.57} & \textbf{36.58} & \textbf{35.85} & \textbf{33.12} & \textbf{33.17} & \textbf{32.77} & \textbf{32.10} & \textbf{32.60} & \textbf{32.57} & \textbf{32.33} & \textbf{32.63} & \textbf{31.67} & \textbf{29.53} & \textbf{29.62} & \textbf{31.84} & \textbf{29.27} & \textbf{31.85} & \textbf{29.38} \\

		\hline 
		bird & 69.24 & 61.32 & 70.66 & 55.88 & 62.35 & 53.69 & 61.01 & 62.45 & \textbf{32.85} & \textbf{34.23} & \textbf{34.98} & \textbf{31.51} & \textbf{30.57} & 51.67 & 52.42 & 56.54 & 54.99 & 53.96 & 53.30 & 44.67 \\
		butterfly & 66.15 & 44.49 & 57.73 & \textbf{34.52} & \textbf{33.60} & \textbf{33.07} & \textbf{32.34} & \textbf{31.18} & \textbf{29.59} & \textbf{30.38} & \textbf{30.10} & \textbf{30.66} & \textbf{22.03} & \textbf{23.05} & \textbf{22.38} & \textbf{19.29} & \textbf{36.34} & \textbf{35.91} & \textbf{19.59} & \textbf{34.52} \\
		cat & 48.70 & 46.97 & 60.42 & 55.17 & 57.77 & 55.77 & 62.20 & 59.31 & 59.54 & 58.77 & 60.54 & 58.50 & 55.08 & 57.85 & 60.58 & 56.08 & 58.24 & 60.32 & 54.68 & 56.33 \\
		dog & 52.75 & 56.23 & 54.28 & 52.07 & 51.00 & 53.57 & 55.75 & 55.66 & 53.21 & 53.05 & 55.14 & 53.35 & 51.66 & 54.09 & 51.60 & 53.61 & 52.97 & \textbf{27.15} & \textbf{26.25} & \textbf{26.22} \\
		horse & 55.31 & \textbf{6.91} & 44.31 & 49.62 & 54.33 & 53.96 & 52.13 & 55.57 & 44.68 & 43.41 & 44.34 & 47.19 & \textbf{39.14} & 41.32 & 50.15 & 42.73 & 43.48 & 42.38 & 41.35 & \textbf{39.34} \\
		ostrich & 50.12 & 51.54 & \textbf{22.74} & 59.65 & 66.39 & 68.84 & 64.42 & 61.47 & 58.25 & 56.98 & 56.49 & 56.93 & 57.58 & 56.59 & 55.65 & 54.31 & 52.54 & 54.66 & 49.24 & 43.95 \\
		shark & \textbf{34.70} & 43.86 & \textbf{9.86} & \textbf{13.26} & 54.03 & 56.19 & 52.55 & 54.23 & 55.48 & 58.02 & 55.36 & 58.97 & 56.42 & 56.92 & 54.58 & 56.39 & 54.42 & 58.02 & 55.42 & 57.70 \\
		spider & 59.48 & \textbf{27.44} & 53.71 & 57.38 & 72.30 & 68.20 & 53.74 & 60.11 & 52.54 & \textbf{33.47} & \textbf{34.27} & \textbf{33.06} & \textbf{34.51} & \textbf{25.98} & \textbf{23.48} & \textbf{24.28} & \textbf{20.83} & \textbf{23.73} & \textbf{23.42} & \textbf{23.34} \\
		turtle & 55.95 & 42.22 & 55.92 & 62.14 & 54.61 & 40.83 & \textbf{28.57} & 56.12 & 62.73 & 63.13 & 65.21 & 65.25 & 64.18 & 62.78 & 62.38 & 62.23 & 58.75 & 58.99 & 57.13 & 56.91 \\
		 
		\hline

	\end{tabular}
\end{sidewaystable}

\begin{table}
	\caption{Influence of the input placement. For every algorithm and input the minimum, average and maximum normalized symmetric difference are given.}	
	\label{tab:input-placement}
	\centering
	\begin{tabular}{rccccccccc}
		\toprule		
		& \multicolumn{3}{c}{\textbf{Optimal}} & \multicolumn{3}{c}{\textbf{Hausdorff}} & \multicolumn{3}{c}{\textbf{Fr\'echet}} \\
		\cmidrule(lr){2-4}
		\cmidrule(lr){5-7}
		\cmidrule(lr){8-10}
		& min & avg & max & min & avg & max & min & avg & max \\
		\midrule
		Africa & 0.139 & 0.159 & 0.18 & 0.139 & 0.159 & 0.18 & 0.142 & 0.165 & 0.191\\
		Antarctica & 0.13 & 0.143 & 0.154 & 0.13 & 0.145 & 0.16 & 0.133 & 0.15 & 0.157\\
		Australia & 0.133 & 0.147 & 0.167 & 0.133 & 0.149 & 0.172 & 0.137 & 0.154 & 0.171\\
		Brazil & 0.156 & 0.173 & 0.199 & 0.156 & 0.173 & 0.199 & 0.162 & 0.182 & 0.23\\
		China & 0.157 & 0.174 & 0.209 & 0.157 & 0.174 & 0.209 & 0.163 & 0.184 & 0.214\\
		France & 0.135 & 0.149 & 0.162 & 0.135 & 0.149 & 0.162 & 0.14 & 0.157 & 0.173\\
		Greece & 0.323 & 0.341 & 0.36 & 0.325 & 0.343 & 0.372 & 0.343 & 0.398 & 0.512\\
		Italy & 0.292 & 0.334 & 0.364 & 0.292 & 0.335 & 0.369 & 0.292 & 0.352 & 0.39\\
		Lang.-Rouss. & 0.192 & 0.233 & 0.268 & 0.192 & 0.233 & 0.268 & 0.197 & 0.245 & 0.282\\
		Noord Brabant & 0.153 & 0.176 & 0.195 & 0.153 & 0.176 & 0.195 & 0.159 & 0.186 & 0.221\\
		Limburg & 0.245 & 0.265 & 0.296 & 0.257 & 0.276 & 0.305 & 0.259 & 0.344 & 0.524\\
		Great Britain & 0.242 & 0.262 & 0.282 & 0.242 & 0.265 & 0.312 & 0.255 & 0.289 & 0.322\\
		Vietnam & 0.364 & 0.392 & 0.416 & 0.412 & 0.456 & 0.494 & 0.665 & 0.847 & 0.995\\
		Switzerland & 0.168 & 0.181 & 0.194 & 0.168 & 0.181 & 0.194 & 0.171 & 0.193 & 0.21\\
		\midrule
		Bld 1 & 0.511 & 0.56 & 0.594 & 0.68 & 0.757 & 0.825 & 0.833 & 0.863 & 0.896\\
		Bld 2 & 0.235 & 0.284 & 0.342 & 0.235 & 0.284 & 0.342 & 0.243 & 0.302 & 0.36\\
		Bld 3 & 0.135 & 0.159 & 0.189 & 0.135 & 0.159 & 0.189 & 0.138 & 0.166 & 0.199\\
		Bld 4 & 0.109 & 0.124 & 0.136 & 0.109 & 0.124 & 0.136 & 0.113 & 0.126 & 0.136\\
		Bld 5 & 0.183 & 0.234 & 0.264 & 0.197 & 0.247 & 0.274 & 0.19 & 0.271 & 0.472\\
		Bld 6 & 0.144 & 0.157 & 0.175 & 0.144 & 0.158 & 0.175 & 0.149 & 0.167 & 0.182\\
		Bld 7 & 0.114 & 0.158 & 0.188 & 0.114 & 0.162 & 0.217 & 0.119 & 0.16 & 0.195\\
		Bld 8 & 0.162 & 0.181 & 0.203 & 0.162 & 0.181 & 0.203 & 0.162 & 0.19 & 0.218\\
		Bld 9 & 0.187 & 0.279 & 0.392 & 0.187 & 0.28 & 0.392 & 0.206 & 0.295 & 0.407\\
		Bld 10 & 0.293 & 0.308 & 0.322 & 0.293 & 0.311 & 0.337 & 0.313 & 0.333 & 0.382\\
		Bld 11 & 0.349 & 0.379 & 0.414 & 0.353 & 0.391 & 0.431 & 0.37 & 0.432 & 0.523\\
		\midrule
		bird & 0.28 & 0.304 & 0.323 & 0.28 & 0.307 & 0.342 & 0.284 & 0.324 & 0.349\\
		butterfly & 0.21 & 0.225 & 0.234 & 0.21 & 0.228 & 0.246 & 0.21 & 0.24 & 0.281\\
		cat & 0.206 & 0.254 & 0.292 & 0.23 & 0.276 & 0.333 & 0.223 & 0.275 & 0.317\\
		dog & 0.279 & 0.318 & 0.347 & 0.316 & 0.36 & 0.421 & 0.295 & 0.337 & 0.377\\
		horse & 0.286 & 0.343 & 0.401 & 0.314 & 0.384 & 0.504 & 0.291 & 0.363 & 0.42\\
		ostrich & 0.308 & 0.346 & 0.392 & 0.363 & 0.419 & 0.48 & 0.363 & 0.398 & 0.437\\
		shark & 0.31 & 0.342 & 0.374 & 0.321 & 0.355 & 0.383 & 0.333 & 0.381 & 0.407\\
		spider & 0.568 & 0.586 & 0.619 & 0.63 & 0.708 & 0.808 & 0.582 & 0.625 & 0.694\\
		turtle & 0.261 & 0.279 & 0.303 & 0.267 & 0.282 & 0.303 & 0.271 & 0.303 & 0.351\\

		\bottomrule
	\end{tabular}
\end{table}

\begin{table}
	\centering
	\caption{\label{tab:mst-methods-extended} Normalized symmetric difference, as an increase percentage w.r.t.\ optimal, of the algorithms. For the Hausdorff algorithm, results for the various heuristic improvements are shown. In the second row, \textit{None} means that no postprocessing heuristic was used; \textit{A}, \textit{R} and \textit{S} mean additions, removals and shifts, respectively. In the third row, \cmark\ and \xmark\ indicate whether $Q_4$ was chosen arbitrarily (\xmark) or using the symmetric difference heuristic (\xmark).}
	\makebox[\textwidth][c]{
	\begin{tabular}{rccccccccccc}
		\toprule
		& \textbf{Optimal} & \multicolumn{6}{c}{\textbf{Hausdorff}} & \textbf{Fr\'echet} \\
		\cmidrule(lr){3-8}
		\textit{postproc.} & & \multicolumn{2}{c}{\textit{None}} &  \multicolumn{2}{c}{\textit{A\,/\,R}} & \multicolumn{2}{c}{\textit{A\,/\,R\,/\,S}} \\
		\cmidrule(lr){3-4}
		\cmidrule(lr){5-6}
		\cmidrule(lr){7-8}
		\textit{$Q_4$ heur.} & & \xmark & \cmark & \xmark & \cmark &\xmark & \cmark \\
		\midrule
		Africa & 0.159 & $+\,358\,\%$ & $+\,291\,\%$ & $+\,5\,\%$ & $+\,0\,\%$ & $+\,0\,\%$ & $+\,0\,\%$ & $+\,4\,\%$\\
		Antarctica & 0.143 & $+\,332\,\%$ & $+\,267\,\%$ & $+\,2\,\%$ & $+\,2\,\%$ & $+\,2\,\%$ & $+\,2\,\%$ & $+\,5\,\%$\\
		Australia & 0.147 & $+\,350\,\%$ & $+\,288\,\%$ & $+\,3\,\%$ & $+\,1\,\%$ & $+\,3\,\%$ & $+\,1\,\%$ & $+\,4\,\%$\\
		Brazil & 0.173 & $+\,346\,\%$ & $+\,274\,\%$ & $+\,0\,\%$ & $+\,0\,\%$ & $+\,0\,\%$ & $+\,0\,\%$ & $+\,5\,\%$\\
		China & 0.174 & $+\,372\,\%$ & $+\,276\,\%$ & $+\,32\,\%$ & $+\,0\,\%$ & $+\,4\,\%$ & $+\,0\,\%$ & $+\,5\,\%$\\
		France & 0.149 & $+\,365\,\%$ & $+\,269\,\%$ & $+\,41\,\%$ & $+\,0\,\%$ & $+\,0\,\%$ & $+\,0\,\%$ & $+\,5\,\%$\\
		Greece & 0.341 & $+\,260\,\%$ & $+\,188\,\%$ & $+\,81\,\%$ & $+\,1\,\%$ & $+\,3\,\%$ & $+\,1\,\%$ & $+\,17\,\%$\\
		Italy & 0.334 & $+\,294\,\%$ & $+\,215\,\%$ & $+\,30\,\%$ & $+\,0\,\%$ & $+\,0\,\%$ & $+\,0\,\%$ & $+\,5\,\%$\\
		Lang.-Rouss. & 0.233 & $+\,339\,\%$ & $+\,248\,\%$ & $+\,40\,\%$ & $+\,0\,\%$ & $+\,4\,\%$ & $+\,0\,\%$ & $+\,5\,\%$\\
		Noord Brabant & 0.176 & $+\,339\,\%$ & $+\,254\,\%$ & $+\,23\,\%$ & $+\,0\,\%$ & $+\,7\,\%$ & $+\,0\,\%$ & $+\,6\,\%$\\
		Limburg & 0.265 & $+\,314\,\%$ & $+\,225\,\%$ & $+\,75\,\%$ & $+\,4\,\%$ & $+\,35\,\%$ & $+\,4\,\%$ & $+\,30\,\%$\\
		Great Britain & 0.262 & $+\,296\,\%$ & $+\,215\,\%$ & $+\,50\,\%$ & $+\,1\,\%$ & $+\,37\,\%$ & $+\,1\,\%$ & $+\,10\,\%$\\
		Vietnam & 0.392 & $+\,268\,\%$ & $+\,207\,\%$ & $+\,38\,\%$ & $+\,18\,\%$ & $+\,22\,\%$ & $+\,16\,\%$ & $+\,116\,\%$\\
		Switzerland & 0.181 & $+\,338\,\%$ & $+\,256\,\%$ & $+\,59\,\%$ & $+\,0\,\%$ & $+\,20\,\%$ & $+\,0\,\%$ & $+\,7\,\%$\\
		\textbf{Maps} & \textbf{0.223} & $\boldsymbol{+\,316\,\%}$& $\boldsymbol{+\,238\,\%}$ & $\boldsymbol{+\,39\,\%}$ & $\boldsymbol{+\,3\,\%}$ & $\boldsymbol{+\,11\,\%}$ & $\boldsymbol{+\,3\,\%}$ & $\boldsymbol{+\,23\,\%}$\\
		\midrule
		Bld 1 & 0.560 & $+\,245\,\%$ & $+\,201\,\%$ & $+\,78\,\%$ & $+\,42\,\%$ & $+\,57\,\%$ & $+\,35\,\%$ & $+\,54\,\%$\\
		Bld 2 & 0.284 & $+\,288\,\%$ & $+\,188\,\%$ & $+\,48\,\%$ & $+\,0\,\%$ & $+\,0\,\%$ & $+\,0\,\%$ & $+\,7\,\%$\\
		Bld 3 & 0.159 & $+\,339\,\%$ & $+\,248\,\%$ & $+\,3\,\%$ & $+\,0\,\%$ & $+\,3\,\%$ & $+\,0\,\%$ & $+\,4\,\%$\\
		Bld 4 & 0.124 & $+\,357\,\%$ & $+\,274\,\%$ & $+\,9\,\%$ & $+\,0\,\%$ & $+\,0\,\%$ & $+\,0\,\%$ & $+\,1\,\%$\\
		Bld 5 & 0.234 & $+\,246\,\%$ & $+\,167\,\%$ & $+\,55\,\%$ & $+\,3\,\%$ & $+\,20\,\%$ & $+\,5\,\%$ & $+\,16\,\%$\\
		Bld 6 & 0.157 & $+\,300\,\%$ & $+\,229\,\%$ & $+\,29\,\%$ & $+\,1\,\%$ & $+\,1\,\%$ & $+\,1\,\%$ & $+\,7\,\%$\\
		Bld 7 & 0.158 & $+\,284\,\%$ & $+\,203\,\%$ & $+\,19\,\%$ & $+\,3\,\%$ & $+\,19\,\%$ & $+\,3\,\%$ & $+\,1\,\%$\\
		Bld 8 & 0.181 & $+\,344\,\%$ & $+\,266\,\%$ & $+\,22\,\%$ & $+\,0\,\%$ & $+\,5\,\%$ & $+\,0\,\%$ & $+\,5\,\%$\\
		Bld 9 & 0.279 & $+\,245\,\%$ & $+\,161\,\%$ & $+\,55\,\%$ & $+\,0\,\%$ & $+\,26\,\%$ & $+\,0\,\%$ & $+\,6\,\%$\\
		Bld 10 & 0.308 & $+\,266\,\%$ & $+\,178\,\%$ & $+\,54\,\%$ & $+\,1\,\%$ & $+\,19\,\%$ & $+\,1\,\%$ & $+\,8\,\%$\\
		Bld 11 & 0.379 & $+\,216\,\%$ & $+\,164\,\%$ & $+\,47\,\%$ & $+\,4\,\%$ & $+\,10\,\%$ & $+\,3\,\%$ & $+\,14\,\%$\\
		\textbf{Buildings} & \textbf{0.257} & $\boldsymbol{+\,270\,\%}$& $\boldsymbol{+\,197\,\%}$ & $\boldsymbol{+\,47\,\%}$ & $\boldsymbol{+\,9\,\%}$ & $\boldsymbol{+\,21\,\%}$ & $\boldsymbol{+\,8\,\%}$  & $\boldsymbol{+\,17\,\%}$\\
		\midrule
		bird & 0.304 & $+\,273\,\%$ & $+\,206\,\%$ & $+\,51\,\%$ & $+\,1\,\%$ & $+\,2\,\%$ & $+\,1\,\%$ & $+\,7\,\%$\\
		butterfly & 0.225 & $+\,322\,\%$ & $+\,246\,\%$ & $+\,53\,\%$ & $+\,1\,\%$ & $+\,12\,\%$ & $+\,1\,\%$ & $+\,7\,\%$\\
		cat & 0.254 & $+\,269\,\%$ & $+\,208\,\%$ & $+\,24\,\%$ & $+\,8\,\%$ & $+\,9\,\%$ & $+\,9\,\%$ & $+\,8\,\%$\\
		dog & 0.318 & $+\,255\,\%$ & $+\,193\,\%$ & $+\,63\,\%$ & $+\,13\,\%$ & $+\,45\,\%$ & $+\,13\,\%$ & $+\,6\,\%$\\
		horse & 0.343 & $+\,229\,\%$ & $+\,168\,\%$ & $+\,94\,\%$ & $+\,14\,\%$ & $+\,63\,\%$ & $+\,12\,\%$ & $+\,6\,\%$\\
		ostrich & 0.346 & $+\,244\,\%$ & $+\,190\,\%$ & $+\,76\,\%$ & $+\,24\,\%$ & $+\,52\,\%$ & $+\,21\,\%$ & $+\,15\,\%$\\
		shark & 0.342 & $+\,268\,\%$ & $+\,214\,\%$ & $+\,26\,\%$ & $+\,5\,\%$ & $+\,4\,\%$ & $+\,4\,\%$ & $+\,12\,\%$\\
		spider & 0.586 & $+\,173\,\%$ & $+\,136\,\%$ & $+\,71\,\%$ & $+\,24\,\%$ & $+\,40\,\%$ & $+\,21\,\%$ & $+\,7\,\%$\\
		turtle & 0.279 & $+\,274\,\%$ & $+\,199\,\%$ & $+\,61\,\%$ & $+\,1\,\%$ & $+\,13\,\%$ & $+\,1\,\%$ & $+\,9\,\%$\\
		\textbf{Animals} & \textbf{0.333} & $\boldsymbol{+\,246\,\%}$& $\boldsymbol{+\,188\,\%}$ & $\boldsymbol{+\,60\,\%}$ & $\boldsymbol{+\,12\,\%}$ & $\boldsymbol{+\,29\,\%}$ & $\boldsymbol{+\,11\,\%}$   & $\boldsymbol{+\,8\,\%}$\\
		\bottomrule
	\end{tabular}
	}
\end{table}

\end{document}